\documentclass[12pt]{amsart}
\usepackage{url, calc}
\usepackage{amsmath,amsxtra,amssymb,latexsym,epsfig,amscd,amsthm,fancybox,epsfig}
\usepackage[mathscr]{eucal}
\usepackage{graphicx}
\usepackage{enumitem}
\usepackage{hyperref}
\hypersetup{
    colorlinks=true,
    linkcolor=blue,
    filecolor=magenta,
    urlcolor=cyan,
    citecolor=blue,
}

\usepackage[english]{babel}
\usepackage{bm}

\usepackage[normalem]{ulem}

\usepackage{tikz}
\usetikzlibrary{intersections}
\newcommand{\BU}{\mathbf{U}}

\newcommand{\bal}{\bm{\alpha}}
\newcommand{\BA}{\mathbf{A}}

\usepackage{multicol,xcolor}
\usepackage{epsfig} 
\usepackage{epstopdf}
\usepackage{cases}
\usepackage{subfig}
\usepackage{color}
\usepackage{hyperref}
\usepackage{bigints}
\usepackage{natbib}
\newtheorem{thm}{Theorem}[section]
\newtheorem {asp}{Assumption}[section]
\newtheorem{lm}{Lemma}[section]
\newtheorem{rmk}{Remark}[section]
\newtheorem{pron}{Proposition}[section]
\newtheorem{deff}{Definition}[section]

\newtheorem{lem}{Lemma}[section]
\newtheorem{prop}{Proposition}[section]
\theoremstyle{definition}

\theoremstyle{remark}

\numberwithin{equation}{section}

\DeclareMathOperator{\suppo}{supp}
\DeclareMathOperator{\Conv}{Conv}
\newcommand{\eps}{\varepsilon}

\newcommand{\M}{\mathcal{M}}
\newcommand{\F}{\mathcal{F}}

\newcommand{\E}{\mathbb{E}}
\newcommand{\BE}{\mathbf{E}}
\newcommand{\BB}{\mathbf{B}}

\newcommand{\BD}{\mathbf{D}}

\newcommand{\BX}{\mathbf{X}}

\newcommand{\BW}{\mathbf{W}}
\newcommand{\BN}{\mathbf{N}}

\newcommand{\bx}{\mathbf{x}}
\newcommand{\BY}{\mathbf{Y}}
\newcommand{\by}{\mathbf{y}}
\newcommand{\BZ}{\mathbf{Z}}
\newcommand{\ba}{\mathbf{a}}
\newcommand{\bb}{\mathbf{b}}
\newcommand{\bc}{\mathbf{c}}

\newcommand{\bp}{\mathbf{p}}

\newcommand{\bz}{\mathbf{z}}
\newcommand{\bu}{\mathbf{u}}

\newcommand{\N}{\mathbb{N}}

\newcommand{\PP}{\mathbb{P}}

\newcommand{\K}{\mathcal{K}}

\newcommand{\R}{\mathbb{R}}
\newcommand{\Z}{\mathbb{Z}}

\newcommand{\U}{\mathcal{U}}

\newcommand{\Se}{\mathcal{S}}

\newcommand{\wtd}{\widetilde}
\numberwithin{equation}{section}

\newcommand{\op}{{\mathcal L}}

\newcommand{\bed}{\begin{displaymath}}
\newcommand{\eed}{\end{displaymath}}
\newcommand{\bea}{\bed\begin{array}{rl}}
\newcommand{\eea}{\end{array}\eed}

\newcommand{\barray}{\begin{array}{ll}}
\newcommand{\earray}{\end{array}}
\newcommand{\diag}{{\rm diag}}

\newcommand{\1}{\boldsymbol{1}}
\newcommand{\0}{\boldsymbol{0}}
\newcommand{\bdelta}{\boldsymbol{\delta}}

\newcommand{\dist}{\mathrm{dist}}

\def\bar{\overline}
\def\hat{\widehat}
\def\a.s{\text{\;a.s.\;}}

\title[A general theory of coexistence and extinction]{A general
theory of coexistence and extinction for stochastic ecological
communities}

\author[A. Hening]{Alexandru Hening }
\address{Department of Mathematics\\
Tufts University\\
Bromfield-Pearson Hall\\
503 Boston Avenue\\
Medford, MA 02155\\
United States
}
\email{alexandru.hening@tufts.edu}

\author[D. Nguyen]{Dang H. Nguyen }
\address{Department of Mathematics \\
University of Alabama\\
345 Gordon Palmer Hall\\
Box 870350 \\
Tuscaloosa, AL 35487-0350 \\
United States}
\email{dangnh.maths@gmail.com}

\author[P. Chesson]{Peter Chesson}
\address{Department of Ecology and Evolutionary Biology\\
The University of Arizona\\
Tucson, AZ\\
United States
}
\email{pchesson@email.arizona.edu}

\usepackage[acronym,nomain]{glossaries}              

\newglossary[slg]{symbolslist}{syi}{syg}{Notation} 

\makeglossaries                                   

    \newglossaryentry{Z}{name=\ensuremath{\mathbb{Z_+}},
        description={The set of positive integers.},
        type=symbolslist}
\newglossaryentry{X(t)}{name=\ensuremath{\mathbf{X}(t)},
        description={The species densities $(X_1(t),\dots,X_n(t))$ at time $t$.},
        type=symbolslist}
    \newglossaryentry{R_+^n}{name=\ensuremath{\mathbb{R}_+^n},
        description={The positive orthant $[0,\infty)^n$.},
        type=symbolslist}
\newglossaryentry{X_i}{name=\ensuremath{X_i(t)},
        description={The density of species $i$ at time $t$.},
        type=symbolslist}

\newglossaryentry{Y}{name=\ensuremath{\mathbf{Y}(t)},
        description={The value of the auxiliary variable at time $t$.},
        type=symbolslist}
\newglossaryentry{BZ}{name=\ensuremath{\BZ(t)},
description={The process $(\BX(t),\BY(t))$ at time $t$.},
type=symbolslist}
\newglossaryentry{xi}{name=\ensuremath{\xi(t)},
        description={Random variable describing the environment at time $t$ (continuous time case) or on the interval $[t,t+1)$ (discrete time case).},
        type=symbolslist}
\newglossaryentry{S}{name=\ensuremath{\Se},
        description={The state space $\R_+^n\times \R^{\kappa_0}$ of the process $\BZ$.},
    type=symbolslist}
\newglossaryentry{Fi}{name=\ensuremath{F_i},
        description={The fitness functions of species $i$.},
    type=symbolslist}
    \newglossaryentry{G}{name=\ensuremath{G},
        description={Function describing the dynamics of the auxiliary variable.},
    type=symbolslist}
    \newglossaryentry{V}{name=\ensuremath{V},
        description={Lyapunov function coming up in Assumption \ref{a:main} and in Assumption \ref{a:sde}.},
    type=symbolslist}
    \newglossaryentry{cTx}{name=\ensuremath{\bc^\top\bx},
        description={The scalar product $\bc^\top\bx = \sum_{i} c_ix_i$.},
    type=symbolslist}

    \newglossaryentry{constants}{name=\ensuremath{\gamma_1,\gamma_3, C, \rho},
        description={Strictly positive constants coming up in Assumption \ref{a:main}.},
    type=symbolslist}
     \newglossaryentry{h}{name=\ensuremath{h},
        description={The function $$h(\bz,\xi)=\left(\max_{i=1}^n \left\{\max
\left\{F_i(\bz,\xi),
\frac{1}{F_i(\bz,\xi)}\right\}\right\}\right)^{\gamma_3}$$ coming up in Assumption \ref{a:main}.},
    type=symbolslist}
     \newglossaryentry{aob}{name=\ensuremath{a \circ b},
        description={ $a \circ b := (a_1b_1, a_2b_2,\dots, a_nb_n)$.},
    type=symbolslist}

    \newglossaryentry{Se_0}{name=\ensuremath{\Se_0},
        description={ The extinction set $\Se_0:=\{(\bx,\by)\in\Se~:~\min_i x_i=0\}$, where at leaste one of the $n$ species is extinct.},
    type=symbolslist}

    \newglossaryentry{Se_+}{name=\ensuremath{\Se_+},
        description={ The persistence set $\Se_+:=\Se\setminus\Se_0$.},
    type=symbolslist}

    \newglossaryentry{Se_eta}{name=\ensuremath{\Se_\eta},
        description={ $\Se_\eta:=\{(\bx,\by)\in\Se~:~\min_i x_i\leq \eta\}$.},
    type=symbolslist}
    \newglossaryentry{P}{name=\ensuremath{\PP},
        description={ The probability measure.},
    type=symbolslist}
     \newglossaryentry{E}{name=\ensuremath{\E},
        description={ The expectation under the probability measure $\PP$.},
    type=symbolslist}
    \newglossaryentry{Ez}{name=\ensuremath{\E_\bz},
        description={ The expectation under the probability measure $\PP_\bz$.},
    type=symbolslist}
     \newglossaryentry{Pz}{name=\ensuremath{\PP_\bz},
        description={ $\PP_\bz(\cdot)=\PP(\cdot~|~\BZ(0)=\bz)$.},
    type=symbolslist}
\newglossaryentry{delta}{name=\ensuremath{\delta_\bz},
        description={ The Dirac mass at $\bz$.},
    type=symbolslist}

    \newglossaryentry{tPi_t(B)}{name=\ensuremath{\widetilde\Pi_t(B)},
        description={ The random occupation measure of a Borel set $B$: in discrete time $\widetilde\Pi_t(B):=\frac{1}{t}\sum_{s=1}^t\delta_{\BZ(s)}(B)$ and in continuous time $\tilde \Pi_t(B) = \frac{1}{t}\int_0^t \1_{\{\BZ(s)\in B\}}\,ds$.},
    type=symbolslist}

  \newglossaryentry{Pi_t(B)}{name=\ensuremath{\Pi_{t,\bz}(B)},
        description={ The occupation measure $\Pi_{t,\bz}(B):=\E_{\bz}\widetilde\Pi_t(B)$.},
    type=symbolslist}

\newglossaryentry{Bb}{name=\ensuremath{\mathcal{B}_b},
        description={ The set $\mathcal{B}_b:=\{H:\Se\to\R~|~H~\text{Borel, bounded}\}$ of bounded Borel functions.},
    type=symbolslist}

\newglossaryentry{Cb}{name=\ensuremath{\mathcal{C}_b},
        description={ The set  $\mathcal{C}_b:=\{H:\Se\to\R~|~H~\text{continuous, bounded}\}$ of bounded continuous functions.},
    type=symbolslist}
\newglossaryentry{Pp}{name=\ensuremath{P},
        description={ The transition operator of the discrete-time process $\BZ(t)$.},
    type=symbolslist}
\newglossaryentry{P_t}{name=\ensuremath{P_t},
        description={ The $t$-transition operator or semigroup operator that acts on functions $H\in \mathcal{B}_b$ as  $P_t(H)(\bz) = \E_\bz(H(\BZ(t))$.},
    type=symbolslist}

\newglossaryentry{Se^i}{name=\ensuremath{\Se^i},
        description={The set $\Se^i:=\{\bz=(\bx,\by)\in\Se~:~x_i>0\}$ is the subset of the state space where species $i$ persists.},
    type=symbolslist}

\newglossaryentry{Se(mu)}{name=\ensuremath{\Se(\mu)},
        description={ The species supported by the ergodic measure $\mu$, i.e. $\Se(\mu)=\{1\leq i\leq
n~:~\mu(\Se^i)=1\}$.},
    type=symbolslist}

\newglossaryentry{r_i(mu)}{name=\ensuremath{r_i(\mu)},
        description={  The expected per-capita growth rate of species $i$ when
introduced in the community described by $\mu$. In the discrete-time case $r_i(\mu) =\int_{\Se}\E[\log F_i(\bz,\xi(1))]\,\mu(d\bz),$ while in the continuous case $r_i(\mu) =
\int_{\Se}\left(f_i(\bz)-\frac{\sigma_{ii}g_i^2(\bz)}{2}\right)\mu(d\bz)$.},
    type=symbolslist}
\newglossaryentry{fi}{name=\ensuremath{f_i},
        description={ The drift of the dynamics of species $i$ is $X_if_i(\BX)\,dt$ in the continuous time setting.},
    type=symbolslist}
\newglossaryentry{gi}{name=\ensuremath{g_i},
        description={ The diffusion term of the dynamics of species $i$ is $X_ig_i(\BX)\,dE^i$ in the continuous time setting.},
    type=symbolslist}
\newglossaryentry{BM}{name=\ensuremath{\BE, \BW, \BB},
        description={ Brownian motions on $\R^n, \R^{\kappa_0}, \R^{n+\kappa_0}$.},
    type=symbolslist}
\newglossaryentry{M}{name=\ensuremath{\M},
        description={ The set of ergodic invariant probability measures supported on $\Se_0$.},
    type=symbolslist}
\newglossaryentry{M1}{name=\ensuremath{\M^1},
        description={ The set of ergodic invariant probability measures supported on $\partial \R_+^{n}$ and which are attractors: $\M^1:=\left\{\mu\in\M : \mu ~~\text{satisfies Assumption} ~~\ref{a.extnn}\right\}$.  This is in the continuous time setting without an auxiliary variable.},
    type=symbolslist}
\newglossaryentry{M2}{name=\ensuremath{\M^2},
        description={ The set of ergodic invariant probability measures supported on $\partial \R_+^{n}$ and which are not attractors: $\M^2:=\M\setminus\M^1$. This is in the continuous time setting without an auxiliary variable.},
    type=symbolslist}

    \newglossaryentry{R_+^mu}{name=\ensuremath{\R_+^\mu},
        description={ The set $\R_+^\mu:=\{(x_1,\dots,x_n)\in\R^n_+: x_i=0\text{ if } i\in \Se(\mu)^c\}$. This is in the continuous time setting without an auxiliary variable.},
    type=symbolslist}

        \newglossaryentry{R_+^mu_circ}{name=\ensuremath{\R_+^{\mu,\circ}},
        description={ The set $\R_+^{\mu,\circ}:=\{(x_1,\dots,x_n)\in\R^n_+: x_i=0\text{ if } i\in \Se(\mu)^c\text{ and }x_i>0\text{ if  }x_i\in \Se(\mu)\}$. This is in the continuous time setting without an auxiliary variable.},
    type=symbolslist}
    \newglossaryentry{partial R_+^mu_circ}{name=\ensuremath{\partial\R_+^{\mu}},
        description={ The set $\partial\R_+^{\mu}:=\R_+^\mu\setminus\R_+^{\mu,\circ}$. This is in the continuous time setting without an auxiliary variable.},
    type=symbolslist}

 \newglossaryentry{Conv M}{name=\ensuremath{\Conv(\M)},
        description={ The set of invariant probability measures supported on $\Se_0$.},
    type=symbolslist}
     \newglossaryentry{Gamma}{name=\ensuremath{\Gamma_\bz},
        description={ The set $\Gamma_\bz:=\{\tilde
\bz\in\Se~|~\tilde \bz~\text{is accessible from}~\bz\}$ of points $\tilde \bz$ such that for every neighborhood $U$ of $\tilde \bz$ there is $t\geq 0$ for which $P_t(\bz, U)>0$.},
    type=symbolslist}

\newglossaryentry{phi}{name=\ensuremath{\phi},
        description={A function that comes up in Assumption \ref{a:ext} which is needed for our extinction results.},
    type=symbolslist}

\newglossaryentry{Se^I}{name=\ensuremath{\Se^I},
        description={For a nonempty subset $I\subset \{1,\dots,n\}$ we define
\newline $\Se^I:=\{(\bx,\by)\in\Se~|~x_i=0, i\notin I\}$. $\Se^I$ is the subspace in which all species not in $I$ are absent, and some or all species from $I$ are present. The set $\Se^I$ represents a subcommunity where we can define persistence and extinction sets relative to that subcommunity.},
  type=symbolslist}

\newglossaryentry{Se^empty}{name=\ensuremath{\Se^\emptyset},
description={Let $\Se^\emptyset = \{(0,\by)\in\Se\}$ be the set where all species are extinct.},
type=symbolslist}

\newglossaryentry{Se_0^I}{name=\ensuremath{\Se_0^I},
description={If we restrict the
process to $\Se^I$ then the extinction set, where at least one species from $I$ is extinct, is given by
$\Se_0^I:=\{\bz\in \Se^I~|~\prod_{j\in I}x_j=0\}$.},
type=symbolslist}

\newglossaryentry{Se_+^I}{name=\ensuremath{\Se_+^I},
description={If we restrict the
process to $\Se^I$ then the
persistence set, where all species from $I$ persist, is given by $\Se_+^I:=\Se^I\setminus \Se_0^I$.},
type=symbolslist}

\newglossaryentry{M^I}{name=\ensuremath{\M^I},
description={$\M^I:=\{\mu\in\M~|~\mu(\Se^I)=1\}$ is the set of ergodic measures supported on the subspace $\Se^I$.},
type=symbolslist}

\newglossaryentry{M^{I,+}}{name=\ensuremath{\M^{I,+}},
description={$\M^{I,+}:=\{\mu\in\M~|~\mu(\Se_+^I)=1\}$ be the set of ergodic measures supported on the subspace $\Se^{I}_+$.},
type=symbolslist}

\newglossaryentry{M^{I,partial}}{name=\ensuremath{\M^{I,\partial}},
description={$\M^{I,\partial}:=\{\mu\in\M~|~\mu(\Se_0^I)=1\}$ is the set of ergodic probability measures supported on $\Se^I_0$. },
type=symbolslist}

\newglossaryentry{Uomega}{name=\ensuremath{\U=\U(\omega)},
description={The
(random) set of weak$^*$-limit points of $(\widetilde\Pi_t)_{t\in \N}$.},
type=symbolslist}

\newglossaryentry{E1}{name=\ensuremath{E_1},
description={The set of subsets of $\{1,\dots,n\}$ such that if $I\in E_1$ then the invariant measures living on $\Se_+^I$ are attractors. See Definition \ref{a.extn}.},
type=symbolslist}

\newglossaryentry{E2}{name=\ensuremath{E_2},
description={The set $E_2:=\mathcal{P}(n)\setminus E_1$.},
type=symbolslist}

\newglossaryentry{Pn}{name=\ensuremath{\mathcal{P}(n)},
description={The set of all subsets of $\{1,\dots,n\}$.},
type=symbolslist}

\newglossaryentry{L}{name=\ensuremath{\mathcal{L}},
description={The infinitesimal generator of the process $\BZ(t)$ in the continuous time setting.},
type=symbolslist}

 \newglossaryentry{constants2}{name=\ensuremath{\gamma_4, \gamma_5, C_4},
        description={Strictly positive constants coming up in Assumption \ref{a:sde}.},
    type=symbolslist}

  \newglossaryentry{Delta}{name=\ensuremath{\Delta},
        description={The set  $\Delta:=\{\bx\in
\R^n_+~|~\sum_i x_i=1\}$.},
    type=symbolslist}

    \newglossaryentry{Gammag}{name=\ensuremath{\Gamma, \Sigma, (\sigma_{ij})},
        description={$\Gamma$ is a $(n+{\kappa_0})\times (n+{\kappa_0})$ matrix such that
$\Gamma^\top\Gamma=\Sigma=(\sigma_{ij})_{(n+{\kappa_0})\times (n+{\kappa_0})}$. The matrix $\Sigma$ encodes the covariance structure of the Brownian motions from the continuous time setting \eqref{e:system}.},
    type=symbolslist}

\newglossarystyle{notationlong}{%
\setglossarystyle{long}
  {\end{longtable}}%

}

\usepackage{blindtext}

\begin{document}
\begin{abstract}

We analyze a general theory for coexistence and extinction of
ecological communities that are influenced by stochastic
temporal environmental fluctuations. The results apply to
discrete time (stochastic difference equations), continuous time
(stochastic differential equations), compact and non-compact
state spaces and degenerate or non-degenerate noise. In
addition, we can also include in the dynamics auxiliary
variables that model environmental fluctuations, population
structure, eco-environmental feedbacks or other internal or
external factors.

We are able to significantly generalize the recent discrete time
results by Benaim and Schreiber (Journal of Mathematical Biology
'19) to non-compact state spaces, and we provide stronger persistence and
extinction results. The continuous time results by Hening and
Nguyen (Annals of Applied Probability '18) are strengthened to
include degenerate noise and auxiliary
variables.

Using the general theory, we work out several examples. In
discrete time, we classify the dynamics when there are one or two
species, and look at the Ricker model, Log-normally distributed
offspring models, lottery models, discrete Lotka-Volterra models
as well as models of perennial and annual organisms. For the
continuous time setting we explore models with a resource
variable, stochastic replicator models, and three dimensional
Lotka-Volterra models.
\end{abstract}

\maketitle

\tableofcontents

\section{Introduction}

Stochastic variation over time in the physical environment
is an ever-present feature of natural ecosystems that has major
effects on the lives of organisms (\cite{RN1}). This fact means
that the dynamics of natural populations of plants and animals
are not modeled well deterministically. Moreover, the usual
deterministic attractors cannot reasonably characterize the
asymptotic dynamics of natural communities.
Although deterministic attractors, such as point equilibria, are
often regarded as central tendencies about which populations fluctuate
(\cite{MayRM1974}), a stochastic environment
can introduce qualitative changes in the nature of long-term
dynamics that depart dramatically from the predictions of a deterministic attractor.
For instance, in some cases, the deterministic attractor may predict species
extinctions while the system in a stochastic environment predicts that all species persist
(\cite{CW81, LC16}).

Prominent ecologists have long argued for stochastic models of
the physical environment in ecological theory (\cite{RN2, RN5,
RN3, RN4}), but models of ecological dynamics in stochastic
environments have posed serious challenges. Nevertheless, over
the years, understanding has steadily improved. Of much
significance is the recent development of general theory and
techniques for establishing population persistence. Following
early beginnings (\cite{C82, CE89, E89}), recent work
(\cite{SBA11, S12, BS19, HN16}) has provided key results for
analyzing both discrete- and continuous-time variable
environment models with multiple species. At first, these
developments were restricted to the case of white noise
environments (\cite{SBA11}), which means that the vector of
population sizes is a Markovian state variable. However,
\cite{BS19} goes beyond this case to allow auxiliary state
variables permitting Markovian environmental variation,
structured populations and feedbacks from population densities
through other species. As remarkable as these developments are,
they nevertheless have the serious limitation that the state
space is required to be compact. Many ecological models naturally have non compact state spaces.
Although it has been argued that finiteness of the universe, and finite capabilities
of organisms, justify restriction to compact state spaces, modifying existing ecological models to compact
state spaces requires unnatural assumptions, and restricts consideration to only a subset of ecological
models that are useful in practice.

A key innovation of this work is removal of this compactness
restriction. With this generalization, a number of ecological models that
have led to new theoretical understanding can be fully rigorously analyzed. Moreover, boundedness restrictions on environmental variables
in discrete-time models are removed, allowing a broader array of valuable ecological models to be analyzed.
As a further generalization, in case not all species persist, we identify conditions for convergence to an invariant measure supported on the persistence set or the extinction set. Our results make use of the significant recent advances by
\cite{BS19, B18, HN16}. In particular, we explicitly use some of
the examples and computations by \cite{SBA11, BS19} in order to
highlight how our analysis can strengthen earlier results as well
as settle some previous conjectures. This work also highlights how the results apply to key ecological models that have been instrumental in showing strong effects of stochastic environmental variation on species coexistence. In this way, it strives to connect key ecological findings with the mathematics of stochastic persistence.

The paper is structured as follows. In Section \ref{s:results0} we present the framework and results in discrete time. The continuous time results appear in Section \ref{s:results}. Section \ref{s:app_disc} showcases how our results can be used to better understand specific discrete time ecological examples. Similarly, Section \ref{s:cont} applies the developed theory to continuous time examples. We discuss the relevance of our work in Section \ref{s:discussion}.

\printglossary[type=symbolslist,style=notationlong]   

\glsaddall

\section{Discrete Time Results} \label{s:results0}
We study a system of $n$ species that interact nonlinearly
and are influenced by temporal environmental variation.
The species densities at time $t\in \Z_+:=\{k\in\Z~|~k\geq 0\}$
are denoted by the vector
$\BX(t)=(X_1(t),\dots,X_n(t))$ from $\R_+^n:=[0,\infty)^n$. The
species dynamics are influenced by and influence the auxiliary
variable $\BY(t)$, which takes values in $\R^{\kappa_0}$. This
variable can correspond to eco-environmental feedbacks, forcing, the structure of each species or other factors. The dynamics of the species
and of the feedback variable are changed by stochastic temporal
variations. These environmental variations are
described by the sequence of independent identically distributed
random variables $\xi(0), \xi(1),\dots$ which take values in a
Polish space $\Xi$, that is, a separable completely metrizable topological space of which multidimensional Euclidean space is a key example. The state of the environment on the interval
$[t,t+1)$ is described by  $\xi(t)$, which then affects what the species densities will be at
 time $t+1$. This is a deviation from a convention in some papers on this topic where the environment on $(t,t+1]$ is given as $\xi(t+1)$, i.e. just a difference in the time index, which has no effect on the results. However, in our view this leads to more consistent terminology, and matches the ecological literature given that the  species $i$ has a fitness
function $F_i$ that at time $t$ depends on the densities
$\BX(t)$ of the various species, the auxiliary variable
$\BY(t)$, and the environmental variable $\xi(t)$. Similarly,
the auxiliary variable at time $t+1$ will be a function of
$\BX(t), \BY(t)$ and $\xi(t)$. The process
$\BZ(t)=(\BX(t), \BY(t))$ tracks jointly the species
densities and the auxiliary variable and lives on the state space
$\Se:=\R_+^n\times\R^{\kappa_0}$.
As a result the dynamics will be
\begin{equation}\label{e:system_discrete}
\begin{aligned}
X_i(t+1)&=X_i(t) F_i(\BX(t), \BY(t),\xi(t)), ~i=1,\dots,n\\
\BY(t+1) &=G(\BX(t),\BY(t), \xi(t)).
\end{aligned}
\end{equation}
We define $F(\cdot)=(F_1(\cdot),\dots,F_n(\cdot))$.
\begin{asp}\label{a:main}
The following assumptions are made throughout the paper:
\begin{itemize}
\item [\textbf{A1)}] For each $i=1,\dots,n$, the fitness
function $F_i(\bz,\xi)$ is continuous in $\bz=(\bx,\by)$ for every $\xi$, where
$\bx=(x_1,\dots x_n)$ and $\by=(y_1,\dots y_{\kappa_0})$,
measurable in $(\bz,\xi)$, and strictly positive.
\item [\textbf{A2)}] The function
$G:\R_+^n\times\R^{\kappa_0}\times \Xi\to\R^{\kappa_0}$ is continuous in
$\bz = (\bx, \by)$ and measurable in $(\bz, \xi)$.
\item [\textbf{A3)}] There exists a function $V:\Se\to\R_+$ and
constants $\gamma_1,\gamma_3, C>0$ and $\rho\in(0,1)$ such that
for all $\bz\in \Se$
\begin{itemize}
\item[i)] $|V(\bx,\by)|\geq |\bx|^{\gamma_1}+1$,
\item[ii)] $\lim_{|\bz|\to\infty} V(\bz)=\infty,$ and
\item[iii)] $\E \left[V(\bx\circ F(\bz,\xi(1)),
G(\bz,\xi(1)))h(\bz,\xi(1))\right]\leq \rho V(\bz)+C$
where \newline $a \circ b = (a_1b_1, a_2b_2,\dots, a_nb_n)$ is the Hadamard product,
\begin{equation}\label{e:h}
h(\bz,\xi)=\left(\max_{i=1}^n \left\{\max
\left\{F_i(\bz,\xi),
\frac{1}{F_i(\bz,\xi)}\right\}\right\}\right)^{\gamma_3}.
\end{equation}
\end{itemize}
\end{itemize}
\end{asp}
Assumptions A1) and A2) make sure that $\BZ(t)$ is a
well-behaved Markov process. Measurability of $F_i(\bz,\xi)$ in $(\bz,\xi)$ is with respect to the natural Borel $\sigma$-algebra of the Polish space $\Xi$. Assumption A3) is also required in order
to ensure that $\BZ(t)$ returns to compact subsets of $\Se$
exponentially fast and that the growth rates do not change too
abruptly near infinity. A standard condition to ensure that $\BZ(t)$ returns to compact subsets of $\Se$
exponentially fast, which is critical to obtain exponential fast convergence to an invariant measure of a Markov chain, is A3)-ii) and a `weaker' version of A3)-iii), namely: $$\E \left[V(\bx\circ F(\bz,\xi(1)),
G(\bz,\xi(1)))\right]\leq \rho V(\bz)+C.$$

Assumption A3) is a Lyapunov type condition one needs in order to make sure the solution $\BZ(t)$ is well behaved and bounded in some sense.
A standard condition which ensures that $\BZ(t)$ returns to compact subsets of $\Se$
exponentially fast, which is critical to obtain exponential fast convergence to an invariant measure of a Markov chain, is A3)-ii) together with a `weaker' version of A3)-iii): $$\E \left[V(\bx\circ F(\bz,\xi(1)),
G(\bz,\xi(1)))\right]\leq \rho V(\bz)+C.$$
This means that $\E_\bz(e^{r\tau})<\infty$ for some $r>0$ where $\tau$ is the first time the process enters a compact set. This finiteness is important in the study of Markov chains. Note that after visiting a compact set the process can leave this compact set for some time, before it returns again.

We next explain why condition A3)-iii) is needed.
In order to better estimate the sample paths of the process, we want the growth rates to not change too abruptly.
Since $X_i(t+1)=X_i(t)F_i(\BX(t),\BY(t),\xi(t))$, to make sure that the growth rates do not change too much in one time step we want to be able to control $X_i(t) F_i(\BX(t),\BY(t),\xi(t))$. We can do this if we know that
$$
\left[F_i(\BX(t),\BY(t),\xi(t))\right]^{1-\gamma_3}\leq  F_i(\BX(t),\BY(t),\xi(t))\leq \left[F_i(\BX(t),\BY(t),\xi(t))\right]^{1+\gamma_3}
$$
for some $\gamma_3>0$.
This is why $h(\bz,\xi)$ has the form \eqref{e:h} in A3)-iii).
Finally, A3)-i) is needed because we need some boundedness of the moments of $\BX$.

In general, if $F_i$ is smaller than 1 when $\bx$ is large but bounded below in a certain sense, i.e., $\E \left(\frac{1}{F_i(\bx,\xi)}\right)$ is bounded, then A3) will be satisfied.

\begin{rmk}
Note that in general it is not always easy to construct the Lyapunov function $V$. There is no definite answer for whether such a $V$ exists. However, in many applications, one can construct such a function. The examples presented throughout the paper show how one can do this. Furthermore, under stronger assumptions, for example, when the process enters and then stays forever in a compact set, the Lyapunov function is not required - see Remark \ref{r:compact}.
\end{rmk}

\begin{rmk}
When we say that $F_i(\bz,\xi)$ is measurable in $(\bz,\xi)$ we note that the Polish space $\Xi$ comes with has a natural Borel $\sigma$-algebra.
\end{rmk}
\begin{rmk}\label{r:compact}
If it is known that the dynamics remain in a compact subset $K\subset
\R_+^n\times \R^{\kappa_0}$, instead of A3), the following assumption is sufficient.
 \begin{itemize}
\item [\textbf{A3C)}] There exists a constant $\gamma_3>0$ such
that for all $\bz\in \Se$ we have
\[
\E \left[h(\bz,\xi(1))\right]<\infty
\]
where $$h(\bz,\xi)=\left(\max_{i=1}^n \left\{\max
\left\{F_i(\bz,\xi),
\frac{1}{F_i(\bz,\xi)}\right\}\right\}\right)^{\gamma_3}.$$
\end{itemize}
\end{rmk}
\begin{rmk}

Assumption A3) (respectively A3C)) is strictly weaker than the
assumptions made by \cite{BS19}. In particular \cite{BS19}
assume
\begin{itemize}
\item [\textbf{A3')}] There is a compact subset $K\subset
\R_+^n\times \R^{\kappa_0}$ such that all solutions $\BZ(t)$
satisfy $\BZ(t)\in K$ for $t\in \Z_+$ sufficiently large.
\item [\textbf{A4')}] For all $i=1,2,\dots,n $,
$\sup_{\bz,\xi}|\ln F_i(\bz,\xi)|<\infty.$
\end{itemize}
Note that assumption A3) need not imply that $\BZ(t)$ remain in
a compact set, as examples from Section \ref{s:app_disc} attest.
Note also that A3') and A4') imply A3C).

\end{rmk}

There are various ways in which one can define the persistence
of the species from the studied ecosystem. We define the
extinction set, where at least one species is extinct, by
\[
\Se_0:=\{(\bx,\by)\in\Se~:~\min_i x_i=0\}
\]
and the persistence set by
\[
\Se_+:=\Se\setminus\Se_0.
\]
For any $\eta>0$ let
\[
\Se_\eta:=\{(\bx,\by)\in\Se~:~\min_i x_i\leq \eta\}
\]
be the subset of $\Se$ where at least one species is within
$\eta$ of extinction. We say \eqref{e:system_discrete} is
\textit{stochastically persistent in probability} (\cite{C82})
if for all $\eps>0$ there exists $\eta(\eps)=\eta>0$ such that
for all $\bz\in \Se\setminus\Se_0$
\[
\liminf_{t\to\infty} \PP_{\bz}\{\BZ(t)\notin S_\eta\}>1-\eps.
\]

In essence, this means that if no species are extinct initially, with high probability their densities will not be found close to any of the axes at any given time in the future. We can also examine the realized frequencies with which the species have specific ranges of values for the period of time 1 to $t$ with the ``random occupation measure."
For any $t\in \N$ define the \textit{random occupation
measure}
\[
\widetilde\Pi_t(B):=\frac{1}{t}\sum_{s=1}^t\delta_{\BZ(s)}(B)
\]
where $\delta_{\BZ(s)}$ is the Dirac measure at $\BZ(s)$ and $B$
is any Borel subset of $\Se$. We note that $\Pi_t$ is a random
probability measure and $\Pi_t(B)$ tells us the proportion of
time the system spends in $B$ up to time $t$.
We also define the occupation measure
\[
\Pi_{t,z}(B):=\E_{\bz}\widetilde\Pi_t(B)
\]

We say
\eqref{e:system_discrete} is \textit{almost surely
stochastically persistent} (\cite{S12, BS19}) if for all
$\eps>0$ there exists $\eta(\eps)=\eta>0$ such that for all
$\bz\in \Se\setminus\Se_0$
\[
\PP_\bz\left(\liminf_{t\to\infty} \widetilde \Pi_t(\Se\setminus \Se_\eta)>1-\eps\right)=1.
\]

\subsection{Invariant measures and criteria for
persistence}Following the seminal deterministic work by
\cite{H81} and the stochastic work by \cite{CE89,
SBA11, B18, HN16} one needs to look at the Lyapunov exponents
(expected per-capita growth rate) of certain invariant
probability measures in order to give conditions for
persistence. Intuitively, these invasion rates tell us if a
species tends to increase or decrease when it is introduced into
a subcommunity of species that is fluctuating according to a stationary probability distribution that it achieves asymptotically. We start by
defining some of the required mathematical concepts.

The transition operator $P: \mathcal{B}_b\to \mathcal{B}_b$ of the
process $\BZ$ is an operator which acts on bounded Borel functions
$\mathcal{B}_b:=\{H:\Se\to\R~|~H~\text{Borel and bounded}\}$ as
\[
PH(\bz)=\E_\bz[H(\BZ(1))]:=\E[H(\BZ(1))~|~\BZ(0)=\bz],
~\bz\in\Se.
\]
Thus, this operator gives conditional expectations of functions for changes of one time step. For multiple time steps, we can define for any $t\in\Z_+$ the $t$-time
transition operator $P_t: \mathcal{B}_b\to \mathcal{B}_b$ via
\[
P_tH(\bz)=\E_\bz[H(\BZ(t))], ~\bz\in\Se.
\]
For any Borel set $B\subset \Se$ the function $\1_{B}$ is the
indicator function, which is $1$ on $B$ and $0$ on the
complement, $B^c$, of B. Sometimes we will write $$P_t(\bz,
B):=P_t \1_{B}(\bz)$$
if $B$ is a Borel set and $\bz\in \Se$ any initial condition, thus defining $t$-step transition probabilities.
We note that our assumptions imply that the process $\BZ(t)$ is
Feller (\cite{B18, BS19}): if $C_b(\Se)$ is the set of continuous
real valued functions defined on $\Se$, then the mapping
\[
(t,\bz)\to P_tf(\bz)
\]
is continuous. On the other hand, the operator $P_t$ can also be used to obtain the probability distribution for $\BZ(t)$ starting with initial probability distribution (probability measure) $\mu$ for $\BZ(0)$. By duality, the operator $P_t$ acts on Borel probability measures $\mu$ by $\mu\to \mu P_t$ where $\mu
P_t$ is the probability measure given by
\[
\int_{\Se}H(\bz) (\mu P_t)(d\bz):=\int_{\Se}P_tH(\bz) \mu(d\bz)
\]
for all $H\in C_b(\Se)$, and is thus the probability distribution of $\BZ(t)$.

A Borel probability measure $\mu$ on $\Se$ is called an
\textit{invariant probability measure} if
\[
\mu P_t = \mu, ~t\geq 0.
\]
The above equation tells us that if the system starts with a certain distribution $\mu$ then it will have this distribution at any future time $t\geq 0$. Intuitively, this is the random analogue of a fixed point from dynamical systems.

The building blocks of the invariant probability measures are
the \textit{ergodic} invariant probability measures. These can
be characterized as the extreme points of the set of all
invariant probability measures. Equivalently, an invariant
probability measure is ergodic if it cannot be written as a
nontrivial convex combination of invariant probability measures. The set of ergodic invariant probability measures analogous to the set of stable points of a dynamical system. The invariant probabilities measures are then analogous to randomly starting the dynamical system at different stable points with given probabilities.

The invariant and ergodic probability measures are related to the long term behavior of the system. Suppose $\BZ(0)$ has distribution $\mu$ and $\mu$ is an invariant probability measure. One can then use Birkhoff's ergodic theorem to note that if $H:\Se\to\R$ is a measurable function that is integrable with respect to $\mu$, that is
\[
\int_\Se |H(\bz)|\,\mu(d\bz)<\infty
\]
then there exists a measurable integrable function $\bar H$ such that, with probability one
\[
\lim_{t\to\infty}\frac{1}{t}\sum_{s=0}^t H(\BZ(s)) = \bar H(\BZ(0))
\]
The measure $\mu$ is a ergodic probability measure if $\bar H$ is a constant function for all bounded measurable functions $h$ and one has, with probability one,
\[
\lim_{t\to\infty}\frac{1}{t}\sum_{s=0}^t H(\BZ(s)) = \int_\Se H(\bz)\mu(d\bz).
\]
Thus in the case of an ergodic initial distribution, the empirical observed averages are equal to the expectation, whereas with a nonergodic invariant initial measure, the empirical process follows one of the ergodic measures from which the initial distribution is constructed, but which of these is uncertain.

Note that these results only tell us about the long term behavior if $\BZ(0)$'s initial distribution is an invariant probability measure. The methods developed below describe the long term behavior of the system for a wide class of ecological models, when $\BZ(0)$ is any positive initial condition.

For any species $i$ let $\Se^i:=\{\bz=(\bx,\by)\in\Se~:~x_i>0\}$
be the subset of the state space for which this species has
strictly positive density. If $\mu$ is an ergodic measure and because $\Se^i$ is an invariant set, i.e. if $\BZ(0)\in \Se^i$ then with probability 1 one has $\BZ(t)\in \Se^i, t\geq0$, we note that
$\mu(\Se^i)\in\{0,1\}$. This allows us to define the
\textit{species support} of $\mu$ by $\Se(\mu)=\{1\leq i\leq
n~:~\mu(\Se^i)=1\}$.

Suppose that we have a strict subcommunity of species at
stationarity characterized by the ergodic measure $\mu$.
Assume that at least one species will be extinct, so that $\Se(\mu)\neq
\{1,\dots,n\}$. Let $i\notin \Se(\mu)$ and introduce species $i$
into the system at an infinitesimally low density. It turns out
(\cite{H81, CE89, C94, BS19}) that a key quantity is the
\textit{expected per-capita growth rate} of species $i$ when
introduced in the community described by $\mu$
\begin{equation}\label{e:r}
r_i(\mu) =\int_{\Se}\E[\log F_i(\bz,\xi(1))]\,\mu(d\bz).
\end{equation}
This quantity is the invasion rate in the parlance of invasibility analysis and  tells us whether a species $i\notin \Se(\mu)$
tends to increase or decrease if introduced at a low density. In
addition, if $i\in \Se(\mu)$ then we can show that $r_i(\mu)=0$,
since the species supported by $\mu$ already are at stationarity
so they do not tend to grow or decrease exponentially fast.
Denote by $\M$ the set of all ergodic invariant probability
measures supported on $\Se_0$. For a subset
$\wtd\M\subset \M$, denote by $\Conv(\wtd\M)$ the convex hull of
$\wtd\M$,
that is the set of probability measures $\pi$ of the form
$\pi(\cdot)=\sum_{\mu\in\wtd\M}p_\mu\mu(\cdot)$
with $p_\mu\geq 0,\sum_{\mu\in\wtd\M}p_\mu=1$.
Using this notation $\Conv(\M)$ is the set of all
invariant probability measures supported on $\Se_0$.
\begin{rmk}
Note that one can show (see \cite{B18}) that under some natural assumptions the set $\Conv(\M)$ is convex and compact and $\mu$ is ergodic if and only if it cannot be written as a nontrivial convex combinations of invariant probability measures. The ergodic decomposition theorem tells us that any invariant probability measure is a convex combination of ergodic measures. Furthermore, it can be shown that any two ergodic probability measures are either identical or mutually singular and that the topological supports of any mutually singular invariant measures are disjoint.
\end{rmk}

\begin{rmk}
In order to make the paper accessible to a wide variety of readers we will present our results without proofs in this section. The interested reader is asked to see Appendices \ref{s:a} and \ref{s:b} for the complete proofs of the various propositions and theorems.
\end{rmk}

The first result tells us that expected growth rates are well-defined and always zero for ``resident species" in the parlance of invasibility analysis.
\begin{prop}\label{p:rate}
Suppose $\mu$ is an ergodic invariant measure. Then $r_i(\mu)$
exists and is finite. Moreover,
\[
r_i(\mu)=0, i\in \Se(\mu).
\]
\end{prop}
Using this Proposition we can prove the following persistence
result.
\begin{thm}\label{t:pers1_disc}
Suppose that for all $\mu\in\Conv(\M)$ we have
\begin{equation}\label{e:per}
\max_{i} r_i(\mu)>0.
\end{equation}
Then \eqref{e:system_discrete} is almost surely stochastically
persistent and stochastically persistent in probability.
\end{thm}
This result is a generalization to non compact state spaces of
Theorem 1 by \cite{BS19}.

Although criterion \eqref{e:per} implies intuitively that species tend to increase when rare, and so are pushed away from the extinction set $\Se_0$, some extra condition, which is supplied here by A3), is required, as the following example from \cite{C82} shows. Stochastic
persistence depends not only on the recovery of species from low
densities but also on how violently a population can crash to
low densities. Suppose $N(t)$ is the population density of a
small mammal population that has the dynamics of geometric growth
up to a threshold $K>0$, followed by random crashes
from high density
\[
N(t+1) = \left\{
        \begin{array}{ll}
            2 N(t) & \quad N(t) < K \\
            \xi(t) N(t) & \quad N(t)\geq K.
        \end{array}
    \right.
\]
The environmental variable $\xi(t)$ is assumed to take values in
$(0,1)$ while $\xi(1), \xi(2),\dots$ is an i.i.d. sequence. Let $\delta_0$ be the Dirac mass at zero. This is the invariant measure that describes the `empty subcommunity' where all species are extinct.
Clearly
\[
r(\delta_0) = \ln 2 >0
\]
and $N(t)$ increases at low values no matter what the
distribution of $\xi(t)$ is. Moreover,
$\PP(\lim_{t\to\infty} N(t)=0)=0$ because $N(t)$ cannot only remain below $K$ for finite periods of time. Nevertheless, there are
situations in which $N(t)$ converges to $0$ in probability. For
simplicity assume that $\log_2 N(0), \log_2 \xi(t) \in \Z$. This
implies that $\log_2 N(t)$ is a homogenous Markov chain on $\Z$.
If $\xi(t)$ is such that $\E \log_2 \xi(t) = -\infty$, then the
process $\log_2 N(t)$ is null recurrent and the expected return
times to a given state are infinite. Given that $\log_2 N(t)$ is supported on $(-\infty, 1 + \log_2 K]$, it follows that for all
$x\in \Z$
\[
\lim_{t\to\infty}\PP(\log_2 N(t) >x)=0
\]
and therefore that $N(t)\to 0$ in probability. As a result
$N(t)$ is not persistent in probability even though
$r(\delta_0)>0$.

The above example shows that one needs to ensure that the
population will not violently crash to low densities. In the
generality of our setting, this is accomplished by part A3) of
Assumption \ref{a:main}.
\begin{rmk}
We note that the condition above is equivalent \citep{HN16} to being able to
find positive numbers $p_1,\dots,p_n$ such that
\[
\sum_{i=1}^np_ir_i(\mu)>0
\]
for all $\mu\in\M$. This says that there exist weights of the species such that the weighted average of the expected per-capita growth rates is positive for all ergodic measures supporting a strict subset of the community. This criterion has first been introduced by \cite{H81}.
\end{rmk}
In order to prove stronger persistence results, we need some
assumptions about the points of the state space the process can
visit. We will follow the notation, methods and results
developed by \cite{MT, B18}. A point $\tilde \bz\in \Se$ is said
to be \textit{accessible} from $\bz\in\Se$ if for every
neighborhood $U$ of $\tilde \bz$ there exists $t\geq 0$ such
that $P_t(\bz,U)= P_t\1_U(\bz)>0$. Define $$\Gamma_\bz:=\{\tilde
\bz\in\Se~|~\tilde \bz~\text{is accessible from}~\bz\}$$ and for
$A\subset \Se$
\[
\Gamma_A=\bigcap_{\bz\in A}\Gamma_\bz.
\]
Note that $\Gamma_A$ is the set of points which are accessible
from every point of $A$. We say a set $A$ is \textit{accessible}
if for all $\bz\in \Se_+$
\[
\Gamma_{\bz}\cap A\neq \emptyset.
\]
\begin{thm}\label{t:pers2_disc}
In addition to the assumptions of Theorem \ref{t:pers1_disc}
suppose there exist $\bz^*\in \Gamma_{\Se_+}$, a neighborhood
$U$ of $\bz^*$, a non-zero measure $\iota$ on $\Se_+$, and a
probability measure $\gamma$ on $\mathbb{Z}_{+}$ such that for
all $\bz\in U$
\begin{equation}\label{e:phi}
\sum_{i=1}^\infty\gamma(i)P_i(\bz,\cdot)\geq \iota(\cdot).
\end{equation}
Then there exists a unique invariant probability measure $\pi$
on $\Se_+$ and, with probability one, as $t\to\infty$ the
occupation measures $(\Pi_{t,\bz})_{t\in\N}$ converge weakly to $\pi$
for any $\bz\in\Se_+$.
\end{thm}

This result can be strengthened if the distribution of times, $\gamma$, can be replaced by a single time $m^*$, as follows.
\begin{thm}\label{t:pers3_disc}
Suppose that the assumptions of Theorem \ref{t:pers2_disc} hold,
and there is a positive integer, $m^*$, for which
\[
P_{m^*}(\bz,\cdot)\geq \iota(\cdot)
\]
for all $\bz\in U$.
Then there exists a unique invariant probability measure $\pi$
on $\Se_+$ for which the distribution of $\BZ(t)$
converges in total variation to $\pi$ as $t\to\infty$ whenever
$\BZ(0)=\bz\in\Se_+$.
\end{thm}
These results provide significant generalizations to Theorems 2
and 3 by \cite{SBA11} where only compact state spaces were
considered.
\begin{rmk}
For a deeper understanding of the concepts of accessibility and
irreducibility we refer the reader to the work by \cite{MT} and
Sections 4.2 and 4.3 by \cite{B18}. Having proven Theorem
\ref{t:pers1_disc}, the proofs of Theorems
\ref{t:pers2_disc} and \ref{t:pers3_disc} follow from the well
known convergence results of \cite{MT} and \cite{B18}.
\end{rmk}
\subsection{Extinction}
A complete understanding of community dynamics requires not only understanding when species persist, but also when species go extinct. Mathematically,
we are interested in conditions under which the process modelled
by \eqref{e:system_discrete} converges to the boundary $\Se_0$.

We need one extra condition, that ensures that the process can
get close to the boundary, in order to get our extinction
results.
\begin{asp}\label{a:ext}
There exists a function $\phi:\Se\to(0,\infty)$ and constants
$C, \delta_\phi>0$ such that for all $\bz\in\Se$
\begin{equation}\label{e1:ext}
\E_\bz V(\BZ(1))\leq V(\bz)-\phi(\bz)+C
\end{equation}
and
\begin{equation}\label{e2:ext}
\E_{\bz}\left(V(\BZ(1))-PV(\bz)\right)^2+\E\left|\log
F(\bz,\xi(1))-\E\log F(\bz,\xi(1))\right|^2\leq \delta_\phi
\phi(\bz).
\end{equation}
where $V$ is the function satisfying Assumption \ref{a:main}.
\end{asp}

\begin{rmk}
To obtain the extinction results, we need to manage the fluctuation of the process $\BZ(t)$.
This assumption bounds the variation of the process $\BZ(t)$. We need that the next step position $\BZ(n+1)$ (resp. $\log\BZ(n+1)$) does not fluctuate too widely from the current position $\BZ(n)$ (resp. $\log\BZ(n)$). Mathematically, we need certain boundedness of the quadratic terms $\E_{\bz}\left(V(\BZ(1))-PV(\bz)\right)^2$ and $\E\left|\log
F(\bz,\xi(1))-\E\log F(\bz,\xi(1))\right|^2$. Using well-known Lyapunov arguments, the boundedness can be obtained if the quadratic terms above are bounded by the dissipative part $V(\bz)-\E_\bz V(\BZ(1))$ (which is assumed to be bounded below, so it can be of the form $\phi(\bz)-C$ for some nonnegative function $\phi(z)$).

In many cases, the function $\phi(z)$ in \eqref{e1:ext} can be chosen to be $V(\bz)$ or $V^2(\bz)$. In general the function$\phi(\bz)$ by looking at  $V(\bz)-\E_\bz V(\BZ(1))$.

Having Assumption \ref{a:ext}, we can show the
family of random occupation measures $(\widetilde\Pi_t)_{t\in\N}$ is
tight. It is the discrete time analogue of Assumption 1.4 by
\cite{HN16}.
\end{rmk}
The first result tells us that under a certain condition there
is asymptotic extinction with probability $1$ no matter what the
initial densities are.
\begin{thm}\label{t:ext1_disc}
If there exist positive numbers $p_1,\dots,p_n$ such that
\[
\sum_{i=1}^np_ir_i(\mu)<0
\]
for all ergodic probability measures supported by $\Se_0$ and
$\Se_0$ is accessible then there exists $\alpha>0$ such that for
any $\bz\in\Se$
\[
\PP_\bz\left(\limsup_{t\to\infty}\frac{\ln
d(\BZ(t),\Se_0)}{t}=-\alpha\right)=1,
\]
where $d(\bz, \Se_0)= \min_i x_i$.
\end{thm}
The above theorem tells us when there is extinction but does not
tell us exactly which species go extinct. In order to gain this
extra information we need a few more definitions.

For a nonempty subset $I\subset \{1,\dots,n\}$ define
$\Se^I:=\{(\bx,\by)\in\Se~|~x_i=0, i\notin I\}$. Note that $\Se^I$ is the subspace in which all species not in $I$ are absent, and some or all species from $I$ are present. The set $\Se^I$ represents a subcommunity where we can define persistence and extinction sets relative to that subcommunity.

 Let $\Se^\emptyset = \{(0,\by)\in\Se\}$ be the set where all species are extinct. If we restrict the
process to $\Se^I$ then the extinction set, where at least one species from $I$ is extinct, is given by
$\Se_0^I:=\{\bz\in \Se^I~|~\prod_{j\in I}x_j=0\}$ and the
persistence set, where all species from $I$ persist, by $\Se_+^I:=\Se^I\setminus \Se_0^I$.

Let $\M^I:=\{\mu\in\M~|~\mu(\Se^I)=1\}$ be the set of ergodic measures supported on the subspace $\Se^I$. A measure from $\M^I$ will therefore support all, or a subset of, the species from $I$.  Let $\M^{I,+}:=\{\mu\in\M~|~\mu(\Se_+^I)=1\}$ be the set of ergodic measures supported on the subspace $\Se^{I}_+$. A measure from $\M^{I,+}$ does not put any mass on the extinction set $\Se^I_0$, i.e. all the species from $I$ are present and no species is extinct. Finally, we define $\M^{I,\partial}:=\{\mu\in\M~|~\mu(\Se_0^I)=1\}$ to be the set of ergodic probability measures supported on $\Se^I_0$ respectively. If a measure is in $\M^{I,\partial}$ this measure represents a subcommunity from $I$ where at least one species is extinct.

Denote the
(random) set of weak$^*$-limit points of $(\widetilde\Pi_t)_{t\in \N}$ by
$\U=\U(\omega)$.

\begin{deff}\label{a.extn}
Let $E_1$ be the family of subsets $I$ of $\{1,\dots,n\}$ such
that the following properties hold:
\begin{enumerate}
\item For all $I\in E_1$ we have that $\M^{I,+}\neq \emptyset$ and
  \begin{equation}\label{ae3.1}
\max_{i\notin I} r_i(\mu)<0, \mu\in \Conv\left(\M^{I,+}\right).
\end{equation}

  \item If $I\ne\emptyset$, suppose further that
for any $\nu\in\Conv(\M^{I,\partial})$, we have
\begin{equation}\label{ae3.2}
\max_{i\in I} r_i(\nu)>0.
\end{equation}
\end{enumerate}
Let $\mathcal{P}(n)$ be the set of all subsets of
$\{1,\dots,n\}$ and define $E_2:=\mathcal{P}(n)\setminus E_1$.
\end{deff}
\begin{rmk}
The subset $E_1$ tells us which subspaces $S_+^I$ are
attractors. If $I\in E_1$ then invariant probability measures
living on $S_+^I$ are attractors while invariant measures living
on the boundary $\Se^I_0$ are repellers. This way, if $\BZ(t)$
gets close to $\Se^I_+$ using \eqref{ae3.1} one can see the
process gets attracted towards $\Se^I_+$. The second condition
\eqref{ae3.2} ensures that the process gets repelled from the
boundary $\Se^I_0$ where the behaviour can be different.
\end{rmk}

\begin{thm}\label{t:ex}
For any $I \in E_1$ not equal to  $\{1,..,n\}$, there exists
$\alpha_I>0$ such that, for any compact set
$\K^I\subset \Se^{I}_+$,
$$
\lim_{\dist(\bz,\K^I)\to0, \bz\in
S_+}\PP_\bz\left\{\lim_{t\to\infty}\dfrac{\ln X_i(t)}t\leq
-\alpha_I, i\notin I\right\}=1.
$$
\end{thm}

\textit{\textbf{Biological Interpretation}: Suppose there exists
at least one subspace $\Se^{I}_+$ for some $I\subset
\{1,\dots,n\}$ that is attracting. An attracting subspace is
loosely speaking a subspace where every supported community at
stationarity cannot be invaded by any outsider species. Then, if
we start the system with densities $\bz\in\Se_+$ which are close
to $\Se_+^I$, the species which are not supported by $I$ go
extinct exponentially fast with a probability that can be made
arbitrarily close to $1$.}

\begin{thm}\label{t:ex22}
Assume either that $E_2=\emptyset$ or that for all $J\in E_2$
\[
\max_{i} r_i(\nu)>0, \nu\in\M^{J,+}.
\]
If $\bigcup_{I\in E_1}\Se^{I}_+$ is accessible then
$$
\sum_{I\in E_1} p_{\bz,I}=1
$$
where
$$p_{\bz,I}:=\PP_\bz\left\{\emptyset\neq\U(\omega)\subset\Conv\left(\M^{I,+}\right)
~\text{and}~\lim_{t\to\infty}\frac{\ln
X_j(t)}{t}\in\left\{r_j(\mu):\mu\in\Conv\left(\M^{I,+}\right)\right\},
j\notin I\right\}.$$
\end{thm}

\begin{figure}
\begin{center}
\begin{tikzpicture}[scale=0.7, line width = 1.2 pt]
\draw[->] (0,-0.5) -- (0,6);
\draw (0,6.2) node {$x_3$};
\draw[->](-0.5,0) -- (8,0);
\draw (8.3,0) node {$x_1$};
\draw[->](0.2,0.25) -- (-4,-5);
\draw (-4.1,-5.3) node {$x_2$};
\draw (-0.4,0.3) node {$\delta_0$};

\filldraw[fill=green!20!white,draw=green!50!black] (-1,-3) .. controls +(0.5,0.5) and +(-0.5,0)..(1.8,-2.5) .. controls +(0.5,0) and + (0.3,0.5) .. (2.5,-3.5) .. controls +(-0.3,-0.5) and + (0.3,0.5) .. (0.5,-4).. controls +(-0.3,-0.5) and + (0.1,-0.5) .. (-1.5,-4.5) .. controls +(-0.1,0.5) and + (-0.5,-0.5) .. (-1,-3);
\draw (0.5,-4.5) node {$\mu^1_{12}$};

\filldraw[fill=green!20!white,draw=green!50!black] (4.5,-1) .. controls +(0.5,0.5) and +(-0.5,0)..(6,-0.5) .. controls +(0.5,0) and + (0.3,0.5) .. (7,-1) .. controls +(-0.3,-0.5) and + (0.3,0.5) .. (5,-2).. controls +(-0.3,-0.5) and + (0.1,-0.5) .. (4,-2) .. controls +(-0.1,0.5) and + (-0.5,-0.5) .. (4.5,-1);
\draw (5.5,-2.2) node {$\mu^2_{12}$};

\draw[green!70!black, line width =2pt] (0,2.5)--(0,4.5);
\draw (-0.5,3.5) node {$\mu_3$};

\draw[red!70!black, line width =2pt] (-1,-1.25)--(-3,-3.75);
\draw (-2.6,-2.5) node {$\mu_1$};

\draw[red!70!black, line width =2pt] (2,0)--(6,0);
\draw (4,0.3) node {$\mu_2$};

\filldraw[red!70!black] (0,0) circle (2pt);

\end{tikzpicture}

\underline{}\caption{Example of ergodic invariant measures. Green indicates attractors while red indicates repellers. In this example we picked: $r_i(\delta_0)>0, i=1,2,3$, $r_2(\mu_1)>0, r_1(\mu_2)>0, r_1(\mu_3)<0, r_2(\mu_3)<0, r_3(\mu_{12}^1)<0$ and $r_3(\mu_{12}^2)<0$. This yields $E_1=\{\{1,2\}, \{3\}\}$, $\Se(\mu_1)=\{1\}, \Se(\mu_2)=\{2\}, \Se(\mu_3)=\{3\}, \Se(\mu_{12}^1)=\{1,2\}, \Se(\mu_{12}^2)=\{1,2\}$,  the attractors living on the open $x_1-x_2$ positive quadrant are $\M^{\{1,2\},+}=\{\mu_{12}^1,  \mu_{12}^2\}$ and the attractor living on the open $x_3$ positive axis is $\M^{\{3\},+}=\{\mu_3\}$.}
\label{fig:figure1}

\end{center}
\end{figure}
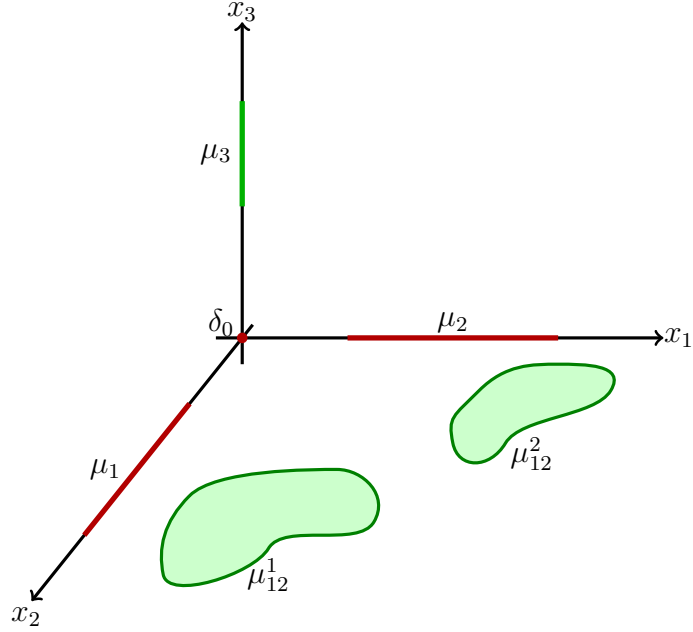

If we assume that there exists a unique attracting ergodic
measure in each attracting subspace we get a stronger result,
which is an immediate corollary of the above.

\begin{thm}\label{t:ex33}
Assume that every ergodic measure constructed as in Definition
\ref{a.extn} is a unique attractor in its subspace, i.e. for any
$I\in E_1$ there exists a unique $\mu_I$ satisfying
\eqref{ae3.1} and \eqref{ae3.2}. Furthermore, assume that
$E_2=\emptyset$ or that for all $J\in E_2$ and $\nu\in\M^{J,+}$
, $\max_{i} r_i(\nu)>0$.
Let $\bz \in S^+$. If $\bigcup_{I\in E_1}\Se^{I}_+$ is accessible then
$$
\sum_{I\in E_1} p_{\bz,I}=1
$$
where
$$p_{\bz,I}:=\PP_\bz\left\{\U(\omega)=\{\mu_I\}
~\text{and}~\lim_{t\to\infty}\frac{\ln X_j(t)}{t}=r_j(\mu_I),
j\notin I\right\}.$$
Furthermore, if for some $I\in E_1$ the subspace $\Se^I_+$ is
accessible one has $p_{\bz,I}>0$.
\end{thm}

\textit{\textbf{Biological Interpretation}: Our result says the following. Suppose that the subspaces
$\Se^{I}_+$ for $I\subset \{1,\dots,n\}$ are either of attracting or repelling type. An attracting subspace
is loosely speaking a subspace where every supported community
at stationarity cannot be invaded by any outsider species. A
repelling subspace is a subspace where every supported community
at stationarity can be invaded by at least one outsider species.
If the attracting subspaces $\Se_+^I, I\in E_1$ are accessible,
i.e. if starting at any population density one can get
arbitrarily close to $\Se_+^I$, then with strictly positive
probability the system converges to $\mu_I$. This means that the
species given by $I$ persist and converge to $\mu_I$ while the
species from $I^c$ go extinct.\newline
We note that the problem becomes much more complicated if there are subspaces that are neither attracting nor repelling. See \cite{HNS20} for an example which involves rock-paper-scissors dynamics.
}

\subsection{Robust persistence and extinction}
Since all mathematical models are merely approximations of real
ecological systems, it is important to study how the extinction
and persistence results change under small perturbations of the
analytical models. Ideally, small perturbations of the model
should lead to similar long-term behavior. We next show that this true in
our setting as long as we add one natural assumption.
Assume we perturb the dynamics of the system
\eqref{e:system_discrete}. We call $\tilde \BZ$ a
$\delta$-perturbation of $\BZ$ if the dynamics is given by
\begin{equation}\label{e:system_discrete_per}
\begin{aligned}
\tilde X_i(t+1)&=\tilde X_i(t) \tilde F_i(\tilde \BX(t), \tilde
\BY(t),\xi(t)), ~i=1,\dots,n\\
\tilde \BY(t+1) &=\tilde G(\tilde \BX(t),\tilde \BY(t),
\xi(t))
\end{aligned}
\end{equation}
and
\begin{equation}\label{e:robust1}
\sup_{\bz\in\Se} \sum_i\E\left[\| \log \tilde F_i(\bx,\by,\xi)-
\log F_i(\bx,\by,\xi) \|+\|\tilde G_i(\bx,\by,\xi)-
G_i(\bx,\by,\xi) \| \right]\leq \delta.
\end{equation}
Let $\circ$ denote the element-wise product and $\1_n$ be the
vector in $\R^n$ whose components are all $1$.
Assume that the function $V$ from Assumption \eqref{a:main} satisfies the
robust estimate
\begin{equation}\label{e:rob}
\E \left[V(\bx^\top (F(\bz,\xi)\circ(\1_n+\wtd\eps_1)),
G(\bz,\xi)+\wtd\eps_2))h(\bz,\xi)\right]\leq \rho V(\bz)+C,
\bz\in\Se
\end{equation}
for all vectors $\wtd\eps_1,\wtd\eps_2\in\R^n$ such that
$|\wtd\eps_1|\vee|\wtd\eps_2|<\delta$. We note that this
estimate is trivially satisfied if the state space is compact.
\begin{thm}\label{t:disc_robust}
Assume the robust estimate \eqref{e:rob} holds and the dynamics
\eqref{e:system_discrete} satisfy the hypotheses of one of the
persistence Theorems \ref{t:pers1_disc}-\ref{t:pers3_disc} or
one of the extinction Theorems \ref{t:ext1_disc}, \ref{t:ex22}
or \ref{t:ex33}. Then there exists $\delta>0$ such that every
$\delta$-perturbation of \eqref{e:system_discrete} satisfies the
same hypotheses and therefore has the same long term behavior
(persistence or extinction).
\end{thm}
This generalizes the results from Proposition 2 by \cite{SBA11}
to the setting when $\BZ$ is not restricted to a compact state
space and the dynamics is influenced both by the species
densities $\BX$ and by the auxiliary variables $\BY$.

\section{Continuous Time Results} \label{s:results}
We continue to use the notation from Section \ref{s:results0}, if not mentioned otherwise. We work on a complete probability space
$(\Omega,\F,\{\F_t\}_{t\geq0},\PP)$ with a filtration
$(\F_t)_{t\geq 0}$ satisfying the usual conditions.
Consider the system

\begin{equation}\label{e:system}
\begin{aligned}
dX_i(t)&=X_i(t) f_i(\BX(t), \BY(t))\,dt+X_i(t)g_i(\BX(t),
\BY(t))\,dE^i(t), ~i=1,\dots,n\\
dY_i(t) &= u_i(\BX(t), \BY(t))\,dt+ h_i(\BX(t),
\BY(t))\,dW_i(t), ~i=1,\dots,k
\end{aligned}
\end{equation}
taking values in $\Se=\R_+^n\times \R^{\kappa_0}$ where $\R_+^n:=[0,\infty)^n$ and
$\R_+^{n,\circ}:=(0,\infty)^n$. We assume
$(\BE(t), \BW(t))=(E^1(t),\dots, E^n(t), W_1(t),\dots, W_{\kappa_0}(t))^T=\Gamma^\top\BB(t)$ where
$\Gamma$ is a $(n+{\kappa_0})\times (n+{\kappa_0})$ matrix such that
$\Gamma^\top\Gamma=\Sigma=(\sigma_{ij})_{(n+{\kappa_0})\times (n+{\kappa_0})}$
and $\BB(t)=(B_1(t),\dots, B_{n+{\kappa_0}}(t))$ is a standard Brownian motion
on $\R^{n+{\kappa_0}}$ which is adapted to the filtration $(\F_t)_{t\geq 0}$.
The SDE \eqref{e:system} is describing the dynamics of $n$
interacting populations $\BX(t)=(X_1(t),\dots,X_n(t))_{t\geq
0}$.

\begin{rmk}
We note that we can also treat compact state spaces for the
species dynamics. That means we can study \eqref{e:system} when
$\BX$ takes values in a compact state space, which for our ecological applications will usually be $\Delta:=\{\bx\in
\R^n_+~|~\sum_i x_i=1\}$. In this setting we assume the drift
and the diffusion terms of $\BX$ are such that the dynamics
remains on $\Delta$, or, in other words, that $\Delta$ is an
invariant set for the dynamics. See \cite{SBA11} for more examples of what the theory looks like in a compact state space.
\end{rmk}
Let $\op$ be the infinitesimal generator of the process $\BZ$. For smooth enough functions $F:\R_+^n\times\R^{\kappa_0}\to \R$ the generator $\op$ acts as

\[
\begin{aligned}
\op F(\bz) &= \sum_i x_if_i(\bz)\frac{\partial F}{\partial x_i}(\bz) +  \sum_i  u_i(\bz)\frac{\partial F}{\partial x_i}(\bz)  + \frac{1}{2}\sum_{i,j}\sigma_{ij}x_ix_jg_i(\bz)g_j(\bz)\frac{\partial^2 F}{\partial x_i \partial x_j}(\bz) \\
&~~+ \frac{1}{2}\sum_{i,j}\sigma_{n+i, n+{\kappa_0}}h_{i}(\bz)h_j(\bz)\frac{\partial^2 F}{\partial y_i \partial y_j}(\bz)+ \frac{1}{2}\sum_{i,j}\sigma_{i,n+j}x_ig_i(\bz)h_j(\bz)\frac{\partial^2 F}{\partial x_i \partial y_j}(\bz).
\end{aligned}
\]
The following assumptions will be made in the continuous time
setting.
\begin{asp}\label{a:sde}
The coefficients of \eqref{e:system} satisfy the following
conditions:
\begin{enumerate}
\item $f_i, g_i, u_i, h_i:\R_+^n\times\R^{\kappa_0}\to \R$ are
locally Lipschitz for any $i=1,\dots,n$.
\item There exist numbers $C_4, \gamma_4, \gamma_5>0$ and a function $V\in
C^2$ satisfying:
\begin{itemize}
\item $V(\bz)\geq 1, \bz\in\Se$
\item $\lim_{|\bz|\to\infty} V(\bz)=\infty$
\item $V(\bz)\geq |\bx|^{\gamma_5}, \bz=(\bx,\by)\in\Se$
\item $$\op V(\bz)\leq \left(C_4-\gamma_4 \sum_{i=1}^n
(1+|f_i(\bz)|+|g_i(\bz)|^2\right) V(\bz), \bz\in\Se.$$
\end{itemize}
\end{enumerate}
\end{asp}
\begin{rmk}\label{r:sde_assum}
Assumption \ref{a:sde} generalizes Assumption 1.1 in \cite{HN16}. Here we assume the existence of a Lyapunov function to obtain a certain boundedness of the solution (in a moment sense). In many cases, the function $V(\bz)$ will be of the form $$V(\bz)=1+\sum_{i=1}^n{c_i x_i}+\sum_{j=1}^{\kappa_0} y_i^2.$$
	When the diffusion is not too large and the drift points towards the origin as $|\bz|$ is large then the assumption is often satisfied. We note that in ecological terms this boils down to requiring that when species densities are large, there is a strong negative drift due to competition. 

Under Assumption \ref{a:sde} one can show like in Lemma 3.1 from \cite{HN16} that $\BZ(t)$ is a Feller process that has pathwise unique solutions. Furthermore, if the process starts with $\BZ(0)\in \Se_+$ then with probability 1 the process stays forever in $\Se_+$.
\end{rmk}

One can associate to the Markov process $\BZ(t)=(\BX(t),\BY(t))$
the semigroup $(P_t)_{t\geq 0}$ defined by its action on bounded Borel
measurable functions $h:\Se\to\R$ via
\[
P_th(\bz)=\E_\bz[h(\BZ(t))], t\geq 0, \bz\in \Se.
\]
The operator $P_t$ can be seen to act by duality on
Borel probability measures $\mu$ by $\mu\to \mu P_t$ where $\mu
P_t$ is the probability measure for which
\[
\int_{\Se}h(\bx) (\mu P_t)(d\bx):=\int_{\Se}P_th(\bx) \mu(d\bx)
\]
for all $h\in C_b(\Se)$.
A probability measure $\mu$ on $\Se$ is called invariant if
$\mu P_t=\mu$ for all $t\geq 0$. The invariant probability
measure $\mu$ is called ergodic if it cannot be written as a
nontrivial convex combination of invariant probability measures.
For any species $i$ let $\Se^i:=\{\bz=(\bx,\by)\in\Se~:~x_i>0\}$
be the subset of the state space for which this species has
strictly positive density. If $\mu$ is an ergodic measure
$\mu(\Se^i)\in\{0,1\}$ we have defined the
\textit{species support} of $\mu$ by $\Se(\mu)=\{1\leq i\leq
n~:~\mu(\Se^i)=1\}$.

We can define in this setting for any $t>0$ the normalized random
occupation measures
\[
\tilde \Pi_t(\cdot) = \frac{1}{t}\int_0^t \1_{\{\BZ(s)\in\cdot\}}\,ds,
\]
the occupation measures
\[
 \Pi_{t,\bz}(\cdot) = \E_\bz\tilde \Pi_t (\cdot),
\]
and the expected per-capita growth rates
\[
r_i(\mu) =
\int_{\Se}\left(f_i(\bz)-\frac{\sigma_{ii}g_i^2(\bz)}{2}\right)\mu(d\bz), i=1,\dots,n.
\]
Many of the discrete time results from Section \ref{s:results0} will hold in the continuous time setting. They represent generalizations of the work by \cite{HN16}.
\begin{prop}\label{p:rate2}
Suppose $\mu$ is an ergodic invariant measure. Then for any
$i\in I$ the quantity $r_i(\mu)$ is well defined and finite.
Moreover
\[
r_i(\mu)=
\int_{\Se}\left(f_i(\bz)-\frac{\sigma_{ii}g_i^2(\bz)}{2}\right)\mu(d\bz)=0,
i\in \Se(\mu).
\]
\end{prop}
\begin{rmk}
The above proposition provides a very powerful tool for
computing $r_\ell (\mu), \ell\notin \Se(\mu)$ when
$f_i(\bz)=\sum_j a_{ij}x_j$ and $g_i(\bz)=1$. In this setting we
get the linear system
\[
0=r_i(\mu)=\sum_{j}a_{ij}\hat x_j -\frac{\sigma_{ii}}{2}, i\in
\Se(\mu)
\]
where
\[
\hat x_j = \int_{\Se} x_j \mu(d\bz).
\]
If this system has unique solutions we can then compute
\[
r_\ell(\mu)=\sum_{j}a_{\ell j}\hat x_j -\frac{\sigma_{ii}}{2},
\ell\notin \Se(\mu).
\]
\end{rmk}

 A point $\tilde \bz\in \Se$ is said
to be \textit{accessible} from $\bz\in\Se$ if for every
neighborhood $U$ of $\tilde \bz$ there exists $t\geq 0$ such
that $P_t(\bz,U)= P_t\1_U(\bz)>0$. Define $$\Gamma_\bz:=\{\tilde
\bz\in\Se~|~\tilde \bz~\text{is accessible from}~\bz\}$$ and for
$A\subset \Se$
\[
\Gamma_A=\bigcap_{\bz\in A}\Gamma_\bz.
\]
Note that $\Gamma_A$ is the set of points which are accessible
from every point of $A$. We say a set $A$ is \textit{accessible}
if for all $\bz\in \Se_+$
\[
\Gamma_{\bz}\cap A\neq \emptyset.
\]

We have the following sequence of results. These can be seen as
the continuous time equivalents of the results from Section
\ref{s:results0}.
\begin{thm}\label{t:pers1_cont}
Suppose Assumption \ref{a:sde} holds and for all $\mu\in\Conv(\M)$ we have
\[
\max_{i} r_i(\mu)>0.
\]
Then \eqref{e:system} is almost surely stochastically persistent
and stochastically persistent in probability.
\end{thm}
Exactly like in Section \ref{s:results0} we define what it
means for a set to be accessible or for the state space to be
irreducible -- see \cite{MT} and Sections 4.2 and 4.3 by
\cite{B18} for more details.
\begin{thm}\label{t:pers2_cont}
 Suppose Assumption \ref{a:sde} holds and for all $\mu\in\Conv(\M)$ we have
\[
\max_{i} r_i(\mu)>0.
\] In addition, assume there exist $\bz^*\in \Gamma_{\Se_+}$, a neighborhood
$U$ of $\bz^*$, a non-zero measure $\iota$ on $\Se_+$, and a
probability measure $\gamma$ on $\R_{+}$ such that for
all $\bz\in U$
\begin{equation}\label{e:phi}
\int_0^\infty P_t(\bz,\cdot) \gamma(dt)\geq \iota(\cdot).
\end{equation} Then, with probability one, the occupation measures
$(\Pi_{t,\bz})$ converge weakly to a unique invariant
probability measure $\pi$ supported on $\Se_+$
for any $\bz\in\Se_+$.
\end{thm}

\begin{thm}\label{t:pers3_cont}
Suppose that the assumptions of Theorem \ref{t:pers2_cont} hold,
and there is a positive number, $m^*>0$, for which
\[
P_{m^*}(\bz,\cdot)\geq \iota(\cdot)
\]
for all $\bz\in U$. Then there exists a unique
invariant probability measure $\pi$ on $\Se_+$ and
the distribution of $\BZ(t)$ converges in total variation to
$\pi$ exponentially fast whenever
$\BZ(0)=\bz\in\Se_+$.
\end{thm}

In order to prove extinction results we need to make the
following extra assumption, which is the analogue of Assumption
\ref{a:ext} from the discrete time setting.
\begin{asp}\label{a:e_cont}
There is $C>0$ such that
$$
\left|\sum_{i=1}^n
V_{x_i}(\bz)x_ig_i(\bz)+\sum_{i=1}^{\kappa_0}V_{y_i}(\bz)h_i(\bz)\right|\leq
CV(\bz)\sum_{i=1}^n (1+|f_i(\bz)|+|g_i(\bz)|^2), \bz\in\Se.
$$
\end{asp}

Suppose the notation is the same as the one from Section \ref{s:results0}.
\begin{thm}\label{t:ex_cont}
If $E_1$ is nonempty, then for any $I\in E_1$, there exists
$\alpha_I>0$ such that, for any compact extinction set
$\K^I\subset \Se^{I}_+$,
we have
$$
\lim_{\dist(\bz,\K^I)\to0, \bz\in
S^\circ}\PP_\bz\left\{\lim_{t\to\infty}\dfrac{\ln X_i(t)}t\leq
-\alpha_I, i\notin I\right\}=1.
$$
\end{thm}
Let $\mathcal{P}(n)$ be the set of all subsets of
$\{1,\dots,n\}$ and define $E_2:=\mathcal{P}(n)\setminus E_1$.
\begin{thm}\label{t:ex2_cont}
Assume either that $E_2=\emptyset$ or that
$\max_{i} r_i(\nu)>0$ for any $\nu$ with $\nu\in\M^{J,+}$ for
some $J\in E_2$.
If $\bigcup_{I\in E_1}\Se^{I}_+$ is accessible then
$$
\sum_{I\in E_1} p_{\bz,I}=1
$$
where
$$p_{\bz,I}:=\PP_\bz\left\{\U(\omega)\subset\Conv\left(\M^{I,+}\right)
~\text{and}~\lim_{t\to\infty}\frac{\ln
X_j(t)}{t}\in\left\{r_j(\mu):\mu\in\Conv\left(\M^{I,+}\right)\right\},
j\notin I\right\}.$$
\end{thm}
\begin{rmk}
We note that if $(E_1,\dots, E_n)$ is non-degenerate, i.e. the
matrix $\Sigma$ is invertible, then each subspace $\Se^I_+$ is
accessible.
\end{rmk}
If we assume that there exists a unique attracting ergodic
measure in each attracting subspace we get a stronger result,
which is an immediate corollary of the above.

\begin{thm}\label{t:ex3_cont}
Assume that each $I\in E_1$ is such that $\M^{I,+} = \{\mu_I\}$, i.e. for any
$I\in E_1$ there exists a unique $\mu_I$ satisfying
\eqref{ae3.1} and \eqref{ae3.2}. Furthermore, assume that that
$E_2=\emptyset$ or that
$\max_{i} r_i(\nu)>0$ for any $\nu\in\M^{J,+}$ for some $J\in
E_2$.
If $\bigcup_{I\in E_1}\Se^{I}_+$ is accessible then
$$
\sum_{I\in E_1} p_{\bz,I}=1
$$
where
$$p_{\bz,I}:=\PP_\bz\left\{\U(\omega)=\{\mu_I\}
~\text{and}~\lim_{t\to\infty}\frac{\ln X_j(t)}{t}=r_j(\mu_I),
j\notin I\right\}.$$
Furthermore, if for some $I\in E_1$ the subspace $\Se^I_+$ is
accessible, then $p_{\bz,I}>0$.
\end{thm}
The above theorems generalize the earlier persistence results by
\cite{SBA11} and the persistence and extinction results by
\cite{HN16}.

\subsection{Nondegenerate stochastic differential equations}
We present the simplest setting, with no auxiliary variables, where all the accessibility and irreducibility conditions are satisfied.
The dynamics of the $n$ species is given by
\begin{equation}\label{e:system_2}
dX_i(t)=X_i(t) f_i(\BX(t))\,dt+X_i(t)g_i(\BX(t))\,dE_i(t), ~i=1,\dots,n
\end{equation}
taking values in $[0,\infty)^n$.
Assume the following.
\begin{asp}\label{a.nonde}  The coefficients of \eqref{e:system} satisfy the following conditions:
\begin{itemize}
\item[(1)] $\diag(g_1(\bx),\dots,g_n(\bx))\Gamma^\top\Gamma\diag(g_1(\bx),\dots,g_n(\bx))=(g_i(\bx)g_j(\bx)\sigma_{ij})_{n\times n}$
is a positive definite matrix for any $\bx\in\R^{n}_+$.
\item[(2)] $f_i(\cdot), g_i(\cdot):\R^n_+\to\R$ are locally Lipschitz functions for any $i=1,\dots,n.$
\item[(3)] There exist $\bc=(c_1,\dots,c_n)\in\R^{n,\circ}_+$ and $\gamma_b>0$ such that
\begin{equation}\label{a.tight}
\limsup\limits_{\|x\|\to\infty}\left[\dfrac{\sum_i c_ix_if_i(\bx)}{1+\bc^\top\bx}-\dfrac12\dfrac{\sum_{i,j} \sigma_{ij}c_ic_jx_ix_jg_i(\bx)g_j(\bx)}{(1+\bc^\top\bx)^2}+\gamma_b\left(1+\sum_{i} (|f_i(\bx)|+g_i^2(\bx))\right)\right]<0.
\end{equation}
\end{itemize}
\end{asp}
Note that part (1) of Assumption \ref{a.nonde} means that the diffusion is nondegenerate. Meanwhile, \eqref{a.tight} is the special case of Assumption \ref{a:sde}, part (2) when there are no auxiliary variables and when the function $V$ is chosen to be $1+\bc^\top\bx$.
As shown in Example 1.1 from \cite{HN16} condition \eqref{a.tight} is satisfied for most ecological models \citep{SBA11, EHS15, HNY16, HN16,  B18}.
One can show \citep{HN16} that Assumption \ref{a.nonde} implies that all the irreducibility and accessibility conditions from Theorems \ref{t:pers1_cont}-\ref{t:ex3_cont} hold.

Consider $\mu\in\M$.
Assume $\mu\neq \bdelta^*$ where $\bdelta^*$ is the Dirac mass at the origin, that is, at $(0,\dots,0)$.. Since the diffusion $\BX$ is non degenerate in each subspace,
there exist $0<n_1<\dots< n_k\leq n$
such that $\suppo(\mu)= \R^\mu_+$ where
$$\R_+^\mu:=\{(x_1,\dots,x_n)\in\R^n_+: x_i=0\text{ if } i\in \Se(\mu)^c\}$$
for
$\Se(\mu):=\{n_1,\dots, n_k\}$ and
$\Se(\mu)^c:=\{1,\dots,n\}\setminus\{n_1,\dots, n_k\}$. If $\mu=\bdelta^*$ then we note that $\R^{\bdelta^*}_+=\{\0\}$.
Let
$$\R_+^{\mu,\circ}:=\{(x_1,\dots,x_n)\in\R^n_+: x_i=0\text{ if } i\in \Se(\mu)^c\text{ and }x_i>0\text{ if  }x_i\in \Se(\mu)\}$$ and $\partial\R_+^{\mu}:=\R_+^\mu\setminus\R_+^{\mu,\circ}$. In this setting all the irreducibility and accessibility conditions hold and we have the following results.
\begin{thm}\label{t:pers4_cont}
Suppose that for all $\mu\in\Conv(\M)$
\[
\max_i\lambda_i(\mu)>0.
\]
 Then there exists a unique
invariant probability measure $\pi$ on $\R_+^{n,\circ}$ and
the distribution of $\BX(t)$ converges in total variation to
$\pi$ exponentially fast whenever
$\BX(0)=\bx\in\R_+^{n,\circ}$.
\end{thm}

The following condition will imply extinction.
\begin{asp}\label{a.extnn}
There exists a $\mu\in\M$ such that
\begin{equation}\label{ae3.1a}
\max_{i\in \Se(\mu)^c}\{\lambda_i(\mu)\}<0.
\end{equation}
If $\R_+^\mu\ne\{\0\}$, suppose further that
for any $\nu\in\Conv(\M_\mu)$ , we have
\begin{equation}\label{ae3.2a}
\max_{i\in \Se(\mu)}\{\lambda_i(\nu)\}>0
\end{equation}
where $\M_\mu:=\{\nu'\in\M:\suppo(\nu')\subset\partial\R_+^\mu\}.$
\end{asp}
Define
\begin{equation}\label{e:M1}
\M^1:=\left\{\mu\in\M : \mu ~~\text{satisfies Assumption} ~~\ref{a.extnn}\right\},
\end{equation}
and
\begin{equation}\label{e:M2}
\M^2:=\M\setminus\M^1.
\end{equation}
To characterize the extinction of specific populations, we need some additional conditions.
\begin{asp}\label{a.extn2}
Suppose that there is a $\delta_1>0$ such that
$$
\lim\limits_{\|\bx\|\to\infty} \dfrac{\|\bx\|^{\delta_1}\sum_i g_i^2(\bx)}{1+\sum_i(|f_i(\bx)|+|g_i(\bx)|^2)}=0.
$$
\end{asp}
\begin{rmk}
This is a technical assumption which says that the noise terms do not grow too fast. We offer here some intuition. An application of It\^o's Lemma yields that		
$$
\dfrac{\ln X_i(t)}t=\dfrac{\ln X_i(0)}t+\dfrac1t\int_0^t\left[f_i(\BX(s))-\dfrac{g_i^2(\BX(s))\sigma_{ii}}2\right]ds+\dfrac1t\int_0^t g_i(\BX(s))dE_i(s).
$$

If $\BX$ is close to the support of an ergodic invariant measure $\mu$ for a long time,
then
$$\dfrac1t\int_0^t\left[f_i(\BX(s))-\dfrac{g_i^2(\BX(s))\sigma_{ii}}2\right]ds$$
can be approximated by the average with respect to $\mu$
$$r_i(\mu)=\int_{\partial \R^n_+}\left(f_i(\bx)-\dfrac{g_i^2(\bx)\sigma_{ii}}2\right)\mu(d\bx).$$
while the term $$\dfrac{\ln X_i(0)}t$$ is negligible as $t\to \infty$. We need Assumption \eqref{a.extn2} in order to guarantee that as $t\to\infty$ we have with probability 1 that
\[
\dfrac1t\int_0^t g_i(\BX(s))dE_i(s) \to 0.
\]
\end{rmk}
To describe exactly which populations go extinct,
we need an additional assumption which ensures that apart from those in $\Conv(\M^1)$,
invariant probability measures are ``repellers''.
\begin{asp}\label{a.extn3}
Suppose that one of the following is true
\begin{itemize}
  \item $\M^2=\emptyset$

  \item For any $\nu\in\Conv(\M^2)$, $\max_i\lambda_i(\nu)>0.$
\end{itemize}
\end{asp}

\begin{thm}\label{t:ex2cont}
Suppose that Assumptions \ref{a.nonde}, \ref{a.extn2} and \ref{a.extn3} are satisfied and $\M^1\neq \emptyset$.
Then for any $\bx\in\R^{n,\circ}_+$
\begin{equation}\label{e0-thm4.2}
\sum_{\mu\in\M^1}p_{\bx}(\mu)=1
\end{equation}
where
$$p_{\bx}(\mu):=\PP_\bx\left\{\U(\omega)=\{\mu\}\,\text{ and }\,\lim_{t\to\infty}\dfrac{\ln X_i(t)}t=\lambda_i(\mu)<0, i\in \Se(\mu)^c\right\}>0, \bx\in\R^{n,\circ}_+, \mu\in \M^1.$$
\end{thm}
\begin{figure}
\begin{center}
\begin{tikzpicture}[scale=0.7, line width = 1.2 pt]
\fill[fill=green!20!white] (0,0) -- (8,0) -- (4,-5) -- (-4,-5) --cycle;
\draw (1.5,-2.5) node {$\mu_{12}$};

\draw[->, line width = 2pt, green!70!black] (0,-0.5) -- (0,6);
\draw (0,6.2) node {$x_3$};
\draw[->, line width = 2pt, color=red!70!black](-0.5,0) -- (8,0);
\draw (8.3,0) node {$x_1$};
\draw[->, line width = 2pt, color=red!70!black](0.2,0.25) -- (-4,-5);
\draw (-4.1,-5.3) node {$x_2$};
\draw (-0.4,0.3) node {$\delta_0$};


\draw[green!70!black, line width =2pt] (0,2.5)--(0,4.5);
\draw (-0.5,3.5) node {$\mu_3$};

\draw[red!70!black, line width =2pt] (-1,-1.25)--(-3,-3.75);
\draw (-2.6,-2.5) node {$\mu_1$};

\draw[red!70!black, line width =2pt] (2,0)--(6,0);
\draw (4,0.3) node {$\mu_2$};

\filldraw[red!70!black] (0,0) circle (2pt);

\end{tikzpicture}

\caption{Example with attractors and repellers when the noise is nondegenerate. In this example we picked: $r_i(\delta_0)>0, i=1,2,3$, $r_2(\mu_1)>0, r_1(\mu_2)>0, r_1(\mu_3)<0, r_2(\mu_3)<0$, and $r_3(\mu_{12})<0$. This yields $E_1=\{\{1,2\}, \{3\}\}$, $\Se(\mu_1)=\{1\}, \Se(\mu_2)=\{2\}, \Se(\mu_3)=\{3\}, \Se(\mu_{12})=\{1,2\}$. Each subspace can have at most one ergodic measure. We have $\M^{\{1,2\},+}=\{\mu_{12}\}$ and $\M^{\{3\},+}=\{\mu_3\}$. Note that the support of every ergodic measure is the whole subspace it lives in, for example: $\mu_1$ lives on $\R_+^{\mu_1,\circ}:=\{(x_1,0,0)\in\R^3_+: ~x_1>0\}$ and $\mu_{12}$ lives on $\R_+^{\mu_{12},\circ}:=\{(x_1,x_2,0)\in\R^3_+: ~x_1, x_2>0\}$.}
\label{fig:figure2}

\end{center}
\end{figure}
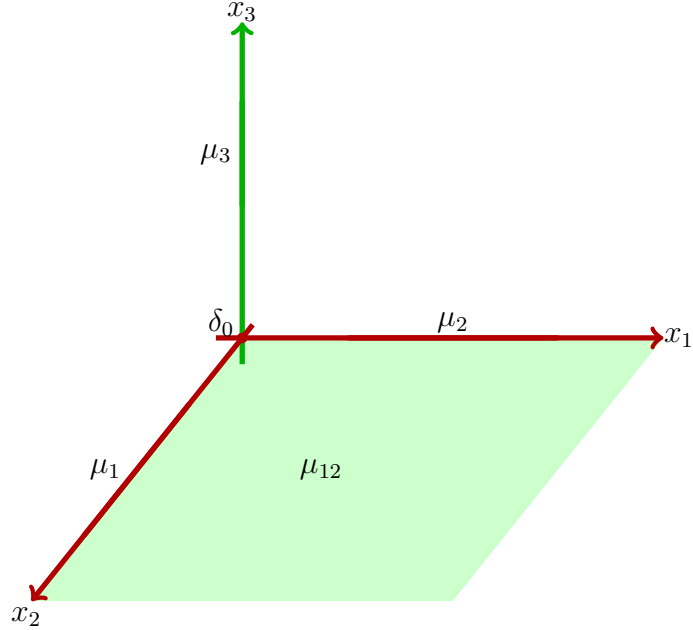
\begin{rmk}
The proofs of the continuous time SDE results from this section
are very similar to the ones for the discrete time results from
Section \ref{s:results0}. They have therefore been omitted.
\end{rmk}
An example showcasing the various definitions is presented in Figure \ref{fig:figure2}.

\section{Applications in discrete time}\label{s:app_disc}
\subsection{Dynamics in Markovian environments}
As formulated in terms of equations \eqref{e:system_discrete}, the discrete time system represents population dynamics subject to i.i.d environmental variation defined by the variable $\xi(t)$. However, the auxiliary variable $\BY(t)$ might also be an environmental variable. Physical environmental variables in ecological models commonly reflect no feedback from the population variables, although in nature such feedback occurs but is not necessarily strong or easy to characterize. If $\BY(t)$ has feedback only from itself, and not population densities, it becomes a Markov process and allows the equations \eqref{e:system_discrete} to represent population dynamics in a Markovian environment as first suggested by \cite{BS19}. In this case the system \eqref{e:system_discrete} simplifies to
\begin{equation}\label{e:1s1v_disc}
\begin{aligned}
X_i(t+1)&=X_i(t) F_i(\BX(t), \BY(t),\xi(t)),\\
\BY(t+1)&= G(\BY(t),\xi(t)).
\end{aligned}
\end{equation}

As formulated here, both $\BY(t)$ and $\xi(t)$ function as environmental variables, specifying the environment during the period $[t,t+1)$, which then influence the value of $\BX(t+1)$. As a special case of \eqref{e:system_discrete}, the general extinction and persistence theorems above continue to apply, although Assumptions \ref{a:main} and \ref{a:ext} might be easier to check. We illustrate the issues that apply to Markovian environments with examples for single species, and pairs of interacting species. We begin with the assumption that $\BY(t)$ has finitely many ergodic
measures $\mu^1_\BY,\dots, \mu^\ell_\BY$.

For an individual species $i$ we have
\[
r_i(\delta_0\times \mu^j_\BY) = \int_{\R^{\kappa_0}}\E\left[\ln
F_i(0,\by,\xi(1))\right]\mu^j_\BY(d\by).
\]
\subsubsection{One species}
Species $i$ will persist on its own if
\[
r_i(\delta_0\times \mu^j_\BY)>0
\]
for $j=1,\dots,\ell.$ This implies that at each ergodic probability measure of the environment, the log growth rate of species $i$ at zero has to be positive. Once one checks which species will persist on their own, the next step is to look at two species systems and see which two species systems persist in isolation. We do this in the next section.

We note that if the Markovian environment $\BY$ has only one ergodic measure $\mu_\BY$ then computations are significantly simplified.
\begin{rmk}
\cite{DF99} give abstract conditions for when the environment
\[
\BY(t+1)= G(\BY(t),\xi(t))
\]
has a unique invariant measure. Suppose for each $\xi$ the function $G(\cdot,\xi)$ is Lipschitz with constant $K_\xi$ and
\begin{itemize}
\item $\E [K_{\xi(1)}]<\infty$
\item $ \E |G(\by_0,\xi(1))-\by_0|<\infty$ for some $\by_0\in\R^{\kappa_0}$
 \item $\E [\ln K_{\xi(1)}]<\infty$
\end{itemize}
By Theorem 1.1 from \cite{DF99} the process $\BY(t)$ has a unique invariant probability measure $\mu_\BY$.
One example would be to assume that $\BY(t+1)=A\BY(t)+\xi(t)$.
Then, if $A$ has spectral radius less than one, the process $\BY(t)$ converges
(\cite{DF99}) to a unique invariant probability measure
$\mu_\BY$.
\end{rmk}
Suppose from now on that $\BY$ has a unique ergodic measure $\mu_\BY$. In this case the dynamics is determined by
\[
r_i(\delta_0\times \mu_\BY) = \int_{\R^{\kappa_0}}\E\left[\ln
F_i(0,\by,\xi(1))\right]\mu_\BY(d\by).
\]
If $r_i(\delta_0\times \mu_\BY)>0$ we have persistence. If
$r_i(\delta_0\times \mu_\BY)<0$ the species goes extinct almost
surely exponentially fast
\[
\lim_{t\to\infty}\frac{\ln X(t)}{t}=r_X(\delta_0\times \mu_\BY).\]

\begin{rmk}One other way of modeling the environment would be to assume
there are finitely many environmental states $\{1,\dots,\ell\}$
and $\BY(t)$ is an irreducible Markov chain with transition
probabilities $p_{ij}:=\PP[\BY(t+1)=j~|~\BY(t)=i]$. The Markov chain
will have a unique stationary distribution
$\nu:=(\nu_1,\dots,\nu_\ell)$. In this setting the only ergodic
measure on $\Se_0$ is $\delta_0\times\nu$ and the
persistence/extinction depends on the sign of
\[
r_X(\delta_0\times \nu) = \sum_{i=1}^\ell\nu_i\E\left[\ln
F(0,i,\xi(1))\right].
\]
If $r_X(\delta_0\times \nu)>0$ we have persistence. If
$r_X(\delta_0\times \nu) <0$ the species goes extinct almost
surely exponentially fast
\[
\lim_{t\to\infty}\frac{\ln X(t)}{t}=r_X(\delta_0\times \nu).
\]
\end{rmk}

\subsubsection{Two species}\label{s:disc_2d}
Consider two interacting species that experience the effects of a Markovian environment
\begin{equation}\label{e:2d_disc}
\begin{aligned}
X_1(t+1)&=X_1(t) f_1(X_1(t), X_2(t),\BY(t), \xi(t)),\\
X_2(t+1)&=X_2(t) f_2(X_1(t), X_2(t),\BY(t),\xi(t))\\
\BY(t+1) &= G(\BY(t),\xi(t)).
\end{aligned}
\end{equation}
We assume that $\BY(t)$ has only one ergodic probability measure $\mu_\BY$. Our results can be applied as follows. The
first step is to check that Assumptions \ref{a:main} and \ref{a:ext} hold - this is
done once one knows more detailed properties of the system
(noise and interaction terms). We exhibit how to check
assumptions in specific examples in the next sections. Once the
assumptions are checked, we look at the measure
$\delta_0\times\mu_\BY$
\[
r_i(\delta_0\times \mu_\BY) = \int_{\R^{\kappa_0}}\E\left[\ln
f_i(0,\by,\xi(1))\right]\mu_\BY(d\by).i=1,2.
\]
If $r_i(\delta_0\times \mu_\BY) >0$ then species $i$ survives on its own and the system
converges to a unique invariant probability measure $\mu_i$
supported on $\Se_+^i := \{\bx\in\Se~|~x_i\neq 0, x_j=0, i\neq
j\}$. Remember that the (random) set of weak$^*$-limit points of
the family of occupation measures $(\widetilde\Pi_t)_{t\in \N}$ is denoted
by $\U=\U(\omega)$. Thus, if we say that $\U(\omega)=\{\mu_1\}$,
this means that for the realization $\omega$ we have $\widetilde\Pi_t \to
\mu_1$ weakly.
\begin{enumerate}[label=(\roman*)]
\item Suppose $r_1(\delta_0\times \mu_\BY), r_2(\delta_0\times \mu_\BY)  >0$. The expected
per-capita growth rates can be computed via
  \[
  r_i(\mu_j)=\int_{(0,\infty)\times \R^{\kappa_0}}\E[\ln f_i(x,\by, \xi(1))]\mu_j(dx d\by).
  \]
  \begin{itemize}
\item If $r_1(\mu_2)>0$ and $r_2(\mu_1)>0$ we have coexistence
and convergence of the distribution of $\BX (t)$ to the unique
invariant probability measure $\pi$ on $\Se_+$.
\item If $r_1(\mu_2)>0$ and $r_2(\mu_1)<0$ we have the
persistence of $X_1$ and extinction of $X_2$. In other words,
for any $\bx\in\Se_+$
    \[
\PP_\bx\left\{\U(\omega)=\{\mu_1\}
~\text{and}~\lim_{t\to\infty}\frac{\ln X_2(t)}{t}=r_2(\mu_1)<0,
\right\}=1.
    \]
\item If $r_1(\mu_2)<0$ and $r_2(\mu_1)>0$ we have the
persistence of $X_2$ and extinction of $X_1$. In other words,
for any $\bx\in\Se_+$
    \[
\PP_\bx\left\{\U(\omega)=\{\mu_2\}
~\text{and}~\lim_{t\to\infty}\frac{\ln X_1(t)}{t}=r_1(\mu_2)<0,
\right\}=1.
    \]
\item If $r_1(\mu_2)<0$ and $r_2(\mu_1)<0$ we have that for any
$\bx\in\Se_+$
    \[
p_{\bx,j}:=\PP_\bx\left\{\U(\omega)=\{\mu_j\}
~\text{and}~\lim_{t\to\infty}\frac{\ln X_i(t)}{t}=r_i(\mu_j)<0,
i\neq j \right\}
    \]
    and
    \[
    p_{\bx,1}+ p_{\bx,2}=1.
    \]
  \end{itemize}
\item Suppose $r_1(\delta_0\times\mu_\BY)>0, r_2(\delta_0\times\mu_\BY)<0$. Then species
$1$ survives on its own and the system converges to the unique invariant
probability measure $\mu_1$ on $\Se^1_+$.
  \begin{itemize}
\item If $r_2(\mu_1)>0$ we have the persistence of both species
and convergence of the distribution of $\BX(t)$ to the unique
invariant probability measure $\pi$ on $\Se_+$.
\item If $r_2(\mu_1)<0$ we have the persistence of $X_1$ and the
extinction of $X_2$. In other words, for any $\bx\in\Se_+$
    \[
\PP_\bx\left\{\U(\omega)=\{\mu_1\}
~\text{and}~\lim_{t\to\infty}\frac{\ln X_2(t)}{t}=r_2(\mu_1)<0,
\right\}=1.
    \]
  \end{itemize}
\item Suppose $r_1(\delta_0\times\mu_\BY)<0, r_2(\delta_0\times\mu_\BY)<0$. Then both
species go extinct almost surely, i.e., for any $\bx\in\Se_+$
\[
\PP_\bx\left\{\lim_{t\to\infty}\frac{\ln
X_i(t)}{t}=r_i(\delta_0\times\mu_\BY)<0 \right\}, i=1,2.
\]
\end{enumerate}

These results are generalizations of those by \cite{E89, BS19,
SBA11}. More specifically, we do not require any monotonocity or compactness, we show which species persist and which go extinct, and we also prove that there is always convergence to an ergodic probability measure supported by the persistence or extinction sets. Higher dimensional systems can be treated on a case by case basis. Usually, it is not possible to find the expected per-capita growth rates if the system is not of Lotka-Volterra type and the dimension is higher than 3. For a full classification in continuous time of three dimensional systems see \cite{HNS20}.

\subsection{Structured populations}

Populations of plants and animals
generally have important internal structure.  The most obvious structure is that the organisms consist of different ages. Importantly, different aged organisms can have different mortality rates, and different contributions to reproduction. Another way of structuring populations, which can be more helpful than age in some cases, is simply by life stage, for example as, egg, larva, and adult.
In addition, populations do not live at points in space, but are spread over an area
or region, with migration across the region. Invariably, the environments occupied by the population in different places differ somewhat
in environmental characteristics, which can also be important to account for in the total growth of the population. When the region inhabited can be treated as discrete habitat patches, each with its own subpopulation of the species in question, the population structure is again discrete, just as stage, and age as an integer would be. There have been multiple studies of the
interplay between dispersal and environmental heterogeneity
\cite{H83, C85, C00, GH02, S04, RHB05, S10, CCL12, DR12}.
We are interested in models of $m$ interacting populations that
live in stochastic environments. The population is structured
because each individual from population $i$ can be in one of
$n_i$ individual states (these could be age, size or location).
Then $\bar X_i(t) = (X_{i1}(t),\dots,X_{in_i}(t))$ will be the
row vector of population densities of individuals in different
states for population $i$ at time $t\in\N$ and $\bar X_i(t)$
will live on $\R_+^{n_i}$. The population state will be given by
the row vector $\bar \BX(t) = (\bar X_1(t),\dots, \bar X_m(t))$
that will live on $\R_+^n$ for $n:=\sum_{i=1}^m n_i$. Just as
before, the environment will be represented by the sequence of
i.i.d random variables $\xi(1), \xi(2),\dots, \xi(t),\dots$ where
$\xi(t)$ is the environment at the time step $t$. Assume that the
environmental random variables are independent of the initial
condition $\bar \BX(0)$ and take values in a probability space
$E$ that is equipped with a $\sigma$-field and a probability
measure $m$. For all $i$, let $A_i(\xi,\bar \BX)=
(a_i^{j,k}(\xi,\bar \BX))$ be a non-negative $n_i\times n_i$
matrix whose $j-k$th entry corresponds to the individuals moving
from state $j$ to state $k$. One can then write the population
dynamics as
\begin{equation}\label{e:struc}
\bar X_i(t+1)=\bar X_i(t) A_i(\xi(t),\bar \BX(t)).
\end{equation}
These models have been studied in \cite{RS14} and in a
particular example in \cite{BS19} under the restrictive
assumption that the dynamics is constrained to a compact subset
$K\subset \R_+^n$. Moreover, there are no general extinction
results in \cite{RS14, BS19}.

We can transform \eqref{e:struc} to our framework as follows.
Define
\[
X_i(t) :=\sum_{j=1}^{n_i} X_{ij}(t)
\]
to be the population size of species $i$, $\BX(t) =
(X_1(t),\dots, X_m(t))$,
\[
Y_{ij}(t)= \frac{X_{ij}(t)}{X_i(t)}
\]
the fraction of population $i$ in state $j$ and $\BY_i(t) =
(Y_{i1}(t),\dots ,Y_{in_i}(t)), \BY(t) =
(\BY_1(t),\dots,\BY_m(t))$. Then, simple calculations show that
\begin{equation}\label{e:transf}
\begin{aligned}
X_i(t+1) &= X_i(t) \sum_{l,j} Y_{il}
a_i^{j,l}(\xi(t),\BX(t),\BY(t)),~ i=1,\dots,m\\
Y_{ij}(t+1) &= \frac{\sum_{l}Y_{il}
a_i^{j,l}(\xi(t),\BX(t),\BY(t))}{\sum_{l,j}Y_{il}
a_i^{j,l}(\xi(t),\BX(t),\BY(t))}, ~i=1,\dots,m~j=1,\dots,n_i.
\end{aligned}
\end{equation}
Note that this system is in the correct form \eqref{e:struc} and
the state space is given by $\Se :=
\R_+^m\times\Delta_1\times\dots\times\Delta_m$ where
$\Delta_i:=\{\bx\in \R_+^{n_i}~:~\sum_{j=1}^{n_i}x_j=1\}$ is the
simplex in $\R^{n_i}$.

The dynamics is well-defined on the extinction set $\Se_0 =
\partial R_+^m \times \Delta_1\times\dots\times\Delta_m$. For
example, if $X_1=0$ then the dynamics is
\begin{equation*}
\begin{aligned}
X_1(t+1)&=0\\
X_i(t+1) &= X_i(t) \sum_{l,j} Y_{il}
a_i^{j,l}(\xi(t),0,X_2(t),\dots, X_m(t),\BY(t)),~
i=2,\dots,m\\
Y_{ij}(t+1) &= \frac{\sum_{l}Y_{il}
a_i^{j,l}(\xi(t),0,X_2(t),\dots,
X_m(t),\BY(t))}{\sum_{l,j}Y_{il}
a_i^{j,l}(\xi(t),0,X_2(t),\dots, X_m(t) ,\BY(t))},
~i=1,\dots,m~j=1,\dots,n_i.
\end{aligned}
\end{equation*}
We will show how our theory can be used to generalize results by
\cite{RS14, BS19}. Suppose there is only one species. We abuse
notation and denote by $(X_1(t),\dots, X_n(t))$ the densities of
$n$ patches and let $X(t)= \sum_{i=1}^n X_i(t), Y_i(t) =
\frac{X_i(t)}{X(t)}$ be the total population size and population
fractions. We get the dynamics given by
\begin{equation}\label{e:transf_1d}
\begin{aligned}
X(t+1) &= X(t) \sum_{l,j} Y_{l}
a_i^{j,l}(\xi(t),X(t),\BY(t)),~ i=1,\dots,m\\
Y_{j}(t+1) &= \frac{\sum_{l}Y_{l}
a^{j,l}(\xi(t),X(t),\BY(t))}{\sum_{l,j}Y_{il}
a^{j,l}(\xi(t),\BX(t),\BY(t))}, ~j=1,\dots,n
\end{aligned}
\end{equation}
with state space $\Se := \R_+ \times\Delta_n$ and extinction set
$\Se_0 = \{0\}\times \Delta_n$. One can show under certain
assumptions that $\BY(t)$ has a unique invariant measure $\hat
\mu$ on $\Se_0$. We can use this to compute a growth rate
$r(\mu)$ such that the following theorem holds.
\begin{thm} Suppose that Assumption \ref{a:main} holds. If $r(\mu)>0$ then $\BX(t)$ is
stochastically persistent. In addition, under certain
irreducibility conditions, there exist a unique invariant
probability measure $\pi$ and the law of $\BX(t)$ converges to
$\pi$ in total variation exponentially fast.

If $r(\mu)<0$ and Assumption \eqref{a:ext} holds, then for all $\bx\in
\R_+$
\[
\PP_\bx \left(\lim_{t\to \infty}X(t)= 0\right)=1. \]
\end{thm}

\subsection{Environment-density interaction models in ecology}
A key question in community ecology, which the results here are ideally placed to address, is the role of temporal environmental fluctuations in species coexistence, i.e. the mutual persistence of interacting species. Models in community ecology have identified situations in which environmental fluctuations are essential for coexistence, or at least promote it, undermine coexistence, or have no effect on it \citep{C94}. To systematize understanding, models of ecological community dynamics have been expressed in the following form,
\begin{equation}\label{e:ED}
N_j(t+1)=G_j(E_j(t),D_j(t))N_j(t),
\end{equation}
with $N_j(t)$ being the population density of species $j$ at time $t$. The quantity $G_j(E_j(t),D_j(t))$ is the multiplication rate, sometimes called the ``finite rate of increase,'' and corresponds to $F_i(\BX(t), \BY(t),\xi(t))$, in the developments above. The two arguments, $E_j(t)$ and $D_j(t)$, have special roles.  The first of these, $E_j(t)$, is a component of the environmental variable $\xi(t)$, and is referred to as the environmental response of species $j$. The second variable, $D_j(t)$, is the species response to density-dependent processes. Thus, $D_j(t)$ will reflect the densities of the interacting species even if it is not directly a function of them. For example, $D_j(t)$ might be a function of resource shortage, or predator abundance, which could be auxiliary variables $\BY(t)$ affected dynamically by the focal species densities, the $N_j(t)$. Originally, $D_j(t)$ was conceived as reflecting competition, and designated $C_j(t)$, but the new notation, $D_j(t)$, is intended to indicate a broader class of species interactions including ``apparent competition,'' which is generated by density-dependent predation \citep{KC10, CK10}. By way of contrast, the environmental response does not reflect species densities, although it affects them. It is important to note, however, that $D_j(t)$ will generally depend on the $E_j(t)$, but the slope of the relationship between $D_j(t)$ and $E_j(t)$ decreases to 0 as the density, $N_j(t)$, of species $j$ approaches 0, a fact that has critical effects on species coexistence \citep{C94, C19}. The relationship between $D_j(t)$ and $E_j(t)$ is quantified by a statistical covariance between $E_j(t)$ and $D_j(t)$, denoted covED$_{j}$. Positive covED$_{j}$ means that favorable environmental conditions for species $j$ are offset by unfavorable density effects limiting the gains in population that the species can have during favorable environmental conditions. However, when covED$_{j}$ is evaluated using an invariant measure for which species $j$ has zero density, $D_j(t)$ will not be a function of $E_j(t)$ and covED$_{j}$ can be low or even negative depending on how much $D_j(t)$ depends on the environmental responses of other species and how the environmental responses of the different species are correlated. These changes in covED$_{j}$ for different invariant measures have critical effects on species coexistence provided one other condition applies, which is generally referred to as "buffered population growth" \citep{C19}.

Buffered population growth is defined in terms of the quantity
\begin{equation}\label{e:interaction}
\gamma_j = \frac{\partial^2  g_j(E,D)}{\partial E \partial D},
\end{equation}
where $g_j(E_j(t),D_j(t)) = \ln G_j(E_j(t),D_j(t)) = \ln N_j(t+1)-\ln N_j(t)$, defining population growth on the log scale.  When $\gamma_j$ is negative, population growth is said to be buffered because it means that high values of $D_j(t)$ (unfavorable density effects) cause lower declines in $g_j(E_j(t),D_j(t))$ when the environment is unfavorable too (low $E_j(t)$). In effect, a double dose of unfavorable conditions is less than doubly unfavorable. As a consequence, $r_j(\mu)$, which is the same as the expected value of $\ln G_j(E_j(t),D_j(t))$, depends negatively on covED$_j$. The lower values of covED$_j$ encountered for invariant measures, $\mu$, for which species $j$ is extinct, thus favor positive $r_j(\mu)$ values and hence persistence of species $j$. This is the "storage effect" species coexistence mechanism. It is quantified for the whole community (the "community average approach" \citep{C08, YC15}) in terms of $\gamma_j$ multiplied by differences between covED values for invariant measures $\mu$ for which a species is extinct compared with those for which it is persistent. Importantly, the quantity $\gamma_j$ is a reflection of the life-history properties of the organisms encoded in the model equations, and covED$_j$ reflects how the direct ($E_j(t)$) and indirect responses ($D_j(t)$) of a species $j$ to the environment are related to each other \citep{CH88}. Below we study several different versions of this general model indicating how these key contributors to species coexistence emerge.

\subsubsection{Perennial organisms: lottery models}
The lottery model of \cite{CW81} was the first model of competition between species in which stochastic temporal environmental variation emerged as a mechanism of species coexistence. It provides an example of the storage effect coexistence mechanism in operation \citep{C83, C94}. In this model, the individual organisms
have exclusive sites or territories that provide the resources needed for the life
of that individual. It is assumed that the space available for
territories is strictly limiting. An individual organism must secure its own territory
if it is to mature (``recruit to adulthood'') and reproduce, but once it has secured a site, it retains
it for the rest of its life. We use a form of the model due to \cite{C83}. In this form,
the variables $N_i(t)$ are the adult densities of the species, and the environmental responses $E_i(t)$ are per capita juvenile production rates, or in simple terms, ``birth rates.'' Thus, $J_i(t) = E_i(t)N_i(t)$ is the number of juveniles of species $i$ in year $t$ seeking a site to mature as an adult.  The relative ability of a juvenile to secure a site is given is a function, $c_i(\mathbf{J}(t))$, of the juvenile densities of all species, and so the fraction of available sites secured during year $t$ by species $i$ is

\[
\frac{c_i(\mathbf{J}(t))J_i(t)}{\sum_{j=1}^n c_j(\mathbf{J}(t))J_j(t))}.
\]
The organisms compete for space given up by adult death, which in year $t$ is
\[
\sum_{j=1}^n \delta_j(t)N_j(t).
\]
where $\delta_j(t)$ is the fraction of the adult population of species $j$ that dies during $(t,t+1)$ and is assumed to be function of $\xi(t)$. As a consequence, the dynamics of the community are given by the following equations
\begin{equation}\label{e:lottery}
N_i(t+1) = (1-\delta_i(t)) N_i(t) + \left(\sum_{j=1}^n
\delta_j(t) N_j(t)\right)
\left(\frac{c_i
(\mathbf{J}(t))J_i(t)}{\sum_{j=1}^n c_j
(\mathbf{J}(t))J_j(t)}\right), ~i=1,\dots,n.
\end{equation}
To put this equation in the $G(E,D)$ form \eqref{e:ED}, we can first define
\[
D_i(t) = \frac{ \sum_{j=1}^n c_j(\mathbf{J}(t))J_j(t)}  {c_i(\mathbf{J}(t)) \sum_{j=1}^n\delta_j(t) N_j(t)}.
\]
The numerator here is the total demand for space taking account of competitive ability, and the denominator is the competitive ability of species $i$ times the supply of space, i.e. $D_j(t)$ is``demand over supply'' adjusted for relative competitive ability between species. The dynamical equation \eqref{e:lottery} now can be put in the form \eqref{e:ED}:

\begin{equation}\label{e:lotstandard}
N_i(t+1) = \left(1 - \delta_i(t)  + \frac{E_i(t)}{D_i(t)} \right)N_i(t).
\end{equation}
In this version of the lottery model, the function $G_i$ will also vary with the environment if $\delta_i(t)$ is temporally variable. Note that $\sum_{j=1}^n N_j(t)$ remains constant for all $t\in \Z_+$ and is equal to the total
area available for individual sites. We set this total space equal to 1 by choice of units. Therefore the dynamics take place on the compact state space
$\Delta:=\{\bx\in \R^n_+~|~\sum_i x_i=1\}$. Note that on a compact state space Assumptions \ref{a:main} and \ref{a:ext} are automatically satisfied.

If there are only two species, variations on this model have been analyzed in
detail by \cite{C82, C83} and \cite{SBA11}. Note that the extinction set is
$\Delta_0 = \{(1,0), (0,1)\}$. As such, we only have to compute
the per capita growth rates of the Dirac masses $\delta_{(1,0)},
\delta_{(0,1)}$. This is trivial and yields
\[
\lambda_1:=r_1(\delta_{(0,1)}) = \E\left[\ln\left(1-\delta_1(1)
+ \delta_1(1)
\frac{c_1(0,E_2(0))E_1(0))}{c_2(0,E_2(0))E_2(0)}\right)
\right]
\]
and
\[
\lambda_2:=r_2(\delta_{(1,0)}) = \E\left[\ln\left(1-\delta_2(1)
+ \delta_2(1)
\frac{c_2(E_1(0),0)E_2(0)}{c_1(E_1(0),0)E_1(0))}\right)
\right]
\]
It is not hard to see that the irreducibility and accessibility
assumptions are satisfied for many choices of $c_i$,
$\delta_i$ and $E_i$. For example, it is sufficient for $\delta_i(0) > 0$, a.s., and the ratio
$c_1(\mathbf{J}(0))E_1(0)/c_2(\mathbf{J}(0))E_2(t))$ to have a positive continuous probability density
function over the domain $(0,\infty)$, conditional on the $\delta_i(0)$ and $\mathbf{N}(0)$. Then next result presents generalizations of \cite[ Theorems 3.5, 5.1 and 5.2]{C82}.
\begin{thm}\label{t:lott}
The following hold:
\begin{itemize}
\item If $\lambda_i>0, i=1,2$ then the system is stochastically
persistent and furthermore the distribution of $(N_1(t), N_2(t))$ converges to a
unique invariant probability measure $\pi$ on
$\Delta\setminus\Delta_0= \{(x_1, x_2)\in\R_+^2, x_1, x_2>0,
x_1+x_2=1\}$.
\item If $\lambda_i>0$ and $\lambda_j<0$ then $\PP(N_i(t)\to 1,
N_j(t) \to 0)=1$.
\item If $\lambda_1<0$ and $\lambda_2<0$ then $\PP(N_1(t)\to 1
~\text{or}~ N_2(t) \to 0)=1$ and $\PP(N_i(t)\to 1)>0, i=1,2$ .
\end{itemize}
\end{thm}
Approximations to these low-density growth rates, $\lambda_i$, have been developed in number of publications for both the two-species and multispecies cases \citep{C89, C94, C03}.  Most important, the approximations reveal the biological circumstances leading to the three different possibilities in Theorem \ref{t:lott}. In the case where the functions $c_j$ are simply constants, there is no possibility of coexistence unless the environment varies. Most important, sufficient variation in the ratio $E_1/E_2$ about the value 1, while the $\E[\ln(1 - \delta_i(0))]$ are bounded from $-\infty$, guarantees coexistence, i.e. stochastic persistence of both species \citep{CW81}. This example is important as the first demonstration of how a stochastically varying environment could promote species coexistence. This ability was subsequently traced to a negative value for $\gamma_i$ and covED  \citep{CH88, C94}. This negative value of $\gamma_i$ means that the growth of the population is less sensitive to the density-dependent variable $D_i(t)$, in this case competition, when environmental conditions for recruitment, as given by $E_i(t)$, are poor. The outcome is that the population can increase strongly during environmentally favorable times without suffering catastrophic losses from competition under environmentally poor conditions. This property defines the storage component of the storage effect coexistence mechanism. This feature arises quite generally from common life-history properties of natural populations, in particular, consistently high adult survival while recruitment to the adult population is highly variable \citep{CH88}.

The full storage effect mechanism combines this property with covED, which explains why a species at a lower density gains boosts to population growth by having a higher frequency of times when $E_i$ is high while $D_i$ is low. For example, for constant $\delta's$, and species 1 present alone ($ 1 \in \Se(\mu), 2 \notin \Se(\mu)$ in a two-species setting), covED$_1$ is positive because then $D_1 = E_1/\delta_1$. However, in this setting $D_2 = (c_1(E_1,0)/c_2(E_1,0))E_1\delta_1)$, which need not be strongly correlated with $E_2$ leading to a low covED$_2$ for species 2. As a result of this low covED, species 2 can have times when it is strongly favored by the environment while experiencing low competition. The fact that it also has times when it is disfavored by the environment and experiences high competition is not so important due to the negative value of $\gamma_1$.  Species 1, however, has no such advantages due to its high covED, and so there is an overall net gain to species 2 in this setting, and generally to a species at low density invading the population of another species. This is the storage effect mechanism, and explains how the variable environment boosts the $\lambda's$ and thereby promotes coexistence in the lottery and other models, as discussed further below.

These coexistence results depended on highly variable juvenile production rates as the choice for $E_i(t)$. In contrast, if adult survival is made highly variable and $1 - \delta_i(t)$ is chosen as $E_i(t)$, $\gamma_i$ is positive, making the population growth rate of a species more sensitive to competition under unfavorable environmental conditions.  As a consequence, increasing variation in $(1 - \delta_1(t))/(1 - \delta_2(t))$ about 1 leads to the third possibility in Theorem \ref{t:lott}, where $\lambda_1$ and $\lambda_2$ are both negative, and one of the species must go extinct, without the identity of that species being predictable (``random exclusion'') due to a negative storage effect from positive $\gamma_i$ \citep{CW81}.
Finally, the second case, where $\lambda_i > 0$ and $\lambda_j < 0$, giving certain persistence of one species, and certain extinction of the other, occurs when $\delta_1 = \delta_2 = 1$, regardless of the magnitude of environmental variation.  In this case $\gamma_i$ = 0, and there is no storage effect, either positive or negative.  Each of the three outcomes in Theorem \ref{t:lott} is also possible without environmental variation but due to dependence of the $c_i(\mathbf{J}(t))$ on the juvenile densities $\mathbf{J(t)}$. For example, if this dependence specifies that intraspecific effects are stronger than interspecific effects, then coexistence occurs, as is well-understood in deterministic models \cite{C18}. However, the key interest here is how the stochastic environment can create these outcomes.

\subsubsection{Perennial organisms: Ricker recruitment
variation}
Like the lottery model, this model has long-lived adults, and juveniles that compete for the opportunity to recruit into the adult population, but in this case competition is not for space, but instead for a specific resource of richness $S(t)$ that fluctuates over time with the environment. Competition takes the Ricker form \citep{C94}, and the equations can be written
\begin{equation}\label{e:perr}
N_i(t+1)=N_i(t)\left(1-\delta_i+S(t)e^{E_i(t)-D(t)}\right)
\end{equation}
with
\begin{equation}\label{e:Cj}
D(t)=\sum_{j=1}^n \alpha_j e^{E_j(t)}N_j(t).
\end{equation}
Here $E_j(t)$ can be interpreted as the logarithm of the per capita number of
births in the interval $(t,t+1)$. The density response variable, $D(t)$, is the same for all species in this case, and so is not subscripted. This model also takes the $G(E,D)$ form \eqref{e:ED} if $S(t)$ is a constant. More generally, this form can be retained if $E_j(t)$ is replaced by $E_j(t) + \ln S(t)$, but in the definition of $D$ this extra component, $ \ln S(t)$, would need to be subtracted from $E_j(t)$ to preserve the model. Most important, like the lottery model with variable juvenile production, in this model the interaction coefficient $\gamma_i$ is negative, and covED is present due to the dependence of $D_j(t)$ on the  $E_j(t)$. These features allow environmental variation to enable coexistence by the storage effect \citep{C94, C03}. In the absence of environmental variation, coexistence at a stable equilibrium is not possible due to the fact that all species share the same density response variable $D(t)$ \citep{CH97}. However, sustained deterministic fluctuations are a possibility in this model, and can support coexistence under some circumstances by a mechanism termed ``relative nonlinearity'' \citep{KC08}. The theory presented here, however, enables rigorous demonstration of when coexistence can occur in a stochastic environment.

The next lemma gives conditions under which the general assumptions hold for this class of models.
\begin{lem}\label{l_per}
Suppose $\{(E_i(t), S(t))_{i=1,\dots, n}\}, t\in\N$ is a sequence of $n+1$-dimensional random variables, i.i.d. over $t$ such that $\E  \left[S(t)e^{E_j(t)}\right]^2<\infty.$ Then the model given by \eqref{e:per} and \eqref{e:Cj} satisfies Assumption \ref{a:main} by taking a small enough $\gamma_3>0$,
and
\[
V(\bz) = \sum_j z_j + 1.
\]
Assumption
\ref{a:ext} holds with
\[
\phi(\bz) = \delta V(\bz)
\]
for some $\delta>0$.
Moreover, if the support of $\ln S(t) + max_j (E_j(t)\ln \delta_j)$ contains values less than 0 then the boundary is accessible.
\end{lem}
\begin{rmk}
We have chosen to work with $E_j(t)$ as well as $S(t)$ as part
of the environment $\xi$. So in this case there is no $\BY$
auxiliary variable, $\bz=\bx$, $\BX=(N_1(t),\dots,N_n(t))$ and
$\xi(t):=(E_1(t),\dots,E_n(t),S(t))$. We could instead treat
$S(t)$ as an auxiliary variable - the results would be
unchanged.
\end{rmk}
\begin{proof}
See Appendix \ref{s:d}.
\end{proof}

\subsubsection{Perennial organisms: Ricker recruitment without covED}
Models with Ricker recruitment have sometimes been proposed without covED, for example by \cite{E89}. The equation for the dynamics of $N_i(t)$ remain as defined for the previous example
 \eqref{e:perr}, but $D$ now gains a subscript, because it differs by species:
\begin{equation}\label{e:ricker}
D_i(t)=\sum_{j=1}^n \alpha_{ij}N_j(t).
\end{equation}
Here the $\alpha_{ij}$ are positive constants with $\alpha_{ii}$ = 1 and the others less than 1. Thus, the $\alpha$ coefficients differ by the species $i$ experiencing the competition, not just by the species $j$ causing it.  However, the most important difference with the previous model is the assumption that $D_i(t)$ involves the adult densities $N_j(t)$ not the juvenile densities $\exp(E_j(t))N_j(t)$. It is thus assumed that juveniles suffer from competition with the adult organisms, which restricts their ability to recruit into the adult population, while in the previous model instead, the juveniles are competing with each other for specific resources that juveniles need to make this transition. The outcome mathematically is that covED is not present, although $\gamma_i$ remains negative. Thus, the storage effect coexistence mechanism cannot function. However, coexistence can occur by other means. As $S(t)$ no longer plays any additional role, we shall assume $S(t) \equiv 1$  and that $((E_i(t))_{i = 1,..,n,})_{t \in \Z^+}$ is a sequence of i.i.d random
variables.
The next lemma gives conditions under which we can apply our general results.
\begin{lem}\label{l_ricker}
Suppose $\E[ (E_i(t))^2]<\infty.$ Then the model given by
\eqref{e:ricker} satisfies Assumption \ref{a:main} by taking a small
enough $\gamma_3>0$, and
\[
V(\bx) = \sum_j x_j + 1.
\]
Assumption
\ref{a:ext} holds with
\[
\phi(\bx) = \delta V(\bx)
\]
for some $\delta>0$.
Moreover, if the support of $\max_j (E_j(t) + \ln \delta_j)$ contains values less than 0 then the
boundary is accessible.
\end{lem}
\begin{proof}
The proof is very similar to the one of Lemma \ref{l_per} and is
therefore ommitted.
\end{proof}

Suppose we have only two species in \eqref{e:ricker}. Then, using
Lemma \ref{l_ricker} and following the general approach from
Section \ref{s:disc_2d} we need to first look at
\[
r_i(\delta_0) = \E[\ln f_i(0,E_i(1))] = \E [\ln (\exp(E_i(1))+1 - \delta_i)].
\]
Suppose that $r_i(\delta_0)>0, i=1,2$. Then $X_1, X_2$ survive
on their own and have invariant probability measures $\mu_1,
\mu_2$ on $(0,\infty)$. By Section \ref{s:disc_2d} we need to
look at
\[
r_i(\mu_j)=\int_{(0,\infty)}\E\ln [1-\delta_i+
\exp(E_i(1)-\alpha_{ij}x) ]\mu_j(dx).
\]
Even though it is hard to compute $r_i(\mu_j)$ explicitly, one
can see (\cite{E89}) that if $|\delta_1-\delta_2|$ is sufficiently small
then $r_1(\mu_2)>0$ and $r_2(\mu_1)>0$ and one gets coexistence
by the results from Section \ref{s:disc_2d}.
If however, $r_i(\mu_j)<0$ we get the extinction of species $i$.
Coexistence here depends directly on the fact that $\alpha_{ij}$ was assumed less than 1, while $\alpha_{jj}$ = 1.
This is coexistence that would occur in the absence of environmental fluctuations, and is just the classical outcome that
coexistence occurs if intraspecific competition, $\alpha_{jj}$, exceeds interspecific competition $\alpha_{ij}$. The presence of $\delta_i$ in this model means that these $\alpha$ coefficients are not exactly the classical ones of \cite{C00a}, but when the $\delta_i$ are equal they become equivalent to them. As mentioned above, coexistence by the storage effect cannot occur in this model because covED is zero even though $\gamma_i < 0$. Environmental fluctuations may still have a significant role in coexistence, however, in cases in which the $\delta_i$ values differ greatly between species. Then, the coexistence mechanism relatively nonlinearity can occur \citep{KC08}.

\subsubsection{Annual plants with seed banks and predation}
Here instead of perennial organisms, we consider annual plant species. This case has been very important in empirical studies of
coexistence in a variable environment.  For these species, the growing plant survives for less than a year.  It flowers and produces seed once, at the end of its
life. The seeds can be eaten by animals, for example rodents and ants (``seed predators''), and any seeds that escape become mixed in the soil as the ``seed bank.'' Seeds in the seed bank generally have environmentally-dependent germination. This means that in any year, only a fraction of the seeds in the seed bank germinate, and that fraction
depends on the specific environmental conditions of that year. A fraction of the seeds that do not germinate generally survive to the next year. The state variable in annual plant models, $N_i(t)$ for species $i$, is generally the number of seeds in the seed bank at the end of the year after all the adults have died, and seed predation has already occurred. The density $P(t)$ of seed predators is an auxiliary variable. The dynamical equations can now be given as
\[
\begin{aligned}
N_j(t+1)=&N_j(t)\left(s_j(1-E_j(t))+E_j(t)Y_j(t)e^{-A_j(t)-C(t)}\right)\\
P(t+1)=&\sum_{j=1}^n E_j(t)Y_j(t)N_j(t)e^{-C(t)}(1-e^{-A_j(t)})+s_pP(t)\\
=&P(t)\left(s_p+\sum_{j=1}^n E_j(t)Y_j(t)N_j(t)e^{-C(t)}\frac{1-e^{-A_j(t)}}{P(t)}\right)
\end{aligned}
\]
with
\[
C(t)=\sum_{k=1}^nG_j(t)N_j(t).
\]
and
$$A_j(t)=a_j(\mathbf{N}(t),\mathbf{E}(t))P(t).$$
Here $E_j(t)\in (0,1)$ is the germination fraction, $A_j(t)$ is mortality due to predation and $Y_j(t)$, which is maximum seed yield, is expected to vary from year to year
depending on how much rain and nutrients there are. Note that here $Y_j(t)$ is not an auxiliary variable, but an environmental variable. Instead, the predator, $P(t)$, is an auxiliary variable. More details about these models can be found in the work by \cite{KC09, KC10, CK10}.
Most important, by defining $D_j(t) = C(t) + A_j(t)$, this model fits the standard $G(E,D)$ form \eqref{e:ED}. The critical quantity $\gamma_j$ is negative, and covED is present. This means that coexistence by the storage effect can occur in this model. Here, there is a potential for two forms of it, a storage effect due to competition and a storage effect due to predation.

We assume that
$(E_j(t),Y_j(t))_{j=1,\dots,n; t\in \N}$ is a sequence of i.i.d random variables and $\E [(Y_1(1))^2]<\infty$ and $a_j$ is bounded below by a positive nonrandom constant. Note that $\frac{1-e^{-A_j(t)}}{P(t)}$ is continuous even at $P(t)=0$.
Due to the fact that
$E_j(t)N_j(t)e^{-C(t)}$ and $\frac{1-e^{-A_j(t)}}{P(t)}$ are bounded above by a constant as long as $a_j$ is bounded below,
we can follow arguments in Lemma \ref{l_per} to show that
our assumptions hold.

\subsubsection{Additive models: Unbounded offspring distribution}
Theoretical ecologists often use models of the form
\begin{equation}\label{e:log_normal}
N_i(t+1)=N_i(t)e^{E_i(t) - f_i(\BN(t))}
\end{equation}
where $f_i$ is a positive function representing the survivorship
and $\E_1(t),\dots,\E_n(t)$ are continuous random variables, and generally unbounded. For example, they could be normally distributed with
 the nonzero variance representing
the $\ln$ of mean number of offspring produced by the individuals of each
species at time $t$. We can identify $f_i(\BN(t))$ as $D_i(t)$ and we see that this model takes the $G(E,D)$ form of
the previous models, but because $\ln G_i(t) = E_i(t) - D_i(t)$, $\gamma_i = 0$, and the storage effect cannot occur.

For specific forms of $f_i$ we can show
explicitly that our assumptions hold. Moreover, if
$(E_1(t),\dots,E_n(t))_{t\in\N}$ is i.i.d one can also show
for specific $f_i$'s that the process is irreducible (see
\cite{E89} in the case when $f_i$ defines a Ricker model). Note
that when the support of $E_i$ is $\R$ the dynamics from
\eqref{e:log_normal} will not live in a compact state space so
one cannot use the results from \cite{BS19}.

\subsubsection{Additive models: Discrete Lotka-Volterra}

Consider $n$ species interacting according to
\begin{equation}\label{e:LVd}
X_i(t+1)=X_i(t)\exp\left(b_i(t)+\sum_j
a_{ij}(t)X_j(t)\right), i=1,2,\dots,n.
\end{equation}
These equations have been introduced by \cite{HHJ87} as discrete
time analogues of Lotka--Volterra differential equations.This form of discrete
Lotka-Volterra is what we have called Ricker above.  To draw the relationships
with these models, we can equate $E_i(t)$ with $b_i(t)$
and $D_i(t)$ with $-\sum_j a_{ij}(t)X_j(t)$. Then it is seen that $\gamma_i = 0$ and the storage effect cannot occur.
Moreover, as $D_i(t)$ is linear, relative nonlinearity \citep{C94} cannot occur either.  In fact, the only role that the variable environment has is to allow the model
to explore all of the state space, although this is an important role that leads to more complete results than in the deterministic case.

These equations
have further been analyzed by \cite{SBA11, BS19} in a stochastic
setting. Suppose the system is hierarchically ordered, that is,
there exists a permutation of the indices such that $a_t^{ii}<0$
for all $i, t\in \Z_+$ and $a^{ij}(t)\leq 0$ for all $i\leq j$
and $t\in \Z_+$. By \cite{HHJ87, BS19}, if the system is
hierarchically ordered, the coefficients $a_{ii}(t)$ are bounded
above by some negative number, and the coefficients $a_{ij}(t),
b_i(t)$ are bounded, one can show using a comparison argument
that there exists $K>0$ such that $\BX(t)$ enters $[0, K]^n$ and
never leaves it. We can then work under the assumption of a
compact state space.
\begin{rmk}
If the state space of $\BX(t)$ is not compact, the system
\eqref{e:LVd} can experience abrupt fluctuations from high
densities to very small densities. These sudden crashes make it
impossible to use the persistence and extinction criteria we
have developed - see Section 5 in the work by \cite{C82} for a
simple example of why violent population declines cannot be
allowed.
\end{rmk}
We can make use of the linearity of the system \eqref{e:LVd} and
Proposition \ref{p:rate} to compute $r_i(\mu)$ for any ergodic
measure $\mu$. In order to showcase the results, assume there
are only two species, so that
\begin{equation}\label{e:2d_disc}
\begin{aligned}
X_1(t+1)&=X_1(t)e^{b_1(t+1)+
a_{11}(t+1)X_1(t)+a_{12}(t+1)X_2(t)} ,\\
X_2(t+1)&=X_2(t)e^{b_2(t+1)+
a_{22}(t+1)X_2(t)+a_{21}(t+1)X_1(t)}.
\end{aligned}
\end{equation}
We first look at the Dirac delta measure $\delta_0$ at the
origin $(0,0)$
\[
r_i(\delta_0) = \E \ln b_i(1), i=1,2.
\]
If $r_i(\delta_0)>0$ then species $i$ survives on its own and
converges to a unique invariant probability measure $\mu_i$
supported on $\Se_+^i := \{\bx\in\Se~|~x_i\neq 0, x_j=0, i\neq
j\}$. Moreover,
\[
r_i(\mu_i) = \E [b_i(1)] + \E [a_{ii}(1)] \int
x_i\,\mu_i(dx_i)=0
\]
which implies
\[
\int x_i\,\mu_i(dx_i) = \frac{\E [b_i(1)]}{ \E [-a_{ii}(1)]}.
\]
One can use this to compute the per-capita growth rates
\begin{equation}\label{e:growth_2d}
r_i(\mu_j)= \E [b_i(1)] + \E [a_{ij}(1)] \int x_j\,\mu_j(dx_j) =
\E [b_i(1)] + \E [a_{ij}(1)] \frac{\E [b_j(1)]}{ \E
[-a_{jj}(1)]}.
\end{equation}

To ensure that the boundary of the state space is accessible,
one can assume for example that the $a_{ij}$s and the $b_i$'s
are absolutely continuous with respect to Lebesgue measure and
if $b_i>0$ then an interval of the form $(0,L)$ lies in its
support. Having the expressions \eqref{e:growth_2d} for
$r_1(\mu_2)$ and $r_2(\mu_1)$ we can make use of the discussion
from Section \ref{s:disc_2d} to classify the dynamics. See \cite{H19} for a more complete discussion of 2d Ricker models.

We note that these results are in a sense more complete than
what is known in the deterministic setting for discrete
Lotka-Volterra systems where the classification of the long term behavior is not fully known \citep{RS15, RS16, G19}.

\section{Applications in continuous time}\label{s:cont}
\subsection{Structured populations}
 The survival of an organism is influenced by both biotic
(competition for resources, predator-prey interactions) and
abiotic (light, precipitation, availability of resources)
factors. Since these factors are space-time dependent, all types
of organisms have to choose their dispersal strategies: If they
disperse they can arrive in locations with different
environmental conditions while if they do not disperse they face
the temporal fluctuations of the local environmental conditions.
The dispersion strategy impacts key attributes of a population
including its spatial distribution and temporal fluctuations in
its abundance. Continuous-space discrete-time population models
that disperse and experience uncorrelated, environmental
stochasticity have been studied by \cite{HTW88, HTW88b, HTW90}.
They show that the leading Lyapunov exponent $ r$ of the
linearization of the system around the extinction state usually
determines the persistence and extinction of the population.
\cite{ERSS13} studied a linear stochastic model that describes
the dynamics of populations that continuously experience
uncertainty in time and space. In \cite{HNY16} the authors
generalized \cite{ERSS13} to a
density-dependent model of stochastic population growth
that captures the interactions between dispersal and
environmental heterogeneity. We will showcase how one can
recover and extend the results from \cite{HNY16}.

Suppose we have a population with overlapping generations, which
live in a spatio-temporally heterogeneous environment
consisting of $n$ distinct patches. The growth rate of each
patch is determined by both deterministic and stochastic
environmental inputs. We denote by $ X_i(t)$ the
population abundance at time $t\geq 0$ of the $i$th patch and
write $ \BX(t)=(X_1(t),\dots,X_n(t))$ for the vector of
population abundances.

Consider the system
\begin{equation}\label{e4.0}
d X_i(t)=\left(X_i(t)\left(a_i-b_i(X_i(t))\right)+\sum_{j=1}^n
D_{ji}X_j(t)\right)dt+X_i(t)dE_i(t), \, i=1,\dots,n,
\end{equation}
where $D_{ij}\geq0$ for $j\ne i$ is the per-capita rate at which
the population in patch $i$ disperses to patch $j$,
$D_{ii}=-\sum_{j\ne i} D_{ij}$ is the total per-capita
immigration rate out of patch $i$, $\BE(t)=(E_1(t),\dots,
E_n(t))^T=\Gamma^\top\BB(t)$, $\Gamma$ is a $n\times n$ matrix
such that
$\Gamma^\top\Gamma=\Sigma=(\sigma_{ij})_{n\times n}$
and $\BB(t)=(B_1(t),\dots, B_n(t))$ is a vector of independent
standard Brownian motions adapted to the filtration
$\{\F_t\}_{t\geq 0}$.
We make the following assumptions.

\begin{asp}\label{a:competition}
For each $i=1,\dots,n$ the function $b_i:\R_+\mapsto\R$ is
locally Lipschitz and vanishing at $0$. Furthermore, there are
$M_b>0$, $\gamma_b>0$ such that
\begin{equation}\label{e:b}
\dfrac{\sum_{i=1}^n x_i(b_i(x_i)-a_i)}{\sum_{i=1}^n
x_i}>\gamma_b\text{ for any } x_i\geq0, i=1,\dots,n \text{
satisfying } \sum_{i=1}^n x_i\geq M_b
\end{equation}
\end{asp}
\begin{asp}\label{a:dispersion}
The dispersal matrix $D$ is \textit{irreducible}.
\end{asp}
\begin{asp}\label{a:nonsingular}
The covariance matrix $\Sigma$ is non-singular.
\end{asp}
\begin{rmk}
Condition \eqref{e:b} is biologically reasonable because it holds if the $b_i$'s are sufficiently
large for large $x_i$'s. Below, we give some simple scenarios under which Assumption \ref{a:competition} is satisfied.
\begin{itemize}
\item[a)] Suppose $b_i:[0,\infty)\to [0,\infty), i=1,\dots, n$ are locally Lipschitz and vanishing at $0$. Assume that there exist $\gamma_b>0, \tilde M_b>0$ such that
\[
\inf_{x\in  [\tilde M_b,\infty)} b_i(x) - a_i-\gamma_b>0,~ i =1,\dots,n
\]
Then  Assumption \ref{a:competition} holds (see \cite{HNY16}).
\item[b)] We note that (a) is satisfied if for $i=1,\dots,n$ the function $b_i:\R_+\mapsto \R$ is locally Lipschitz, vanishing at $0$
and satisfies $\lim_{x\to \infty} b_i(x)=\infty$.

\item[c)] One natural choice for the competition functions, which is widely used throughout the literature,
is $b_i(x)=\kappa_i x, x\in (0,\infty)$ for some $\kappa_i> 0$. In this case the competition terms become $-x_ib(x_i) = - \kappa_i x_i^2$. It is easy to see that these functions satisfy (b) above.
\end{itemize}
\end{rmk}

Assumption \ref{a:dispersion} is equivalent to forcing the
entries of the matrix $P_t=\exp(tD)$ to be strictly positive for
all $t>0$. This means that it is possible for the population to
disperse between any two patches. Assumption \ref{a:nonsingular}
says that our randomness is non-degenerate, and thus truly
$n$-dimensional.

We define the total abundance of our population at time $t\geq
0$ via
$X(t):=\sum_{i=1}^n  X_i(t)$ and
let $Y_i(t):=\frac{ X_i(t)}{X(t)}$ be the proportion of
the total population that is in patch $i$ at time $t\geq0$. Set
$\BY(t)=(Y_1(t),\dots, Y_n(t))$.
   An application of It\^o's lemma to \eqref{e4.0} yields
\begin{equation}\label{e4.1}
\begin{split}
dY_i(t)=&Y_i(t)\left(a_i-\sum_{j=1}^na_jY_j(t)-b_i(X(t)Y_i(t))+\sum_{j=1}^nY_j(t)b_j(X(t)Y_j(t))\right)dt+\sum_{j=1}^{n}D_{ji}Y_j(t)dt\\
&+Y_i(t)\left(\sum_{j,k=1}^n\sigma_{kj}Y_k(t)Y_j(t))-\sum_{j=1}^n\sigma_{ij}Y_j(t)\right)dt
+Y_i(t)\left[dE_i(t)-\sum_{j=1}^n Y_j(t)dE_j(t)\right]\\
dX(t)=&X(t)\left(\sum_{i=1}^n(a_iY_i(t)-Y_i(t)b_i(X(t)Y_i(t)))\right)dt+X(t)\sum_{i=1}^nY_i(t)dE_i(t)
\end{split}
\end{equation}
We can rewrite \eqref{e4.1} in the following compact equation
for $(\BY(t), X(t))$ where $\bb(\bx)=(b_1(x_1),\dots,
b_n(x_n))$.
\begin{equation}\label{eq.bys}
\begin{split}
d\BY(t)=&\left(\diag(\BY(t))-\BY(t)\BY^\top(t)\right)\Gamma^\top
d\BB(t)\\
&+\BD^\top\BY(t)dt+\left(\diag(\BY(t))-\BY(t)\BY^\top(t)\right)(\ba-\Sigma
\BY(t)-\bb(X(t)\BY(t)))dt\\
dX(t)=&X(t)\left[\ba-{\bb(X(t)\BY(t))}\right]^\top\BY(t)dt+X(t){\BY(t)}^\top
\Gamma^\top d\BB(t),
\end{split}
\end{equation}
where $\BY(t)$ lies in the simplex
$\Delta:=\{(y_1,\dots,y_n)\in\R^{n}_+: y_1+\dots+y_n=1\}$.
Let $\Delta^{\circ}=\{(y_1,\dots,y_n)\in\R^{n,\circ}_+:
y_1+\dots+y_n=1\}$ be the interior of $\Delta$.

Consider equation \eqref{eq.bys} on the boundary $((\by,x):
\by\in\Delta, x=0)$ (that is, we set $X(t)\equiv 0$ in the
equation for $\BY(t)$). We have
the following system
\begin{equation}\label{eq.by}
\begin{split}
d\tilde\BY(t)=&\left(\diag(\tilde\BY(t))-\tilde\BY(t)\tilde\BY^\top(t)\right)\Gamma^\top
d\BB(t)\\
&+\BD^\top\tilde\BY(t)dt+\left(\diag(\tilde\BY(t))-\tilde\BY(t)\tilde\BY^\top(t)\right)(\ba-\Sigma
\tilde\BY(t))dt
\end{split}
\end{equation}
on the simplex $\Delta$.

\cite{ERSS13} proved that the process $(\tilde\BY(t))_{t\geq 0}$
is an irreducible Markov process, which has the strong Feller
property and admits a unique invariant probability measure
$\nu^*$ on $\Delta$. Let
\begin{equation}\label{lambda}
r_X(\nu^*)=\int_{\Delta}\left(\ba^\top{\bf y}-\frac12{\bf
y}^\top\Sigma{\bf y}\right)\nu^*(d{\bf y}).
\end{equation}

\begin{thm}\label{t:struc_cont}
The following hold:
\begin{itemize}
  \item Suppose that $ r_X(\nu^*)>0$.
The process $ \BX(t) = ( X_1(t),\dots,
X_n(t))_{t\geq 0}$ has a unique invariant probability measure
$\pi$ on $\R^{n,\circ}_+$ that is absolutely continuous with
respect to the Lebesgue measure and
\begin{equation}
\lim\limits_{t\to\infty} A^{t}\|P_\BX(t, \mathbf{x},
\cdot)-\pi(\cdot)\|_{\text{TV}}=0,
\;\mathbf{x}\in\R^{n,\circ}_+,
\end{equation}
for some constant $A>0$. Here $P_\BX(t,\mathbf{x},\cdot)$ is the
transition probability of $( \BX(t))_{t\geq 0}$.
  \item Suppose that $ r_X(\nu^*)<0$.
For any $i=1,\dots,n$ and any $\mathbf{x} = (x_1,\dots,x_n)\in
\R_+^{n}$,
\begin{equation}
\PP_\bx\left\{\lim_{t\to\infty}\frac{\ln { X}_i(t)}{t}=
r_X(\nu^*)\right\}=1.
\end{equation}
\end{itemize}
\end{thm}

\subsection{Continuous-time models with a resource variable}
Suppose we have a guild of $n$ species $X_1,\dots,X_n$ whose dynamics is given by
\begin{equation}\label{e1.5.2}
dX_i(t) = X_i(t)( c_{i}(t)R(\BX(t))-m_i) + X_i(t)\sigma_idB_i(t)
\end{equation}
Here $R(X_1(t),\dots,X_n(t))$ is the resource abundance and $c_i(t)$ is the resource uptake rate. The resource abundance is assumed to be given by the algebraic equation
\begin{equation}\label{e2.5.2}
R(\bx) =  R_{\max} - \sum_{j=1}^n c_j(t)x_j(t), \bx\in\R_+^n.
\end{equation}
This model is a special case of the resource competition model with fast resource dynamics of \cite{LC16}. The $c_j(t)$ are environmentally varying resource uptake rates and hence positively affect the growth rates of the species while negatively affecting resource growth and abundance. Because resource uptake is assumed fast, \eqref{e2.5.2} represents the equilibrium of the resource with the consumer densities at time $t$.

One way of modelling the $c_j(t)$ is by using an Ornstein-Uhlenbeck process
\begin{equation}\label{e3.5.2}
d\BU(t) = -\gamma(\BU(t)-\bal)dt + \BA d\BW(t)
\end{equation}
where $\bal$ is the mean of $\BU(t)$ and $\BA$ is a constant matrix. It is well-known that the process $\BU(t)$ converges as $t\to\infty$ to a stationary distribution $\mu_\BX$ that is normal and has mean $\bal$ and covariance matrix $M:=\BA \BA^T/(2\gamma)$. Define
\begin{equation}\label{e4.5.2}
c_i(t) =  \phi_i(U_i(t)):=\frac{A_ie^{U_i(t)}}{1 + B_ie^{U_i(t)}} + C_i , t\geq 0
\end{equation}
where $A_i, C_i, B_i>0$
then we have $0<C_i\leq c_i(t)\leq \frac{A_i}{B_i} + C_i$.

\begin{pron}
Suppose $(\BX(t),\BU(t))$ is given by \eqref{e1.5.2}, \eqref{e2.5.2}, \eqref{e3.5.2} and \eqref{e4.5.2}. Then, setting  $V(\bx,\bu):=1+|\bx|+|\bu|^2$, Assumptions \ref{a:sde} and \ref{a:e_cont} will hold for $(\BX(t),\BU(t))$.
\end{pron}
\begin{proof}
We can check that there exist constants $a_1, a_2>0$ such that
\begin{equation}\label{e5.5.2}
\op |\bu|^2\leq a_1-a_2 \|\bu\|^2 \text{ for any } \bu\in\R^n.
\end{equation}
Since $C_i\leq c_i(t)\leq \frac{A_i}{B_i}+C_i$,
we can see that
\begin{equation}\label{e6.5.2}
\op \sum_{i} x_i\leq  R_{\max}\max_i\left\{\frac{A_i}{B_i}+C_i\right\} \sum_{i} x_i -\min_i\{C_i^2\}\sum_{i} x_i^2
\end{equation}
In view of \eqref{e5.5.2} and \eqref{e6.5.2} and using the Lyapunov function $V(\bx,\bu):=1+|\bx|+|\bu|^2$ we can see that Assumptions \ref{a:sde} and \ref{a:e_cont} hold for the joint process $(\BX(t), \BU(t))$.
 \end{proof}
Related ecological systems have been
studied by \cite{AM80, LC16}.  Most important, this model provides a continuous-time analogue of the models in G(E,D) form in discrete-time examples. In fact, the equations can be rewritten as
$$dX_i(t) = X_i(t)( c_{j}(t)(R_{\max} - m_i) - c_i(t)D_i(t))dt + X_i(t)\sigma_idB_i(t)$$ or $$dX_i(t) = X_i(t)g_i(E_i(t), D_i(t))dt + X_i(t)\sigma_idB_i(t)$$
where, $D_i(t)= \sum_{j=1}^n c_j(t)x_j(t)$, we identify $c_i(t)$ as $E_i(t)$, and $g_i(E_i(t), D_i(t))$ with $\ln G_i(E_i(t), D_i(t))$. We see immediately that $\gamma_i$ is negative, and covED is present. Of most interest in this case, as shown by \cite{LC16}, the $r_i(\mu)$ can be developed explicitly in terms of $\gamma_i$ and covED giving a very clear example mathematically of how the storage effect operates.  It is the sole mechanism of coexistence in this model.

\subsection{One species and one auxiliary variable}

\begin{equation}\label{e:sde_1d}
\begin{aligned}
dX(t)&=X(t) f(X(t), Y(t))\,dt+X(t)g(X(t), Y(t))\,dB(t),\\
dY(t) &= u(Y(t))\,dt+ h(Y(t))\,dW(t).
\end{aligned}
\end{equation}
Suppose we know that $Y(t)$ has a unique invariant probability
measure $\mu_Y$ on $\R$ and Assumption \ref{a:sde} is satisfied (see Remark \ref{r:sde_assum} and Assumption \ref{a.nonde}). Then the long term behaviour is
determined by the expected per capita growth rate

\[
r_X(\delta_0\times \mu_Y) = \int_\R
\left(f(0,y)-\frac{g^2(0,y)}{2}\right)\,\mu_Y(dy).
\]
If $r_X(\delta_0\times \mu_Y)>0$ we have persistence while if
$r_X(\delta_0\times \mu_Y)<0$ then $X$ goes extinct almost surely and
\[
\lim_{t\to\infty}\frac{\ln X(t)}{t}= r_X(\delta_0\times \mu_Y).
\]
\subsection{Stochastic replicator dynamics}
An important class of continuous time dynamics is the one called
by \cite{FH92} \textit{stochastic replicator dynamics}. We set $\Delta:=\{\bx\in
\R^n_+~|~\sum_i x_i=1\}$ and let $\Delta_0 = \{\bx\in\Delta~:~x_j=0~\text{for some}~j\}$ be the extinction set, The
fitness of population $i$ is described by a function
$f_i:\Delta\to\R$ and the number of individuals in population
$i$ is given by
\begin{equation}\label{e:rep1}
dU_i(t)=U_i(t)(f_i(\BX_t)+\sigma_idB_i(t)),
\end{equation}
where $$X_i(t)=\frac{U_i(t)}{\sum_{j}U_j(t)}.$$
The stochastic replicator dynamics has been studied recently by various authors \citep{I05,BHS08, HI09, SBA11}.
\cite{GH03} studied the permanence and impermanence of deterministic replicator dynamics. \cite{BHS08} looked at these conditions for stochastic replicator dynamics. When the deterministic replicator dynamics is permanent and the noise level small, they showed that the stochastic dynamics admits a unique invariant probability measure whose mass is concentrated near the maximal interior attractor of the unperturbed system. When the deterministic dynamics is impermanent and the noise level small or large, \cite{BHS08} showed that the stochastic dynamics converges to the boundary of the state space at an exponential rate. In \cite{SBA11} the authors were able to give conditions for persistence and convergence to a stationary distribution. Here we recover and strengthen the results of \cite{SBA11} by proving stronger persistence as well as extinction results.

Applying Ito's formula to \eqref{e:rep1} one can see (\cite{SBA11}) that
\[
dX_i(t)=X_i(t)
F_i(\BX(t))\,dt+X_i(t)\sum_jg_{ij}(\BX(t))\,dB_i(t),
~i=1,\dots,n
\]
where
\[
F_i(\bx) = f_i(\bx)-\sigma^2_{ii}x_i - \sum_j
x_j(f_j(\bx)-\sigma^2_{jj}x_j)
\]
and
\[
g_{ij}(\bx)= (\delta_{ij}-x_j)\sigma_j.
\]
It is immediate that for two types the only possible ergodic
measures on $\Delta_0$ are the Dirac masses at $(1,0)$ and
$(0,1)$. This then yields
\[
r_1(\delta_{(0,1)}) = f_1(0,1)-f_2(0,1) -
\frac{1}{2}(\sigma_1^2-\sigma_2^2)
\]
and
\[
r_2(\delta_{(1,0)}) = f_2(1,0)-f_1(1,0) -
\frac{1}{2}(\sigma_1^2-\sigma_2^2).
\]
We next follow the example studied by \cite{SBA11} and show how
we can now give a complete description of the dynamics. Suppose
there are three interacting species and $f_1(\bx)=\mu_1+bx_3,
f_2(\bx)=\mu_2$ and $f_3(\bx)=\mu_3-cx_1$. Interactions between
types $1$ and $3$ provide a benefit $b>0$ to type $1$ and a cost
$c>0$ to type $3$. Assume the following inequality holds
\begin{equation}\label{e:replicator_ineq}
\mu_3-\frac{\sigma_3^2}{2}>
\mu_2-\frac{\sigma_2^2}{2}>\mu_1-\frac{\sigma_1^2}{2}>0.
\end{equation}

By \cite{SBA11} it is easy to see that if $Y_1(0)=0, Y_3(0)>0$
then $\BX(t)$ converges almost surely to $(0,0,1)$ since
$Y_i(t)=Y_i(0)\exp\left(\left(\mu_i-\frac{\sigma_i^2}{2}\right)t
+\sigma_iB_i(t)\right)$ and \eqref{e:replicator_ineq} holds.
Similarly, if $Y_3(0)=0, Y_1(0)>0$ then $\BX(t)$ converges
almost surely to $(0,1,0)$.

One can compute the expected per capita growth rates when $\mu$
is the Dirac mass function at $(0,1,0), (1,0,0)$, and $(0,0,1)$:
\begin{equation}\label{e:replicator}
\begin{split}
r_2(1,0,0)&= \mu_2-\frac{\sigma_2^2}{2}-\mu_1+\frac{\sigma_1^2}{2}>0\\
r_3(1,0,0)&= \mu_3-c-\frac{\sigma_3^2}{2}-\mu_2+\frac{\sigma_2^2}{2}\\
r_1(0,1,0)&= \mu_1-\frac{\sigma_1^2}{2}-\mu_2+\frac{\sigma_2^2}{2}<0\\
r_3(0,1,0)&= \mu_3-\frac{\sigma_3^2}{2}-\mu_2+\frac{\sigma_2^2}{2}>0\\
r_1(0,0,1)&= \mu_1+b-\frac{\sigma_1^2}{2}-\mu_3+\frac{\sigma_3^2}{2}\\
r_2(0,0,1)&= \mu_2-\frac{\sigma_2^2}{2}-\mu_3+\frac{\sigma_3^2}{2}<0.
\end{split}
\end{equation}
Note the above imply by Theorem \ref{t:ex3_cont} or Theorem \ref{t:ex2cont} that it is not possible to have ergodic probability measures on $\Delta_{12}:=\{\bx\in\Delta~|~x_1>0, x_2>0, x_3=0\}$ or on $\Delta_{23}:=\{\bx\in\Delta~|~x_1=0, x_2>0, x_3>0\}$. We have the following possibilities.

\textbf{1)}
Assume
\begin{equation}\label{e:ineq}
b>\mu_3-\frac{\sigma_3^2}{2}-\mu_1+\frac{\sigma_1^2}{2}>c.
\end{equation}
Then $r_1(0,0,1)>0$ and $r_3(1,0,0)>0$. This implies that there
exists a unique invariant probability measure $\mu_{13}$ on
$\Delta_{13}:=\{\bx\in\Delta~|~x_1>0, x_2=0, x_3>0\}$ and by \cite{SBA11} that
\[
r_2(\mu_{13})=\frac{b\sigma_3^2-(b-c)\sigma_2^2-2br_3+2(b-c)r_2+2cr_1+2bc}{2(b-c)}.
\]
If $r_2(\mu_{13})>0$ the populations coexist and $\BX(t)$ converges to its
unique stationary distribution $\pi$ on $\Delta_+:=\Delta\setminus
\Delta_0$. See Figure \ref{fig:figure3} for some intuition of where the different ergodic measures live on the simplex $\Delta$.

\begin{figure}
\begin{center}
\begin{tikzpicture}[scale=0.8, line width = 1.2 pt]

\draw[->] (0,-0.5) -- (0,5);
\draw (0,5.2) node {$x_3$};
\draw[->](-0.5,0) -- (6,0);
\draw (6.3,0) node {$x_1$};
\draw[->](0.2,0.25) -- (-3,-3.75);
\draw (-3.1,-4.05) node {$x_2$};

\fill[fill=red!20!white, opacity=0.4] (0,3) -- (4,0) -- (-2,-2.5) -- cycle;
\draw (0.6,0.8) node {$\pi$};

\draw[line width=2pt,color=red!70!black] (0,3) -- (4,0);
\draw (2,2) node {$\mu_{13}$};
\filldraw[red!70!black] (0,3) circle (1.5pt);
\filldraw[red!70!black] (4,0) circle (1.5pt);
\filldraw[red!70!black] (-2,-2.5) circle (1.5pt);
\draw (0.8,3.2) node {$(0,0,1)$};
\draw (4.6,0.3) node {$(1,0,0)$};
\draw (-2.8,-2.3) node {$(0,1,0)$};

\end{tikzpicture}

\caption{Coexistence for a three species stochastic replicator system. The red measures signify repellers and $\pi$ is the invariant probability measure the system converges to.}
\label{fig:figure3}

\end{center}
\end{figure}
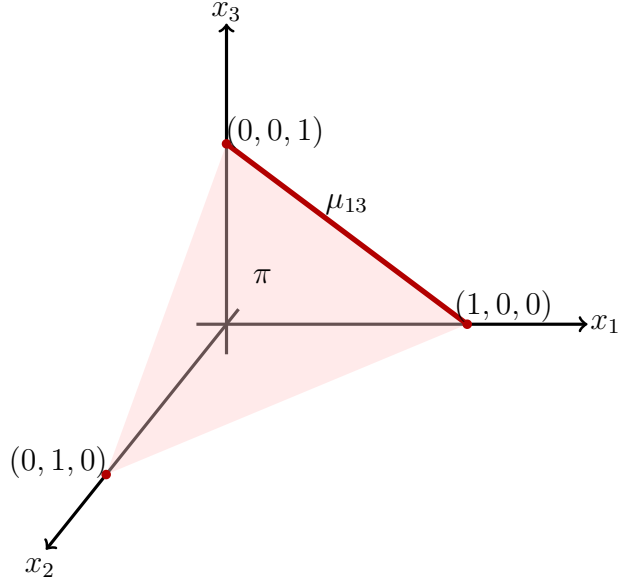

If $r_2(\mu_{13})<0$ we have
   \[
\PP_\bx\left\{\U(\omega)=\{\mu_{13}\}
~\text{and}~\lim_{t\to\infty}\frac{\ln X_2(t)}{t}=r_2(\mu_{13})
\right\}=1
    \]
for all $\BX(0)=\bx\in\Delta_+$.

\textbf{2)} Assume \eqref{e:ineq} does not hold and
\[
b< \mu_3-\frac{\sigma_3^2}{2}-\mu_2+\frac{\sigma_2^2}{2}.
\]
Then $r_1(0,0,1)<0$ and $r_2(0,0,1)<0$ and $\delta_{(0,0,1)}$ is
a transversal attractor. As a result
\[
\PP_\bx\left\{\U(\omega)=\{\delta_{(0,0,1)}\}
~\text{and}~\lim_{t\to\infty}\frac{\ln X_j(t)}{t}=r_j((0,0,1)),
j=1,2 \right\}=1
    \]
for all $\BX(0)=\bx\in\Delta_+$.

\textbf{3)} Assume \eqref{e:ineq} does not hold and
\[
b,c> \mu_3-\frac{\sigma_3^2}{2}-\mu_2+\frac{\sigma_2^2}{2}.
\]
The only ergodic invariant measures are the Dirac measures at
the vertices.

If $b>c$ we can show that \citep{SBA11}
\[
\max_i r_i(\mu)>0, \mu \in \Conv \{\delta_{(1,0,0)}, \delta_{(0,1,0)}, \delta_{(0,0,1)}\}
\]
we have coexistence and $\BX(t)$ converges to a unique invariant probability measure $\pi$ on $\Delta_+$.

If $b<c$ persistence is not possible. However, we are not in the setting of any of our extinction theorems because none of the ergodic measures is an attractor. Different methods have to be used - we show in \cite{HNS20} that
the process converges to the extinction set $\Delta_0$
exponentially fast.

\subsection{Lotka Volterra: Two prey and one predator}
Assume we have two prey and one predator interacting according
to the following Lotka--Volterra system
\begin{equation}\label{e1-ex2}
\begin{cases}
dX_1(t)=X_1(t)[a_1-a_{11}X_1(t)-a_{13}X_3(t)]dt+X_1(t)dE_1(t)\\
dX_2(t)=X_2(t)[a_2-a_{22}X_2(t)-a_{23}X_3]dt+X_2(t)dE_2(t)\\
dX_3(t)=X_3(t)[-a_3-a_{33}X_3(t)+a_{13}X_1(t)+a_{23}X_2(t)]dt+X_3(t)dE_3(t).\\
\end{cases}
\end{equation}
Assume that $a_i, b_i, c_{ii}>0, i=1,2,3$, $c_{23}, c_{31},
c_{12}, c_{21}, c_{31}, c_{32}\geq0.$ Assume that
$r_1(\delta_0)=a_1-\frac{\sigma_{11}}{2}>0$ and
$r_2(\delta_0)=a_2-\frac{\sigma_{22}}{2}>0$ so that species $1$
and $2$ can survive on their own. This implies by Theorem
\ref{t:pers2_cont} that there exist unique ergodic probability
measures $\mu_1$ and $\mu_2$ on
$\R_{1+}:=(0,\infty)\times\{0\}\times\{0\}$ and
$\R_{2+}:=\{0\}\times(0,\infty)\times\{0\}$.

Note that $r_1(\delta_0)=-a_3-\frac{\sigma_{33}}{2}<0$ which
implies that, as expected, the predator cannot survive without
any of the prey species.

Next, we can also see that for $i=1,2$
\[
0=r_i(\mu_i)=a_i-
\frac{\sigma_{ii}}{2}-b_i\int_{\R_{i+}}x_i\mu(dx_i)
\]
which implies
\[
\int_{\R_{i+}}x_i\mu(dx_i) =
\frac{a_i-\frac{\sigma_{ii}}{2}}{a_{ii}}
\]
Using this we can compute for $j=1,2$
\[
\lambda_j(\mu_i)= \int_{\partial\R^3_+}\left(a_j+\sum_{\ell=1}^3
a_{j\ell } x_\ell-\dfrac{\sigma_{j j}}2\right)\mu_i(d\bx)=a_j+
a_{ji} \int_{\partial\R^\circ_{i+}}
x_i\mu_i(d\bx)-\dfrac{\sigma_{jj}}2 =
a_j-\dfrac{\sigma_{jj}}2>0.
\]
This implies by Theorem \ref{t:pers2_cont} that there exists a
unique invariant probability measure $\mu_{12}$ on
$\R_{12+}=(0,\infty)\times (0,\infty)\times \{0\}$.

The predator's expected per capita growth rates are
\[
r_3(\mu_i)=-a_3-\frac{\sigma_{33}}{2}
+a_{i3}\left(\frac{a_i-\frac{\sigma_{ii}}{2}}{a_{ii}}\right),
i=1,2.
\]

Suppose that $r_3(\mu_i)<0, i=1,2$. Then the faces $\R_{13+},
\R_{23+}$ do not support any invariant probability measures.
Lastly, in order to compute $r_3(\mu_{12})$ we note that
\[
r_{1}(\mu_{12})= r_{2}(\mu_{12})=0
\]
which yields
\[
r_3(\mu_{12})=-a_3-\frac{\sigma_{33}}{2} +
a_{13}\left(\frac{a_1-\frac{\sigma_{11}}{2}}{a_{11}}\right)+
a_{23}\left(\frac{a_2-\frac{\sigma_{22}}{2}}{a_{22}}\right).
\]

If $r_3(\mu_{12})>0$ we have the existence of a unique invariant
probability measure $\mu_{123}$ on $(0,\infty)^3$ and the
coexistence of all three species. If $r_3(\mu_{12})<0$ we have
the extinction of the predator and the coexistence of the two
prey species
\[
\PP_\bx\left\{\U(\omega)=\{\mu_{12}\}
~\text{and}~\lim_{t\to\infty}\frac{\ln X_3(t)}{t}=r_3(\mu_{12})
\right\}=1, \bx\in (0,\infty)^3.
\]
We note that Lotka--Volterra SDE models have been studied
extensively recently \cite{HN17, HN17b}. A full
classification of three-dimensional competitive and
predator-prey stochastic Lotka--Volterra systems
can be found in recent work by \cite{HNS20}.
In the model here, each of the prey species experiences intraspecific competition from the negative term $a_{ii}X_i(t)$ in its growth rate. They also experience predation through the term $a_{i3}X_3(t)$, but as the predator depends on these prey species, increasing in density when the prey are abundant, it is a mediator of apparent competition, i.e. the prey species interact with each other negatively through their effects on the predator. In relation to the previous models, $D_i(t) = a_{ii}X_i(t) + a_{i3}X_3(t)$ for prey species $i$. As the model is additive ($\gamma_i = 0$), the storage effect cannot occur and therefore cannot affect the coexistence of the prey species. Moreover, as the per capita growth rates are linear in the species densities, coexistence by relative nonlinearity \citep{C94, KC08} cannot occur either. Although the invasion rates, $r_i(\mu)$, do depend on variances, implying an effect of stochastic variation on persistence, these effects are more akin to a parameter change than to a strong effect of stochastic variation because substituting $a_i$ for $a_i - \sigma_{ii}/2$ in these formulae reproduces the deterministic invasion rates.
\section{Discussion}\label{s:discussion}
Theoretical models in ecology have long been dominated by deterministic models, not because such models were seen as superior or adequate, but because the theory for stochastic models seemed insuperably difficult. Early stochastic models in ecology focused on Markov jump processes \citep{K1948, B1960}, which provided some solutions and useful approximations, but such processes are only stochastic because they have discrete population sizes, and their stochastic elements correspond to independent uncertain events in the lives of individuals, such as mortality and birth. Laws of large numbers in such processes mean that for reasonably large population sizes a deterministic model will be adequate \citep{K1970}. Stochastic environmental variation, on the other hand, which is pervasive in nature and has strong effects on population dynamics, remains important no matter how large the population might be \citep{C78}. However, the dynamical equations in such situations are generally nonlinear and only in limited circumstances yield to analytical solution. Some early approaches used linear approximation \citep{MayRM1974}, and such linear approaches are still being used to approximate population fluctuations in some circumstances \citep{Ripa2003}. It was not until \cite{T78,T81} and \cite{TG80} that the essential nonlinearity of models for multiple interacting species was taken into account. However, these models were Lotka-Volterra models in discrete and continuous time, and were linear additive models, i.e. their per capita growth rates had environmental and density effects combining additively, precluding the storage effect as discussed above, and the density effects were linear precluding coexistence due to relative nonlinearity, as discussed above. In this sense they were disappointing because the stochastic environment added very little to the understanding already available from deterministic models. Soon after, however, stochastic environment models in which the environmental fluctuations have strong effects on species coexistence were developed \cite{CW81, C82, A84}. Most important, the strong effects of the stochastic environments in these models resulted directly from nonlinearities introduced by adding realistic biology to simpler models in which a stochastic environment had no such strong effects \citep{C94}.

These findings that stochastic environmental variation may allow robust species coexistence in situations in which coexistence is precluded in a deterministic environment might have been a great stimulus to the development of stochastic environment models in ecology.  However, the difficulty of their analysis meant that few theoretical ecologists followed this path. An important development was quadratic approximations to the invasion growth rates, $r(\mu)$, in these models \citep{C94, C19}, which led to an understanding of the storage effect coexistence mechanism and relative nonlinearity \citep{AM80, AH02}. However, a rigorous understanding of how these invasion rates are able to demonstrate stochastic persistence and coexistence beyond two-species settings was lacking until the developments of \cite{SBA11}. Since that time, this understanding has been extended, culminating most recently in \cite{BS19} in which auxiliary variables were introduced. As discussed here, auxiliary variables $Y$ allow a substantial increase in the richness of the models that can be analyzed by an understanding of the invasion growth rates of the focal interacting species, the $X$ and $N$ variables in our development here. As we have discussed above, they allow populations to have a great deal of structure, they allow the environmental fluctuations to be Markovian rather than merely i.i.d, and they allow dynamic mediators of the interactions between species to be considered. For example, here we consider predators to be an auxiliary variable when studying the coexistence of annual plants subject to seed predation, and in other models, e.g. \citep{LC16, C20} resources that the organisms compete for are auxiliary variables. Remarkably, despite addition of the auxiliary variables, stochastic persistence of the focal species is nevertheless still determined by their invasion growth rates alone.

The major contributions of this work are key extensions of the recent results of \cite{BS19}, and thorough connection of these mathematical developments with the ecological literature. Here we have removed the compactness requirement on the state space. Although it has been argued that population densities are necessarily bounded in nature, it is difficult to identify exactly what those bounds are, and many very useful population models do not put strict bounds on how large a population can be. For example, the simple logistic model in a variable environment leads to a gamma distribution for population size \citep{MayRM1974}. Probability distributions in which unrealistically high populations have extremely low probability make much more sense, and the results here accommodate this situation. Similarly past work in discrete time has placed strict limits on the magnitude of environmental variables, with the justification often being the finite capabilities of organisms. However, the central limit theorem, which leads to many valuable models in science, yields the unbounded normal distribution. Many variables in biology are found to be well approximated by the normal distribution on the log scale. Thus, a variable such as reproductive rate might be well approximated by a lognormal distribution, yet with its infinite tail violates the conditions for establishing persistence in the otherwise extremely valuable work of \cite{BS19}. We have found here, that rather than impose strict bounds on random variables, bounds on an expected value involving the multiplication rates and a Lyapunov-like function are sufficient. In addition, we extend previous work by characterizing the situations where some of the species become extinct, and the remainder converge to invariant measures on the boundary of the full state space. As a consequence we are able to gain a full picture of the long-term outcomes for these stochastic population models.

The second major contribution of this work is to explain how these mathematical developments connect with ecological theory showing where it validates some past work, and sets the stage for exciting new developments in the future. Although we do not have space to give the full details of how species coexistence results from environmental fluctuations in these models, we have sketched the major issues, and shown how the various ecological models fit within a family of models whose structure determines when the the storage effect coexistence mechanism will be active, promoting coexistence in a variable environment. We have also touched on the coexistence mechanism relative nonlinearity. Both of these mechanisms naturally arise from the mathematical structure of the model, which is a function of the biological details \citep{C94}. A concise development of formulae for the invasion rates, $r(\mu)$, is given in \cite{C19} exhibiting explicitly the involvement of the storage effect and relative nonlinearity in these rates. Thus, although in general we have not been able to give explicit formulae for the invasion rates in this article, that information is to be found in the ecological literature in general terms in  \citep{C94} and \cite{C19}, and in more specific terms for the individual models from the citations above where the models are discussed. It needs to be emphasized that it can often be difficult to understand how environmental variation will affect outcomes in ecological models, but an understanding of the structure of the model in the terms discussed here goes a long way to determining when strong effects of environmental variation on species coexistence will occur. A great many models in ecology simply take on linear additive forms, which miss key features of nature of special relevance to how environmental variation can promote species coexistence. Environmental variation can still have effects in such models by permitting the state variables to explore all of the state space, simplifying the results and thereby allowing a simpler and more complete description of the possible long-term outcomes of the dynamics as arise here for the Ricker and Lotka-Volterra models. Of most importance, however, the stage is set for a great deal of synergy between the ecological applications and the mathematical theory of stochastic persistence. Where the mathematical theory can develop rigorous stochastic persistence and extinction conditions, the ecological theory, exemplified by the discussion of the $G(E,D)$ form for the discrete time multiplication rates, leads to an understanding of where major effects of environmental variation on stochastic extinction and coexistence are to be expected.

\textbf{Acknowledgements:} The authors acknowledge generous
support from the NSF through the grants DMS-1853463 for
Alexandru Hening, DMS-1853467 for Dang Nguyen, and DEB-1353715 for Peter Chesson.
\bibliographystyle{agsm}
\bibliography{LV}
\appendix

\section{Persistence proofs}\label{s:a}

\begin{lm}\label{lmA.1}

For any $\bz\in\Se, t\in\N$ we have the bounds

$$\E_\bz(V(\BZ(t))\leq \rho^t
V(\bz)+C\sum_{s=0}^t\rho^s\leq\rho^t V(\bz)+\dfrac{C}{1-\rho},$$
and
$$
\begin{aligned}
\E_\bz h(\BZ(t), \xi(t))
\leq&
\E_\bz \left[V(\BZ(t)) h(\BZ(t), \xi(t))\right]
\leq& \rho \E_\bz V(\BZ(t))+C\\
\leq&\rho^{t+1} V(\bz)+\dfrac{C}{1-\rho}.
\end{aligned}
$$
If a function $\psi$ satisfies $$\lim_{\bz\to\infty}
\frac{\psi(\bz)}{\E V(\bz,\xi(t))h(\bz,\xi(t))}=0$$
then $\psi$ is $\mu$-integrable for any invariant probability
measure $\mu$.
Moreover,  if $\mu(\Se_+)=1$ then
$r_i(\mu)=0$ for any $i\in I$.
\end{lm}
\begin{proof}

By Assumption A3) we have $PV(\bz)\leq \rho V(\bz)+C$. Using
this and the Markov property yields
$\E_\bz(V(\BZ(t+1))\leq\rho \E_\bz V(\BZ(t))+C$. As a result
$$\E_\bz(V(\BZ(t))\leq \rho^t
V(\bz)+C\sum_{s=0}^t\rho^s\leq\rho^t V(\bz)+\dfrac{C}{1-\rho}.
$$
On the other hand Assumption A3) and the above imply that
\begin{equation}\label{er1a1}
\begin{aligned}
\E_\bz \left[V(\BZ(t+1)) h(\BZ(t), \xi(t))\right]
\leq& \rho \E_\bz V(\BZ(t))+C\\
\leq&\rho^{t} V(\bz)+\dfrac{C}{1-\rho}.
\end{aligned}
\end{equation}
Since by A3) part i) $V(\bz)\geq 1$ we also get that
$$
\begin{aligned}
\E_\bz h(\BZ(t), \xi(t))
\leq&\rho^{t+1} V(\bz)+\dfrac{C}{1-\rho}.
\end{aligned}
$$
This implies that if $\mu$ is an ergodic invariant probability
measure then
$$\E_\mu h(\BZ(t), \xi(t)) = \int \E h(\bz, \xi(1))
\mu(d\bz)\leq \frac{C}{1-\rho}.$$

Since $$|\log F_i(\bz,\xi)|^k\leq c_k \max\left\{F_i(\bz,\xi),\frac{1}{F_i(\bz,\xi)}\right\}^{\gamma_3}\leq c_k h(\bz,\xi)$$
for some $c_k>0$,  we see that for any $k\in\N$ there exists
$C_k>0$ such that
$$ \int \E |\log F_i(\bz,\xi(1))|^k\mu(d\bz)\leq C_k.
$$
The strong law of large numbers for martingales implies that for
$\mu$ almost every $\bz$ we have
 \begin{equation}\label{e5-A1a}
 \lim_{T\to\infty}
\frac1T \sum_0^T\left(\log F_i(\BZ(t+1))-P \log
F_i(\BZ(t))\right)=0, \,\text{when}~ \BZ(0)=\bz.
\end{equation}
Having \eqref{e5-A1a},
we can follow the arguments by \cite{BS19}[Lemma 3 and
Proposition 1] to obtain that
if $\mu(\Se^{I}_+)=1$ then
$r_i(\mu)=0$ for any $i\in I$.

Finally, because of the boundedness \eqref{er1a1}, it is standard to show that if a function $\psi$ satisfies $$\lim_{\bz\to\infty}
\frac{\psi(\bz)}{\E V(\bz,\xi(t))h(\bz,\xi(t))}=0$$
then $\psi$ is $\mu$-integrable for any invariant probability
measure $\mu$. See Lemma 3.3 in \cite{HN16} for a similar proof.
\end{proof}

\begin{lm}\label{lmA.2}
There exist $M, C_2, \gamma_4>0,\rho_2\in(0,1)$ such that
$$
\E_\bz\left[ V(\BZ(1))\prod_{i=1}^n X_i^{p_i}(1)\right]\leq
\left(\1_{\{|\bz|<M\}}(C_2-\rho_2)+\rho_2\right)
V(\bz)\prod_{i=1}^nx_i^{p_i}, ~\bz\in \Se
$$
for any $\bp=(p_1,\dots,p_n)\in\R^n$ satisfying
\begin{equation}\label{e:p}
|\bp|_1:=\sum|p_i|\leq\gamma_4.
\end{equation}
\end{lm}

\begin{proof}
%

We note that for $w_1,\dots,w_n>0$
$$
\begin{aligned}
\prod_{i=1}^n w_i^{p_i}\leq& \prod_{i=1}^n
\left(w_i^{-1}\vee{w_i}\right)^{|p_i|}\\
\leq& \left(\max_{i=1,\dots,n}\left\{
w_i^{-1}\vee{w_i}\right\}\right)^{\sum_i|p_i|}\\
\leq&
1+\frac{\sum_i|p_i|}{\gamma_3}\left(\max_{i=1,\dots,n}\left\{
w_i^{-1}\vee{w_i}\right\}\right)^{\gamma_3}
\end{aligned}
$$
if $\sum_i |p_i|<\gamma_3.$
The last inequality follows from the inequality $x^p\leq 1+px$ for $x\geq 0, p\in(0,1)$.
Thus,

$$
\prod_{i=1}^n F_i^{p_i}(\bz,\xi)
\leq 1+ \frac{|\bp|_1}{\gamma_3} h(\bz,\xi).
$$
Since $\lim_{\bz\to\infty} V(\bz)=\infty$, we can select
$\rho_1\in(\rho,1)$, $M>0$ and $C_1>0$ such that
$$\rho V(\bz)+C\leq (C_1\1_{\{|z|<M\}}+\rho_1)V(\bz).$$
Let $\gamma_4\in(0,\gamma_3)$ be such that
$\left(1+\frac{\gamma_4}{\gamma_3}\right)\rho_1=:\rho_2<1.$
There exists $C_2>0$ satisfying
$$(C_1\1_{\{|z|<M\}}+\rho_1)\left(1+\frac{\gamma_4}{\gamma_3}\right)
\leq (C_2-\rho_2)\1_{\{|z|<M\}}+\rho_2.$$
The above estimates together with \eqref{e:system_discrete},
Lemma \ref{lmA.1}, and \eqref{e:p} yield

$$
\begin{aligned}
\E_\bz \left[V\left(\BZ(1)\right)\prod_{i=1}^n X_i^{p_i}(1)\right]
=&\prod_{i=1}^nx_i^{p_i} \E \left[V(\bx \circ F(\bz,\xi),
G(\bz,\xi))\prod_{i=1}^n F_i^{p_i}(\bz,\xi)\right]\\
\leq&\prod_{i=1}^nx_i^{p_i} \E \left[V(\bx \circ F(\bz,\xi),
G(\bz,\xi))\left(1+\frac{|\bp|_1}{\gamma_3}h(\bz,\xi)\right)\right]\\
\leq & \prod_{i=1}^nx_i^{p_i}(\rho
V(\bz)+C)\left(1+\frac{|\bp|_1}{\gamma_3}\right)\\
\leq & \prod_{i=1}^nx_i^{p_i}(C_1\1_{\{|z|<M\}}+\rho_1)
V(\bz)\left(1+\frac{|\bp|_1}{\gamma_3}\right)\\
\leq &((C_2-\rho)\1_{\{|z|<M\}}+\rho_2)
\prod_{i=1}^nx_i^{p_i}V(\bz).
\end{aligned}
$$
\end{proof}
We will denote by $B_M :=\{\bu\in\Se:\|\bu\|\leq M\}$ the
closed ball of radius $M>0$ around the origin.
Let $\M$ be the set of ergodic invariant probability measures of
$\BX$ supported on the boundary
$\Se_0:=\partial\R^n_+\times\R^{\kappa_0}$. Remember that for a subset
$\wtd\M\subset \M$ we denote by $\Conv(\wtd\M)$ the convex hull of
$\wtd\M$,
that is the set of probability measures $\pi$ of the form
$\pi(\cdot)=\sum_{\mu\in\wtd\M}p_\mu\mu(\cdot)$
with $p_\mu\geq 0,\sum_{\mu\in\wtd\M}p_\mu=1$.

Consider $\mu\in\M$.
Assume $\mu(\{0\}\times\R^{\kappa_0})=0$. Since $\mu$ is ergodic there exist $0<n_1<\dots< n_k\leq n$
such that $\suppo(\mu)\subset
S^\mu:=\R^\mu_+\times\R^{\kappa_0}$ where
$$\R_+^\mu:=\{(x_1,\dots,x_n)\in\R^n_+: x_i=0\text{ if } i\in
I_\mu^c\}$$
for
$I_\mu:=\Se(\mu)=\{n_1,\dots, n_k\}$ and
$I_\mu^c:=\{1,\dots,n\}\setminus\{n_1,\dots, n_k\}$.

$$\R_+^{\mu,\circ}:=\{(x_1,\dots,x_n)\in\R^n_+: x_i=0\text{ if }
i\in I_\mu^c\text{ and }x_i>0\text{ if }x_i\in I_\mu\}$$ and
$\partial\R_+^{\mu}:=\R_+^\mu\setminus\R_+^{\mu,\circ}.$

The following condition ensures persistence.
\begin{asp}\label{a.coexn}
For any $\mu\in\Conv(\M)$ one has
$$\max_i r_i(\mu)>0,$$
where
$$ r_i(\mu):=\int_{\Se_0}\left[\E \ln
F_i(\bz,\xi)\right]\mu(d\bz).$$

\end{asp}

%
%
Remember that the occupation measures are defined as
\[
\Pi_{t,\bz}(\cdot)=\frac{1}{t}\sum_{s=0}^t\PP_\bz(\BZ(s)\in\cdot)\,ds,~\bz\in\Se,
t\in\Z_+
\]
\begin{lm}\label{lm2.4}
Suppose the following
\begin{itemize}
\item The sequences $(\bz_k)_{k\in N}\subset
\R_+^n\times\R^{\kappa_0}, (T_k)_{k\in \N}\subset \N$ are such
that $\|\bz_k\|\leq M$, $T_k>1$ for all $k\in \N$ and
$\lim_{k\to\infty}T_k=\infty$.

\item The sequence $(\Pi_{T_k,\bz_k})_{k\in \N}$
converges weakly to an invariant probability measure
$\pi$.

\item The function $h:\R^n_+\times\R^{\kappa_0}\to\R$ is any
continuous function satisfying
$$\lim_{\bz\to\infty}\frac{|h(\bz)|}{V(\bz)\max_{i=1}^n\{(x_i^{\gamma_3}\wedge
x^{-\gamma_3}_i)\}}
=0
$$
\end{itemize}
Then one has
\[\lim_{k\to\infty}\int_{\R^n_+}h(\bx)\Pi_{T_k,\bz_k}(d\bx)=
\int_{\R^n_+}h(\bx)\pi(d\bx).\]
\end{lm}
\begin{proof}
The proof is almost identical to that of Lemma 3.4 by
\cite{HN16} and is therefore omitted.
\end{proof}
It is shown in \cite[Lemma 4]{SBA11} by the min-max principle
that Assumption \ref{a.coexn} is equivalent to the existence of
$\mathbf p>0$ such that
\begin{equation}\label{e.p}
\min\limits_{\mu\in\M}\left\{\sum_{i}p_i
r_i(\mu)\right\}:=2r^*>0.
\end{equation}
By rescaling if necessary, we can assume that
$|\bp|_1=\gamma_4$.
\begin{lm}\label{lm3.1}
Suppose that Assumption \ref{a.coexn} holds. Let $\bp$ and $r^*$
be as in \eqref{e.p}.
There exists an integer $T^*>0$ such that, for any $T>T^*$,
$\bx\in\partial\R^n_+, \bz= (\bx,\by) \in B_M$ one has
\begin{equation}\label{lm3.1-e0}
\sum_{t=0}^T\E_\bz\left(\ln V(\BZ(t+1))-\ln V(\BZ(t))-\sum
p_i\ln F_i(\BZ(t), \xi(t))\right)\leq-r^*(T+1).
\end{equation}
\end{lm}
\begin{proof}
In view of Lemma \ref{lmA.1},
$$\sup_{t\in\N, \|\bz\|\leq M} \E_\bz V(\BZ(t))<\infty$$
which implies, since
\[
\frac{\E_\bz \ln V(\BZ(T))}{T+1}\leq \frac{\E_\bz
V(\BZ(T))}{T+1}
\]
that
$$\lim_{T\to\infty} \sup_{\|\bz\|\leq M}
\frac1{T+1}\sum_{t=0}^T\E_\bz\left(\ln V(\BZ(t+1))-\ln
V(\BZ(t))\right)=0.$$

With \eqref{e.p} and Lemma \ref{lmA.2} and Lemma \ref{lm2.4},
we can argue by contradiction in the same manner as in Lemma 4.1
by \cite{HN16} to show that
$$\limsup_{T\to\infty} \frac1{T+1}\sup_{\|\bz\|\leq M}
\sum_{t=0}^T\E_\bz\left(-\sum p_i\ln F_i(\BZ(t),
\xi(t))\right)\leq-2r^*.$$
Combining the two above limits finishes the proof.
\end{proof}
Let $n^*\in\N$ be such that
\begin{equation}\label{e:n*}
\rho_2^{1-n^*}>C_2.
\end{equation}

\begin{prop}\label{prop2.1}
Define $U:\Se_+\to\R_+$ by
$$U(\bz)=V(\bz)\prod_{i=1}^nx_{i}^{-p_i}$$ with $\bp$ and $r^*$
satisfying \eqref{e.p} and $T^*>0$ satisfying the assumptions of
Lemma \ref{lm3.1}.
There exist numbers $\theta\in\left(0,\frac{\gamma_4}2\right)$,
$K_\theta>0$, such that for any $T\in[T^*,n^*T^*]\cap \Z$ and
$\bz\in\Se_+, \|\bz\|\leq M$,
$$\E_\bz U^\theta(\BZ(T))\leq
U^\theta(\bx)\exp\left(-\frac{1}{2}\theta r^*T\right)
+K_\theta.$$
\end{prop}
\begin{proof}

\begin{equation}\label{e:G}
\begin{aligned}
\ln U(\BZ(T))=&\ln U(\BZ(0)) + \sum_{t=0}^{T-1} \left(\ln
U(\BZ(t+1))-\ln U(\BZ(t))\right)\\
=&\ln U(\BZ(0)) + G(T)
\end{aligned}
\end{equation}
where
\begin{equation}
\begin{aligned}
G(T)=\sum_{t=0}^{T-1} \left(\ln V(\BZ(t+1))-\ln V(\BZ(t))-\sum
p_i\ln F_i(\BZ(t), \xi(t))\right).
\end{aligned}
\end{equation}
In view of \eqref{e:G} and Lemma \ref{lmA.2}
\begin{equation}\label{e3.4_2}
\E_\bz \exp( G(T))=\dfrac{\E_\bz U(\BZ(T))}{U(\bz)}\leq(C_2)^T.
\end{equation}
Let $\hat U(\cdot):\R^{n,\circ}_+\times\R^{\kappa_0}\mapsto\R_+$
be defined by $\hat U(\bz)=V(\bz)\prod_{i=1}^n x_i^{p_i}$.
We also have
\begin{equation}\label{vhat-1}
\dfrac{\E_\bz \hat U(\BZ(T))}{\hat U(\bz)}\leq  (C_2)^T.
\end{equation}
Note that
\begin{equation}\label{vhat-2}
U^{-1}(\bz)=\hat U(\bz)\frac{1}{V^2(\bz)}\leq \hat U(\bz).
\end{equation}
Using \eqref{vhat-2} and \eqref{vhat-1} yields
\begin{equation}\label{e3.5}
\begin{aligned}
\E_\bz \exp(-G(T))=&\dfrac{\E_\bz U^{-1}(\BZ(T))}{U^{-1}(\bz)}\\\leq&\dfrac{\E_\bz\hat U(\BZ(T))}{V^2(\bz)U^{-1}(\bx)}\\
\leq& \dfrac{\E_\bz\hat U(\BZ(T))}{\hat U(\bz)}\\
\leq& (C_2)^T.
\end{aligned}
\end{equation}

By \eqref{e3.4_2} and \eqref{e3.5} the assumptions of
\cite[Lemma 3.5]{HN16} hold for the random variable $G(T)$.
Therefore,
there exists $\tilde K_2\geq 0$ such that
$$0\leq \dfrac{d^2\tilde\phi_{\bz,T}}{d\theta^2}(\theta)\leq
\tilde K_2\,\text{ for all
}\,\theta\in\left[0,\frac{1}2\right),\,
\bz\in\R^{n,\circ}_+\times \R^{\kappa_0}, \|\bz\|\leq M, T\in
[T^*,n^*T^*]\cap \Z$$
where
$$\tilde\phi_{\bz,T}(\theta)=\ln\E_\bz \exp(\theta G(T)).$$
In view of Lemma \ref{lm3.1} and the Feller property of
$(\BZ(t))$,
there exists a $\tilde\delta>0$ such that
if $\|\bz\|\leq M$, $\dist(\bz,\partial\R^n_+)<\tilde\delta$ and
$T\in [T^*,n^*T^*]\cap \Z$
then
\begin{equation}\label{e3.6}
\begin{aligned}
\E_\bz G(T)\leq -\dfrac34r^*T.
\end{aligned}
\end{equation}
Another application of \cite[Lemma 3.5]{HN16}  yields
$$\dfrac{d\tilde\phi_{\bz,T}}{d\theta}(0)=\E_\bz G(T)\leq
-\dfrac34 r^*T\,\text{ for }\, \bz\in\R^{n,\circ}_+\times
\R^{\kappa_0}, \|\bz\|\leq M,
\dist(\bz,\partial\R^n_+)<\tilde\delta, T\in [T^*,n^*T^*]\cap
\Z.$$
By a Taylor expansion around $\theta=0$, for $\|\bz\|\leq M,
\dist(\bz,\partial\R^n_+)<\tilde\delta, T\in [T^*,n^*T^*]\cap
\Z$ and $\theta\in\left[0,\frac{1}2\right)$ we have
$$\tilde\phi_{\bx,T}(\theta)\leq -\dfrac34
r^*T\theta+\theta^2\tilde K_2 .$$
If we choose any $\theta\in\left(0,\frac{1}2\right)$ satisfying
$\theta<\frac{r^*T^*}{4\tilde K_2}$, we obtain that
\begin{equation}\label{e3.10}
\tilde\phi_{\bz,T}(\theta)\leq -\dfrac12 r^*T\theta\,\,\text{
for all }\,\bz\in\R^{n,\circ}\times\R^{\kappa_0},\|\bz\|\leq M,
\dist(\bz,\partial\R^n_+)<\tilde\delta, T\in [T^*,n^*T^*]\cap
\Z.
\end{equation}
In light of \eqref{e3.10}, we have for all
$\theta<\frac{r^*T^*}{4\tilde K_2}$, $\|\bz\|\leq M,
0<\dist(\bz,\partial\R^n_+)<\tilde\delta, T\in [T^*,n^*T^*]$
that
\begin{equation}\label{e3.11}
\dfrac{\E_\bz U^\theta(\BZ(T))}{U^\theta(\bz)}=\exp
\tilde\phi_{\bz,T}(\theta)\leq\exp\left(-\frac{1}{2}r^*T\theta\right).
\end{equation}
In view of Lemma \ref{lmA.2},
we have for $\bz$ satisfying $\|\bz\|\leq M,
\dist(\bz,\partial\R^n_+)\geq\tilde\delta$ and $T\in
[T^*,n^*T^*]$ that
\begin{equation}\label{e3.12}
\E_\bz U^\theta(\BZ(T))\leq (C_2)^{\theta
n^*T^*}\sup\limits_{\|\bz\|\leq M,
\dist(\bz,\partial\R^n_+)\geq\tilde\delta}\{U^\theta(\bz)\}=:K_\theta<\infty.
\end{equation}
Combining \eqref{e3.11} and \eqref{e3.12} we are done.
\end{proof}

\begin{thm}\label{thm3.1}
Suppose that Assumption \ref{a.coexn} holds. Let $\theta$ be as
in Proposition \ref{prop2.1}, $T^*$ as in Lemma \ref{lm3.1} and
$n^*$ as in \eqref{e:n*}.
There exist numbers $\kappa=\kappa(\theta,T^*)\in(0,1)$ and
$\tilde K=\tilde K(\theta,T^*)>0$ such that
\begin{equation}\label{e:lya}
\E_\bz U^\theta(\BZ(n^*T^*))\leq \kappa U^\theta(\bz)+\tilde
K\,\text{ for all }\, \bz\in\R^{n,\circ}_+\times\R^{\kappa_0}.
\end{equation}
We have\begin{equation}\label{e:lya2}
\limsup_{t\to\infty}\PP_\bz\{|X_i(t)|\vee |X^{-1}_i(t)|>m \text{
for some } i=1,\dots, n\}\leq c_2 m^{-c_3}
\end{equation}
for some positive $c_2, c_3>0.$
Moreover, for any compact set $K\subset
\R^{n,\circ}\times\R^{\kappa_0}$,
\begin{equation}\label{e:lya1} \PP_\bz(\tau_K>k)\leq
c_KU^\theta(\bz)\kappa^k
\end{equation}
If the Markov chain $\BZ(t)$ is irreducible and aperiodic on
$\R^{n,\circ}_+\times\R^{\kappa_0}$,
and a compact set is petite, then there is $c_4>1$ such that
$$ c_4^t\|P_t(\bz,\dots)-\pi\|_{TV}\to 0\text{ as }
t\to\infty.$$
\end{thm}
\begin{proof}
Define
\begin{equation}\label{e:tau}
\tau=\inf\{t\geq0: \|\BZ(t)\|\leq M\}.
\end{equation}
By Lemma \ref{lmA.2} for all $\bz\in\Se$
$$PU(\bz)\leq \rho_2 U(\bz), ~\|\bz\|\geq M.$$
This implies that the process $\rho^{-t}_2 U(\BZ(t))$ is a
supermartingale and therefore
$$
\begin{aligned}
\E_\bz&\left[ \rho_2^{-\theta(\tau\wedge
n^*T^*)}U^\theta(\BZ(\tau\wedge n^*T^*))\right] \leq
U^\theta(\bz), \bz\in\Se.
\end{aligned}
$$
Thus,
\begin{equation}\label{et1.2}
\begin{aligned}
U^\theta(\bz)\geq&
\E_\bz\left[ \rho_2^{-\theta(\tau\wedge
n^*T^*)}U^\theta(\BZ(\tau\wedge n^*T^*))\right]\\
=&
\E_\bz \left[\1_{\{\tau\leq
(n^*-1)T^*\}}\rho_2^{-\theta(\tau\wedge
n^*T^*)}U^\theta(\BZ(\tau\wedge n^*T^*))\right]\\
&+\E_\bz \left[\1_{\{
(n^*-1)T^*<\tau<n^*T^*\}}\rho_2^{-\theta(\tau\wedge
n^*T^*)}U^\theta(\BZ(\tau\wedge n^*T^*))\right]\\
&+ \E_\bz \left[\1_{\{\tau\geq
n^*T^*\}}\rho_2^{-\theta(\tau\wedge
n^*T^*)}U^\theta(\BZ(\tau\wedge n^*T^*))\right]\\
\geq&
\E_\bz \left[\1_{\{\tau\leq
(n^*-1)T^*\}}U^\theta(\BZ(\tau))\right]\\
&+\rho_2^{-\theta(n^*-1)T^*}\E_\bz \left[\1_{\{
(n^*-1)T^*<\tau<n^*T^*\}}U^\theta(\BZ(\tau))\right]~~\\
&+\rho_2^{-\theta n^*T^*}\E_\bz \left[\1_{\{\tau\geq
n^*T^*\}}U^\theta(\BZ(n^*T^*))\right].\\
\end{aligned}
\end{equation}

By the strong Markov property of $\BZ(t)$ and
Proposition \ref{prop2.1}, we obtain
\begin{equation}\label{et1.3}
\begin{aligned}
\E_\bz&\left[ \1_{\{\tau\leq
(n^*-1)T^*\}}U^\theta(\BZ(n^*T^*))\right]\\
&\leq
\E_\bz \left[\1_{\{\tau\leq
(n^*-1)T^*\}}\big[K_\theta+e^{-\frac{1}{2}\theta
r^*(n^*T^*-\tau)}U^\theta(\BZ(\tau))\big]\right]\\
&\leq K_\theta+ \exp\left(-\frac{1}{2}\theta
r^*T^*\right)\E_\bz\left[\1_{\{\tau\leq
(n^*-1)T^*\}}U^\theta(\BZ(\tau))\right]
\end{aligned}
\end{equation}
Similarly, the strong Markov property of $\BZ(t)$, Jensen's
inequality and
Lemma \ref{lmA.2} imply
\begin{equation}\label{et1.4}
\begin{aligned}
\E_\bz&\left[
\1_{\{(n^*-1)T^*<\tau<n^*T^*\}}U^\theta(\BZ(n^*T^*))\right]\\
&\leq
\E_\bz \left[\1_{\{(n^*-1)T^*<\tau<n^*T^*\}}C_2^{\theta
(n^*T^*-\tau)}U^\theta(\BZ(\tau))\right]\\
&\leq C_2^{\theta
T^*}\E_\bz\left[\1_{\{(n^*-1)T^*<\tau<n^*T^*\}}U^\theta(\BZ(\tau))\right].
\end{aligned}
\end{equation}
Applying \eqref{et1.3} and \eqref{et1.4} to \eqref{et1.2} yields\begin{equation}\label{et1.5}
\begin{aligned}
U^\theta(x)
\geq&
\E_\bz \left[\1_{\{\tau\leq
(n^*-1)T^*\}}U^\theta(\BZ(\tau))\right]\\
&+\rho_2^{-\theta (n^*-1)T^*}\E_\bz \left[\1_{\{
(n^*-1)T^*<\tau<n^*T^*\}}U^\theta(\BZ(\tau))\right]\\
&+\rho_2^{-\theta n^*T^*} \E_\bz \left[\1_{\{\tau\geq
n^*T^*\}}U^\theta(\BZ(n^*T^*))\right]\\
\geq& \exp\left(\frac{1}{2}\theta
r^*T^*\right)\E_\bz\left[\1_{\{\tau\leq
(n^*-1)T^*\}}U^\theta(\BZ(n^*T^*))\right]-\exp\left(\frac{1}{2}\theta
r^*T^*\right)K_\theta\\
&+C_2^{-\theta T^*}\rho_2^{-\theta (n^*-1)T^*}\E_\bz
\left[\1_{\{
(n^*-1)T^*<\tau<n^*T^*\}}U^\theta(\BZ(n^*T^*))\right]\\
&+\rho_2^{-\theta n^*T^*} \E_\bz \left[\1_{\{\tau\geq
n^*T^*\}}U^\theta(\BZ(n^*T^*))\right]\\
\geq & \kappa^{-1}\E_\bz
U^\theta(\BZ(n^*T^*))-K_\theta\exp\left(\frac{1}{2}\theta
\rho^*T^*\right)
\end{aligned}
\end{equation}
where $\kappa=\max\left\{\exp\left(-\frac{1}{2}\theta
r^*T^*\right), C^{\theta T^*}\rho_2^{\theta(n^*-1)T^*},
\rho_2^{\theta n^*T^*}\right\}<1$ by \eqref{e:n*}.
The proof of \eqref{e:lya} is complete by taking
\[
\tilde K=K_\theta\exp\left(\frac{1}{2}\theta
\rho^*T^*\right)\kappa.
\]
Having \eqref{e:lya}, the claims \eqref{e:lya1} and
\eqref{e:lya2} follow by \cite{BS19}[Proposition 3.3 and Theorem
3.1].
\end{proof}

\section{Extinction Proofs}\label{s:b}

For $I\subset\{1,\dots,n\}$, denote by
$\M^I, \M^{I,+}, \M^{I,\partial}$ the set of ergodic
probability measures
on $\Se^I, \Se^{I}_+, \Se_0^I$ respectively.

\begin{lm}\label{lmB.1}
Assume that there exists a function $\phi:\Se\to\R_+$ and
constants $C, \delta_\phi>0$ such that for all $\bz\in\Se$
\begin{equation}\label{e:Vphi}
P V(\bz)\leq V(\bz)-\phi(\bz)+C
\end{equation}
and
\begin{equation}\label{e1-B3}
\E_{\bz}\left(V(\BZ(1))-PV(\bz)\right)^2+\E\left|\log
F(\bz,\xi(1))-\E\log F(\bz,\xi(1))\right|^2\leq \delta_\phi
\phi(\bz).
\end{equation}
Then, the family of random occupation measures $(\widetilde\Pi_t)_{t\in
\N}$ is tight.
and with probability one
  \begin{equation}\label{e5-B3}
\lim_{T\to\infty}\frac1T\sum_{t=0}^T\left(\log
F(\BZ(t),\xi(t))-\E\left[\log
F(\BZ(t),\xi(t))\big|\F_t\right]\right)=0.
\end{equation}
where $(\F_t)_{t\in \N}$ is the filtration generated by the
process $\BZ$.
\end{lm}
\begin{proof}
Suppose $\BZ(0)=\bz\in\Se$. We have
\begin{equation}\label{e2-B3}
\begin{aligned}
V(\BZ(t+1))\leq& V(\BZ(t))-\phi(\BZ(t))+C
+(V(\BZ(t+1))-PV(\BZ(t))\\
\leq & V(\BZ(t))-\frac12\phi(\BZ(t))+2C
-\left(\frac12\phi(\BZ(t))+C -(V(\BZ(t+1))-PV(\BZ(t)))\right)
 \end{aligned}
 \end{equation}
We see from \eqref{e1-B3} that the quadratic variation of the
martingale $V(\BZ(t+1))-PV(\BZ(t))$ is bounded by $\delta
\phi(\BZ(t))$. As a result we can use the strong law of large
numbers for martingales and the bound \eqref{e1-B3}, to get
$$
 \lim_{T\to \infty}
\frac{\frac1T\sum_{t=0}^T(V(\BZ(t+1))-PV(\BZ(t)))}{\frac1T\sum_{t=0}^T \left(\phi(\BZ(t))+C\right)}=\lim_{T\to \infty}
\frac{\sum_{t=0}^T(V(\BZ(t+1))-PV(\BZ(t)))}{\sum_{t=0}^T \left(\phi(\BZ(t))+C\right)}=0\,\text{ a.s.}
$$
which implies
 \begin{equation}\label{e3-B3}
 \limsup_{T\to \infty}
\frac1T\sum_{t=0}^T \left(\frac12\phi(\BZ(t))+C
-(V(\BZ(t+1))-PV(\BZ(t)))\right)\leq 0 \text{ a.s.}
 \end{equation}
 Taking sums in \eqref{e2-B3}, noting that
 \[
\liminf_{T\to \infty}\frac1T \sum_0^T (V(\BZ(t+1)) - V(\BZ(t)))=
\liminf_{T\to 0}\frac1T (V(\BZ(T+1)) - V(\bz)) \geq 0\text{
a.s.}
 \]
 (because $V$ is nonnegative)
  and using \eqref{e3-B3} yields that with probability one
 $$\liminf_{T\to \infty}
 \frac1T\sum_0^T \left(-\frac12\phi(\BZ(t))+2C\right)\geq 0.
 $$
 As a result
 \begin{equation}\label{e4-B3}
 \limsup_{T\to \infty}
 \frac1T\sum_0^T \phi(\BZ(t))\leq 4C,
 \end{equation}
almost surely.
Since $\lim_{|\bz|\to\infty}\phi(\bz)=\infty$, the boundedness of $\widetilde{\Pi}_T\phi$ in \eqref{e4-B3}  implies that the family of randomized
occupation measures $\{\widetilde{\Pi}_t, t\in\N\}$ is tight.
Moreover, the strong law of large numbers together with
\eqref{e4-B3} and \eqref{e1-B3} implies \eqref{e5-B3}.
\end{proof}

\begin{thm}\label{t:exxx}
If $E_1$ is nonempty, then for any $I\in E_1$, there exists
$\alpha_I>0$ such that, for any a compact set $\K^I\subset
\Se^{I}_+$,
we have
$$
\lim_{\dist(\bz,\K^I)\to0, \bz\in
S^\circ}\PP_\bz\left\{\lim_{t\to\infty}\dfrac{\ln X_i(t)}t\leq
-\alpha_I, i\in I^c\right\}=1.
$$
\end{thm}

\begin{thm}\label{t:exx2}
If $E_2$ is empty
or $\max_{i}\{ r_i(\nu)\}>0$ for any $\nu$ with
$\mu(\Se^{J}_+)=1$ for some $J\in E_2$ and $\cup_{I\in
E_1}\Se^{I}_+$ is accessible then
$$
\sum_{I\in E_1} p_{\bz,I}=1
$$
where
$$p_{\bz,I}=\PP_\bz\left\{\emptyset\neq
\U(\omega)\subset\Conv\{\M^{I,+}\}
~\text{and}~\lim_{t\to\infty}\frac{\ln
X_j(t)}{t}\in\{r_j(\mu):\mu\in\Conv(\M^{I,+})\}, j\in
I^c\right\}.$$
\end{thm}
\begin{proof}
Once Theorem \ref{t:exxx} is proved,
Theorem \ref{t:exx2} can be obtained using the fact that any
weak limit of a family of random occupation measures
is an invariant probability measure supported on $S_0$
(\cite{B18, HN16}) and by using the arguments from Lemma 5.8,
Lemma 5.9 and Theorem 5.2 by \cite{HN16}.
\end{proof}
Fix $I\in E_1$. Since by Lemma \ref{lmB.1} the family
$(\widetilde\Pi_t)_{t\in \N}$ of random occupation measures is tight,
condition \eqref{ae3.2} is equivalent to the existence of
$0<\hat p_i<\gamma_3/n, i\in I$
$$\inf_{\nu\in \Conv (\M^{I,\partial})}\sum_{i\in I}\hat p_i
r_i(\nu)>0.$$
As a result, there exists a small $\check p\in (0,\gamma_3/n)$
such that
\begin{equation}\label{e3.2}
\begin{aligned}
\sum_{i\in I}\hat p_i r_i(\nu)-\check p\max_{i\notin I}\{
r_i(\nu)\}>0 \text{ for any }\nu\in\Conv(\M^{I,\partial}).
\end{aligned}
\end{equation}
Define $\tilde p_i$ by
$\tilde p_i=\hat p_i$ if $i\in I_\mu$ and $\tilde p_i=-\check p$ if
$i\in I_\mu^c$.
In view of \eqref{e3.2}, \eqref{ae3.1} and Lemma \ref{lmB.1},
there is $ r_e>0$ such that for any $\nu\in\Conv(\M^I)$,
\begin{equation}\label{e3.3}
\sum_{i\in I}\hat p_i r_i(\nu)-\check p\max_{i\in I^c}\left\{
r_i(\nu)\right\}>3r_e.
\end{equation}
\begin{lm}\label{lm4.2}
Let $I\in E_1$ and suppose that Assumption \ref{a.extn} holds.
Suppose $\hat p_i, \check p, r_e$ are the quantities from
\eqref{e3.3} and $n^*$ is defined by \eqref{e:n*}.
There exist constants $T_e\geq 0$, $\delta_e>0$ such that, for
any $T\in[T_e,n^*T_e]\cap \Z$, $\|\bz\|\leq M, x_i<\delta_e,
i\in I^c$, we have
\begin{equation}\label{e3.4}
\begin{aligned}
\sum_{t=0}^T\E_\bz\left(\ln V(\BZ(t+1))-\ln V(\BZ(t))-\sum_{i\in
I} \ln F_i(\BZ(t), \xi(t))+\check p \max_{i\in I^c} \ln
F_i(\BZ(t), \xi(t))\right)dt\leq-r_e(T+1) .
\end{aligned}
\end{equation}
\end{lm}
\begin{proof}
This is very similar to the proof of Lemma \ref{lm3.1} and is
therefore omitted.
\end{proof}
\begin{prop}\label{prop4.1}
Let $I\in E_1$ and suppose that Assumption \ref{a.extn} holds.
There exists $\theta\in(0,1)$ such that for any
$T\in[T_e,n^*T_e]\cap \Z$ and $\bz\in\Se_+$ satisfying $
\|\bz\|\leq M,$ $x_i<\delta_e,$ $i\in I^c$ one has
$$\E_\bz W_\theta(\BZ(T))\leq \exp\left(-\frac{1}{2}\theta
r_eT\right) W_\theta(\bz)$$
where
$M, T_e, \hat p_i,\check p, \delta_e, n^*$ are as in Lemma
\ref{lm4.2} and
$$W_\theta(\bz):=\sum_{i\in I^c}\left[V(\bz)\dfrac{x_i^{\check
p}}{\prod_{j\in I} x_j^{\hat p_j}}\right]^\theta, \bz\in\Se_+.$$
\end{prop}
\begin{proof}
For $i\in I^c$, let
$W(\bz,i):=V(\bz)\dfrac{x_i^{\check p}}{\prod_{j\in I} x_j^{\hat
p_j}}.$
Similarly to Proposition \ref{prop2.1}, by making use of Lemma
\ref{lm4.2}, one can find a $\theta>0$ such that for
$T\in[T_e,n^*T_e]\cap \Z$,
$\bz\in\R^{n,\circ}_+\times\R^{\kappa_0}$ with $\|\bz\|\leq M,$
and $x_i<\delta_e$ we have
$$\E W^\theta(\BZ(T),i)\leq \exp\left(-\frac{1}{2}\theta
r_eT\right) W^\theta(\bz, i).$$
The proof is complete by noting that
$$W_\theta(\bz)=\sum_{i\in I^c}W^\theta(\bz,i).$$
\end{proof}

\begin{proof}[Proof of Theorem \ref{t:exxx}]
Let $I\in E_1$.
By Lemma \ref{lmA.2} for any $i\in I^c$
$$P W(\bz,i)\leq \rho_2 W(\bz,i), |\bz|\geq M.$$
Then using Jensen's inequality,
 we have
\begin{equation}\label{et3.1}
P W_\theta(\bz)\leq \rho_2^\theta W_\theta(\bz) \text{ if }
|\bz|\geq M.
\end{equation}
Define the constants
$$C_U:=\sup\left\{\dfrac{\prod_{i\in I} x_i^{\hat p_i}}{V(\bz)}:
\bz\in\R^{n,\circ}_+\times\R^{\kappa_0}\right\}<\infty,$$
$$
\varsigma:=\dfrac{\delta_e^{\check p\theta}}{C_U^\theta}$$
and the stopping time
$$\eta:=\inf\left\{t\geq0: W_\theta(\BZ(t))\geq
\varsigma\right\}.$$

Clearly, if $W_\theta(\bz)<\varsigma$, then $\eta>0$ and
for any $i\in I^c$, we get
\begin{equation}\label{e:ine}
X_i(t)\leq \delta_e\,, t\in [0,\eta).
\end{equation}
Let $$\wtd W_\theta(\bz):=\varsigma\wedge W_\theta(\bz).$$
We have from the concavity of $x\mapsto x\wedge \varsigma$ that
$$\E_\bz \wtd W_\theta(\BZ(T))\leq\varsigma\wedge \E_\bz
W_\theta(\BZ(T)).$$
Let $\tau$ be defined as in \eqref{e:tau}. By \eqref{et3.1} we
have that
$$
\begin{aligned}
\E_\bz&\left[ \rho_2^{-\theta(\tau\wedge\eta\wedge
n^*T_e)}W_\theta(\BZ(\theta\gamma_b(\tau\wedge\xi\wedge
n^*T_e))\right]
&\leq W_\theta(\bz), \bz\in\Se_+.
\end{aligned}
$$
As a result for all $\bz\in\Se_+$
\begin{equation}\label{et3.3}
\begin{aligned}
W_\theta(\bz)\geq&
\E_\bz\left[ \rho_2^{-\theta(\tau\wedge\eta\wedge
n^*T_e)}W_\theta(\BZ(\tau\wedge\eta\wedge n^*T_e))\right]\\
\geq& \E_\bz
\left[\1_{\{\tau\wedge\eta\wedge(n^*-1)T_e=\tau\}}W_\theta(\BZ(\tau))\right]\\
&+ \E_\bz
\left[\1_{\{\tau\wedge\eta\wedge(n^*-1)T_e=\eta\}}W_\theta(\BZ(\eta))\right]\\
&+\rho_2^{-\theta (n^*-1)T_e} \E_\bz
\left[\1_{\{(n^*-1)T_e<\tau\wedge\eta\leq
n^*T_e\}}W_\theta(\BZ(\tau\wedge\eta))\right]\\
&+\rho_2^{-\theta n^*T_e} \E_\bz \left[\1_{\{\tau\wedge\eta>
n^*T_e\}}W_\theta(\BZ(n^*T_e))\right].\\
 \end{aligned}
\end{equation}

By the strong Markov property of $(\BZ(t))$ and
Proposition \ref{prop4.1} (which we can use because of
\eqref{e:ine}) we see that for all $\bz\in\Se_+$
\begin{equation}\label{et3.4}
\begin{aligned}
\E_\bz&\left[
\1_{\{\tau\wedge\eta\wedge(n^*-1)T_e=\tau\}}W_\theta(\BZ(n^*T_e))\right]\\
&\leq
\E_\bz
\left[\1_{\{\tau\wedge\eta\wedge(n^*-1)T_e=\tau\}}\exp\left(-\frac{1}{2}\theta
r_e(n^*T_e-\tau)\right)W_\theta(\BZ(\tau))\right] \\
 &\leq \exp\left(-\frac{1}{2}\theta  r_eT_e)\right)
\E_\bz\left[\1_{\{\tau\wedge\eta\wedge(n^*-1)T_e=\tau\}}W_\theta(\BZ(\tau))\right].
\
 \end{aligned}
\end{equation}
Similarly, by the strong Markov property of $(\BZ(t))$ and
Lemma \ref{lmA.2}, we obtain for any $\bz\in\Se_+$ that
\begin{equation}\label{et3.5}
\begin{aligned}
\E_\bz&\left[ \1_{\{(n^*-1)T_e<\tau\wedge\eta\leq
n^*T_e\}}W_\theta(\BZ(n^*T_e))\right]\\
&\leq
\E_\bz \left[\1_{\{(n^*-1)T_e<\tau\wedge\eta\leq
n^*T_e\}}C_2^{\theta
(n^*T_e-\tau\wedge\eta)}W_\theta(\BZ(\tau\wedge\eta ))\right] \\
&\leq C_2^{\theta
T_e}\E_\bz\left[\1_{\{(n^*-1)T_e<\tau\wedge\eta\leq
n^*T_e\}}W_\theta(\BZ( \tau\wedge\eta ))\right].
 \end{aligned}
\end{equation}
If $W_\theta(\bz)<\varsigma$ then applying \eqref{et3.4},
\eqref{et3.5} and the inequality $\wtd W_\theta(\BZ(n^*T_e))\leq
W_\theta(\BZ(n^*T_e\wedge\eta))$ to \eqref{et3.3} yields
\begin{equation}\label{et3.8}
\begin{aligned}
\wtd W_\theta(\bz)=W_\theta(\bz)
\geq& \E_\bz
\left[\1_{\{\tau\wedge\eta\wedge(n^*-1)T_e=\tau\}}W_\theta(\BZ(\tau))\right]\\
&+ \E_\bz
\left[\1_{\{\tau\wedge\eta\wedge(n^*-1)T_e=\eta\}}W_\theta(\BZ(\eta))\right]\\
&+\rho_2^{-\theta (n^*-1)T_e} \E_\bz
\left[\1_{\{(n^*-1)T_e<\tau\wedge\eta\leq
n^*T\}}W_\theta(\BZ(\tau\wedge\eta))\right]\\
&+\rho_2^{-\theta n^*T_e} \E_\bz \left[\1_{\{\tau\wedge\eta>
n^*T_e\}}W_\theta(\BZ(n^*T_e))\right]\\
\geq&\exp\left(\frac{1}{2}\theta r_eT_e)\right)\E_\bz
\left[\1_{\{\tau\wedge\eta\wedge(n^*-1)T_e=\tau\}}W_\theta(\BZ(n^*T_e))\right]\\
&+ \E_\bz \left[\1_{\{\tau\wedge\eta\wedge(n^*-1)T_e=\eta\}}\wtd
W_\theta(\BZ(n^*T_e))\right]\\
&+\rho_2^{-\theta (n^*-1)T_e}C_2^{-\theta T_e}\E_\bz
\left[\1_{\{(n^*-1)T_e<\tau\wedge\eta\leq
n^*T_e\}}W_\theta(\BZ(n^*T_e))\right]\\
&+\rho_2^{-\theta n^*T_e}\E_\bz \left[\1_{\{\tau\wedge\eta>
n^*T_e\}}W_\theta(\BZ(n^*T_e))\right]\\
\geq& \E_\bz \wtd W_\theta(\BZ(n^*T_e)) \,\quad\text{ (since }
\wtd W_\theta(\cdot)\leq W_\theta(\cdot)).
 \end{aligned}
\end{equation}
Clearly, if $W_\theta(\bz)\geq\varsigma$ then
\begin{equation}\label{et3.8a}
\E_\bz \wtd W_\theta(\BZ(n^*T_e)) \leq \varsigma=\wtd
W_\theta(\bz).
\end{equation}
As a result of \eqref{et3.8}, \eqref{et3.8a} and the Markov
property of $(\BZ(t))$, the sequence
$\{w_2(k): k\geq 0\}$ where $w_2(k):=\wtd
W_\theta(\BZ(kn^*T_e))$ is a supermartingale.
Define the discrete stopping time $$\eta^*:=\inf\{k\in\N: W_\theta(\BZ(kn^*T_e))\geq\varsigma\}.$$
Moreover, we also deduce from \eqref{et3.8} that
$$
\E_\bz \1_{\{\eta^*>1\}}W_\theta(\BZ(n^*T_e))\leq \kappa_e
W_\theta(\bz), \bz\in\Se_+
$$
where $\kappa_e^{-1}=\min\left\{\rho_2^{-\theta
(n^*-1)T_e}C_2^{-\theta T_e},\exp\left(\frac{1}{2}\theta
r_eT_e)\right),\rho_2^{-\theta n^*T_e}\right\}>1.$
As a result, $\{w_3(k): k\geq 0\}$ with
$$w_3(k):=\kappa_e^{-k}\1_{\{\eta^*>k\}}W_\theta(\BZ(kn^*T_e))$$is also a supermartingle.
For any $\eps\in(0,1)$, if $W_\theta(\bz)\leq \varsigma\eps$ we
have
\begin{equation}\label{e:EY_ineq}
\E_\bz w_2(k\wedge\eta^*)\leq\E_\bz w_2(0)=W_\theta(\bz)\leq
\varsigma\eps\,, k\geq 0.
\end{equation}
Subsequently, \eqref{e:EY_ineq} combined with the Markov
inequality and the fact that $w_2(\eta)=W_\theta(\BZ(\eta
n^*T_e))\geq\varsigma$ yields
$$\PP_\bz\{\eta^*<k\}\leq \varsigma^{-1}\E_\bz
w_2(k\wedge\eta^*)\leq \eps, ~\text{if}~k\in \N,
W_\theta(\bz)\leq\varsigma\eps.
$$
Next, let $k\to\infty$ to get
\begin{equation}\label{et3.7}
\PP_\bz\{\eta^*<\infty\}\leq \eps\,\text{ if
}\,W_\theta(\bz)\leq\varsigma\eps.
\end{equation}
Let $\kappa_3\in(\kappa_e,1)$.
Since $P_t \wtd W_\theta(\bz)\leq C_2^{\theta t} \wtd
W_\theta(\bz)$ for any $t>0, \bz\in\Se_+$,
we have that
$$\PP_\bz\left\{ \wtd W_\theta(\BZ(t))\leq c, t=1,\dots,
n^*T_e\right\}\leq \dfrac{n^*T_e C_2^{\theta n^*T_e}\wtd
W_\theta(\bz)}{c}.$$
The last inequality, the Markov property of $\BZ(t)$, and the
fact that $(w_3(k))_{k\geq 0}$ is a supermartingale imply that
for any $\bz\in\Se_+$
$$
\begin{aligned}\PP_\bz&\left\{\1_{\{\eta^*>k\}}\wtd
W_\theta(\BZ(t))\leq
\left(\frac{\kappa_e}{\kappa_3}\right)^{-k}, t\in [kn^*T_e+1,
(k+1)n^*T_e]\right\}\\
&\leq n^*T_e C_2^{\theta n^*T_e}\E_\bz
w_3(k)\dfrac{\kappa_e^{-k}}{\kappa_3^{-k}}\leq n^*T_e
C_2^{\theta n^*T_e}\dfrac{\E_\bz w_3(0)}{\kappa_3^{-k}}\leq
n^*T_e C_2^{\theta n^*T_e}\dfrac{W_\theta(\bz)}{\kappa_3^{-k}}
\end{aligned}
$$
Since $\sum_k\left(n^*T_e C_2^{\theta
n^*T_e}\dfrac{W_\theta(\bz)}{\kappa_3^{-k}}\right)<\infty$,
we deduce from the Borel-Cantelli lemma that
$$\PP_\bz\left\{\limsup_{k\to \infty} \sup_{t\in [kn^*T_e+1,
(k+1)n^*T_e]}\1_{\{\eta^*>k\}}\wtd
W_\theta(\BZ(t))\left(\frac{\kappa_3}{\kappa_e}\right)^{k}<1\right\}=1.$$
This and \eqref{et3.7} imply that if $W_\theta(\bz)\leq
\varsigma\eps$ then
$$\PP_\bz\left\{\limsup_{t\to \infty}
\left(\left(\frac{\kappa_3}{\kappa_e}\right)^{\frac{t}{n^*T_e+1}}
W_\theta(\BZ(t))\right)<1\right
\}\geq1-\eps.$$
Since
$\sup_{\bz\in\R^n\times\R^{\kappa_0}}\frac{x_i^{\theta\check
p}}{W_\theta(\bz)}<\infty, i\in I^c$,
we can easily obtain the extinction result from Theorem
\ref{t:ex}.

\end{proof}
\section{Robustness proofs}\label{s:c}
Let $\circ$ denote the element-wise product and $\1_n$ be the
vector in $\R^n$ whose components are all $1$.
Assume that the function $V$ from Assumption $A3)$ satisfies the
robust estimate
\begin{equation*}
\E \left[V(\bx^\top (F(\bz,\xi)\circ(\1_n+\wtd\eps_1)),
G(\bz,\xi)+\wtd\eps_2))h(\bz,\xi)\right]\leq \rho V(\bz)+C
\end{equation*}
for any vectors $\wtd\eps_1,\wtd\eps_2\in\R^n$ such that
$|\wtd\eps_1|\vee|\wtd\eps_2|<\delta$.
Note that $$\tilde h(\bz,\xi)=\max_{i=1}^n \left\{\wtd
F_i(\bz,\xi),\frac{1}{\wtd F_i(\bz,\xi)}\right\}^{\gamma_3}.$$
Consequently $\tilde h(\bz,\xi)\leq e^\delta h(\bz,\xi)$ and
$$\E \left[V(\bx^\top (\wtd F(\bz,\xi)), \wtd G(\bz,\xi))\wtd
h(\bz,\xi)\right]\leq \rho V(\bz)+C$$
if $\delta>0$ is sufficiently small.
This shows that there exist
$C_2>0,\rho_2\in(0,1),\gamma_4\in(0,\gamma_3)$ such that
$$
\E_\bz V\left(\wtd \BZ(1)\prod_{i=1}^n\wtd
X_i^{p_i}(1)\right)\leq
\left(\1_{\{|\bz|<M\}}(C_2-\rho_2)+\rho_2\right)
V^{\frac12}(\bz)\prod_{i=1}^nx_i^{p_i}
$$
for any $\bp=(p_1,\dots,p_n)\in\R^n$ satisfying
$$|\bp|_1:=\sum|p_i|\leq\gamma_4.$$

Analogously to Lemma \ref{lmA.1} one can show the following
uniform bound
\begin{equation}\label{t-u-int}
\sup_{|\bz|\leq M, t\in\N} \left(\E_\bz
V(\wtd\BZ(t))+\E_{\bz}\tilde
h(\wtd\BZ(t),\xi(t))\right)=K_{M}<\infty
\end{equation}
when $\delta$ is sufficiently small.

Slight changing in the factor on the right hand side of
\ref{lm3.1-e0}, we have 
\begin{lm}\label{lmC.1}
Suppose that Assumption \ref{a.coexn} holds. Let $\bp$ and $r^*$
be as in \eqref{e.p}.
There exists a $\wtd T^*>0$ such that, for any $T>\wtd T^*$,
$\bz\in\Se_0, \|\bz\|\leq M$ one has
\begin{equation}\label{lmC1-e1}
\sum_{t=0}^T\E_\bz\left(\ln V(\BZ(t+1))-\ln V(\BZ(t))-\sum
p_i\ln F_i(\BZ(t), \xi(t))\right)\leq-1.5r^*(T+1).
\end{equation}
\end{lm}
On the other hand, it obviously follows from \eqref{e:robust1}
that for any $\eps>0, M>0, T>0, n_0>0$, we have
$$\E_\bz\left\{\1_{\{|\wtd\BZ(t))|\vee|\BZ(t)|<n_0\}}|\wtd\BZ(t))-\BZ(t)|\right\}\leq\eps,\,
\bz\leq M, t\in\{0,\dots, T\}$$
when $\delta$ is sufficiently small.
This and the uniform integrability \eqref{t-u-int} imply that,
for any $\eps, T>0$, there exists $\delta>0$ such that for any
$\delta$-pertubation of \eqref{e:system}, we have
\begin{equation}\label{lmC1-e2}
\begin{aligned}
\E_{\bz}|\ln \wtd F_i(\wtd\BZ(t))-\ln F_i(\BZ(t))|
\leq&\E_{\bz}|\ln \wtd F_i(\wtd\BZ(t))-\ln\wtd  F_i(\BZ(t))|\\
&+\E_{\bz}|\ln\wtd  F_i(\BZ(t))-\ln F_i(\BZ(t))|
\eps,\, \bz\in\Se_0,|\bz|\leq M, t\in\{0,\dots, T\}.
\end{aligned}
\end{equation}

By virtue of \eqref{lmC1-e1} and \eqref{lmC1-e2},
for sufficiently small $\delta$,
\begin{equation}
\sum_{t=0}^T\E_\bz\left(\ln V(\wtd\BZ(t+1))-\ln
V(\wtd\BZ(t))-\sum p_i\ln F_i(\wtd\BZ(t)),
\xi(t))\right)\leq-r^*(T+1)
\end{equation}
for any $\bz\in\Se_0, |\bz|\leq M$ and $T\geq\wtd T^*.$
With this estimate,
it follows from the arguments in Appendix \ref{s:a} that
Theorem \ref{thm3.1} holds for $\wtd\BZ(t)$.
Similarly, Theorem \ref{t:ex} and Theorem \ref{t:ex22} hold.

\section{Proof of Lemma \ref{l_per}}\label{s:d}
\begin{lem}
Suppose $\{(E_i(t), S(t))_{i=1,\dots, n}\}, t\in\N$ is a sequence of $n+1$-dimensional random variables, i.i.d. over $t$ such that $\E  \left[S(t)e^{E_j(t)}\right]^2<\infty.$ Then the model given by \eqref{e:per} and \eqref{e:Cj} satisfies Assumption \ref{a:main} by taking a small enough $\gamma_3>0$,
and
\[
V(\bz) = \sum_j z_j + 1.
\]
Assumption
\ref{a:ext} holds with
\[
\phi(\bz) = \delta V(\bz)
\]
for some $\delta>0$.
Moreover, if the support of $\ln S(t) + max_j (E_j(t)\ln \delta_j)$ contains values less than 0 then the boundary is accessible.
\end{lem}
\begin{proof}
We denote by $K$ below a positive generic constant.
Using \eqref{e:Cj}
\begin{equation}\label{e:ineq1}
\sum_j N_j(t)e^{E_j(t)-D_j(t)}\leq \sum_j N_j(t) e^{E_j(t)}\exp\{-\alpha_{j} N_j(t) e^{E_j(t)}\}<  K<\infty.
\end{equation}
On the other hand
$$
\left(1-\delta_j+ S(t)e^{E_j(t)-D_j(t)}\right)\vee \left(1-\delta_j+ S(t)e^{E_j(t)-D_j(t)}\right)^{-1}
\leq K+\sum_j S(t)e^{E_j(t)-D_j(t))}
$$
Using this and the inequality $$(a+b)^\gamma\leq 2^\gamma (a^\gamma+b^\gamma)$$ we see that for any $\eps>0$ there is $\gamma_3>0$ sufficiently small such that
$$
\begin{aligned}
h(t):=&\max_j\left(\left(1-\delta_j)+ S(t)e^{E_j(t)-D_j(t)}\right)^{\gamma_3}\vee \left(1-\delta_j+S(t) e^{E_j(t)-D_j(t)}\right)^{-\gamma_3}\right)\\
\leq& \left(1+\eps+\left(S(t)\sum e^{E_j(t)-D_j(t)}\right)^{\gamma_3}\right)
\end{aligned}
$$
Using \eqref{e:ineq1} and the fact that $\E S(t)<\infty$ we get
$$\E^t\sum_j \left((1-\delta_j)N_j(t)(1+\eps+\left(S(t)e^{E_j(t)-D_j(t)}\right)^{\gamma_3})\right)
\leq (1-0.5\check d) \sum_j N_j(t) +K
$$
where $\E^t[\cdot]:=\E[\cdot~|~N_1(t),\dots, N_n(t)]$ and $\check d=\min\{\delta_i\}$.
Since $N_j(t) e^{E_j(t)}\exp\{-a_{j}(1+\gamma_3) N_j(t) e^{E_j(t)}$ is bounded above by a nonrandom constant, then
$$
N_j(t)e^{(1+\gamma_3)(E_j(t)-D_j(t))}\leq e^{\gamma_3 E_j(t)} N_j(t) e^{E_j(t)}\exp\{-a_{j}(1+\gamma_3) N_j(t) e^{E_j(t)}\}<  Ke^{\gamma_3 E_j(t)}
$$
Therefore, since $\E \left[S^2(t)e^{\gamma_3 E_j(t)}\right]<\infty$ we have
$$
\E^t S^{1+\gamma_3}(t)N_j(t)e^{(1+\gamma_3)(E_j(t)-D_j(t))}
\leq K\E^t \left(S^{1+\gamma_3}(t)e^{\gamma_3 E_j(t)}\right)\leq K_2.
$$

As a result,
$$\E^t \sum_j \left(1-\delta_j+S(t)e^{E_j(t)-D_j(t)}\right)N_j(t)h(t)
\leq \sum_j (1-0.5\delta_j) N_j(t) + K_2
\leq (1-0.25\check d) \sum_j N_j(t)+K_2.$$
This implies that Assumption A3) is satisfied with $V(\bz)=\sum_j N_j+1=\sum_j z_j+1$ and small $\gamma_3>0$.
Note that
$$\sum_j N_j(t+1)-P_1 \sum_j N_j(t)=S(t)\sum_j e^{E_j(t)-D_j(t)} N_j(t)- \E^t S(t) \sum_j e^{E_j(t)-D_j(t)} N_j(t)$$
This shows that
$$
\begin{aligned}
\E^t \left(\sum_j N_j(t+1)-P_1 \sum_j N_j(t)\right)^2\leq& \E^t \left(S(t)\sum_j e^{E_j(t)-D_j(t)} N_j(t)\right)^2\\
\leq& K_3^2\E^t  \left[S(t)e^{E_j(t)}\right]^2 = K_3^2 \E^t \left[S(t)e^{E_j(t)}\right]^2
\end{aligned}
$$
since $\sum_j e^{-D_j(t)} N_j(t)$ is bounded by a constant $K_3$.
If $\E \left[S(t)e^{E_j(t)}\right]^2<\infty$ then Assumption
\ref{a:ext} is satisfied with $\phi(\bz)=\delta(|\bz|+1)$
for some $\delta>0$.

Moreover, if support of $\ln S(t) + \max_j \{E_j(t)\ln \delta_j\}$ contains values less than $0$ then it is clear that the boundary is accessible since when
$\ln S(t) + \max_j \{E_j(t) \ln \delta_j\}$ is less than a negative constant, $N_j(t+1)\leq \rho N_j(t)$ for $\rho<1$.
\end{proof}

\end{document}